\newif\iffac
\facfalse

\iffac
\documentclass{fac} 
\else 
\documentclass{article}[11pt] 
\fi

\newtheorem{theorem}{Theorem}[section]

\iffac
\title[Dynamic Intransitive Noninterference Revisited]
{Dynamic Intransitive Noninterference Revisited} 
\else 
\title{Dynamic Intransitive Noninterference Revisited
} 
\fi 

\iffac
\author[S. Eggert and R. van der Meyden]
    {
    Sebastian Eggert$^1$  and Ron van der Meyden$^2$\\ 
    $^1$ Institut f\"ur Informatik, Kiel University, Kiel, Germany. e-mail: research@sebastian-eggert.de\\ 
    $^2$ School of Computer Science and Engineering, UNSW Australia, Sydney NSW 2052, Australia. e-mail: meyden@cse.unsw.edu.au\\  
  }
  \else 
 \author{
    Sebastian Eggert\\
    Institut f\"ur Informatik, \\ Kiel University, \\
    Kiel, Germany. \\
    e-mail: research@sebastian-eggert.de\\ 
    \and 
    Ron van der Meyden\\ 
    School of Computer Science and Engineering, \\ UNSW Australia, \\
    Sydney NSW 2052, Australia. \\
    e-mail: meyden@cse.unsw.edu.au\\  
  } 
\fi

\iffac
\correspond{Ron van der Meyden, 
School of Computer Science and Engineering, UNSW Australia, Sydney NSW 2052 Australia. e-mail: meyden@cse.unsw.edu.au}

\pubyear{201*}
\pagerange{\pageref{firstpage}--\pageref{lastpage}}
\fi 

\newcommand{\commentout}[1]{}
\usepackage[utf8x]{inputenc}
\usepackage{amssymb}
\usepackage[cmex10]{amsmath}
\usepackage{paralist}
\usepackage[noend]{algorithm2e}
\usepackage{marginnote}
\usepackage{subfigure}
\usepackage{float}
\usepackage{xspace}
\usepackage{tikz}
\usepackage{version}
\usepackage[hidelinks]{hyperref}

\newcommand{\Prop}{\mathit{Prop}}  
\newcommand{\M}{{\cal K}}

\newtheorem{definition}{Definition}
\newtheorem{proposition}{Proposition}
\newtheorem{lemma}{Lemma}
\newtheorem{corollary}{Corollary}

\newcounter{example} 
\newenvironment{example}{\refstepcounter{example}~\\[5pt]\noindent {\bf Example \theexample:} }{\qedhere\\[5pt]}

\newcommand{\dsrc}[3]{{\tt dsrc}(#1, #2, #3)}
\newcommand{\lpurgename}{{\tt Lpurge}}
\newcommand{\lesliepurge}[3]{\lpurgename(#1, #2, #3)}
\newcommand{\dipurgename}{{\tt ipurge}}
\newcommand{\dipurge}[3]{\dipurgename(#1, #2, #3)}

\newcommand{\interferes}{\rightarrowtail}
\newcommand{\noninterferes}{\not\rightarrowtail}
\newcommand{\Actions}{A}
\newcommand{\Dom}{D}
\newcommand{\States}{S}

\newcommand{\dom}{{\tt dom}} 
\newcommand{\obs}{{\tt obs}}
\newcommand{\view}{{\tt view}}

\newcommand{\purgeformat}[1]{{\tt #1}}

\newcommand{\taname}{\purgeformat{ta}}
\newcommand{\ta}[2]{\taname_{#2}(#1)}

\newcommand{\maytaname}{\purgeformat{ta}^\Diamond}
\newcommand{\mayta}[2]{\maytaname_{#2}(#1)}

\newcommand{\musttaname}{\purgeformat{ta}^\Box}
\newcommand{\mustta}[2]{\musttaname_{#2}(#1)}

\newcommand{\tasecurity}{TA-security\xspace}

\newcommand{\unwind}{\sim}

\newcommand{\Objects}{O}
\newcommand{\Values}{V}

\newcommand{\contentsnme}{{\tt contents}} 
\newcommand{\observename}{{\tt observe}} 
\newcommand{\altername}{{\tt alter}} 

\newcommand{\contents}[2]{\contentsnme(#1, #2)} % 1=object, 2=state
\newcommand{\observe}[2]{\observename(#1, #2)} % 1=agent, 2=state
\newcommand{\alter}[2]{\altername(#1, #2)} % 1=agent, 2=state
\newcommand{\oset}[1]{{\tt oset}(#1)}
\newcommand{\Osets}{{\tt Osets}}

\newcommand{\powerset}[1]{\mathcal{P}(#1)}

\newcommand{\dynacrel}[1]{\approx_{#1}}

\newcommand{\DRM}[1]{{\tt DRM}#1}

\newcommand{\tags}{\mathcal{T}}

\newcommand{\processes}{P}
\newcommand{\wrt}{with respect to\xspace}%{\mbox{w.\,r.\,t.}\xspace}

\newcommand{\tuple}[1]{\langle #1 \rangle}

\newcommand{\btags}{{\cal N}}
\newcommand{\ptag}[2]{#1_{#2}}
\newcommand{\Data}{\mathit{Data}} 
\newcommand{\Obj}{\mathit{Objects}} 
\newcommand{\PolicyState}{\mathit{PolicyState}} 
\newcommand{\Caps}{\mathit{Capabilities}} 
\newcommand{\restrict}{\upharpoonright}

\newcommand{\A}{\mathcal{A}}

\newcommand{\qed}{$\Box$}
\newcommand{\qedhere}{\hfill\qed}
\usetikzlibrary{arrows,automata,fit,patterns}
\usetikzlibrary{shapes.geometric,shapes.arrows,decorations.pathmorphing}
\usetikzlibrary{matrix,chains,scopes,positioning,arrows,fit}
\usetikzlibrary{shapes,shadows,calc}
\usepgflibrary{arrows}
\tikzset{tikzglobal/.style={
    ->,
    >=stealth',
    shorten >=1pt,
    auto,
    node distance=0.8cm, 
    every path/.style=semithick, 
   initial text={}
 }}

  \tikzset{systemstate/.style={%
    matrix of nodes,
    draw=black, 
    line width=1pt,
    rectangle,
   inner sep=2pt,
    rounded corners
  }}
  
  \tikzset{policy/.style={ %depreciated
   >->,
   >=to
   }
}

 \tikzset{agent/.style={
  draw=none
}}

\tikzset{implies/.style={
	double,
	double equal sign distance,
	-implies
	}
}

 \newcommand{\unw}{\sim^{\mathit{unw}}}

 \newcommand{\unwprime}{\approx^{\mathit{unw}}}

 \newcommand{\rimp}{\Rightarrow}

 \newcommand{\From}{From }

\begin{document}
\label{firstpage}

\iffac
\makecorrespond
\fi 

\iffac \else 
\newenvironment{proof}{\noindent {\bf Proof:}}{\qedhere} 
\newenvironment{proof*}{\noindent {\bf Proof:}}{} 
\fi 

\maketitle

\begin{abstract} 
The paper studies dynamic information flow security policies 
in an automaton-based model. Two semantic
interpretations of such policies are developed, both of which 
generalize the notion of TA-security [van der Meyden ESORICS 2007] for static intransitive noninterference 
policies. One of the interpretations focuses on information flows permitted by 
policy edges, the other focuses on prohibitions implied by absence of policy edges. 
In general, the two interpretations differ, 
but necessary and sufficient conditions are identified 
for the two interpretations to be equivalent. 
Sound and complete proof techniques are developed for both interpretations. 
Two applications of the theory are presented. The first is a general result 
showing that access control  mechanisms are able to enforce a
dynamic information flow policy. The second is a simple 
capability system motivated by the Flume operating system. 
\end{abstract} 

\section{Introduction} 

\label{sec:introduction}

Covert channels, i.e., avenues for information flow unintended by the designers of a system, 
present a significant risk for systems that are required to enforce confidentiality, since they 
constitute vulnerabilities that an attacker may exploit in order to exfiltrate secrets. 
This has motivated a body of research that aims to provide means
by which systems can be assured to be free of covert channels. At the highest levels 
of assurance, one aims for formal verification based on rigorous models of the 
system. This requires formal definitions of what it means for a system to 
comply with a policy that specifies the intended information 
flows, and prohibits the unintended information flows.

How to give a formal semantics to such policies that is suitable for formal verification 
has been a topic of study since the 1970's \cite{Cohen77}. 
\emph{Noninterference policies} \cite{Goguen1982} are a form of policy syntax
in the form of a graph over a set of security levels, or domains, 
with an edge representing that information is permitted to flow
from the source of the edge to the destination. Conversely, absence 
of an edge can be understood as prohibiting (direct) information flow. 
In the original formulation of such policies, corresponding to classical multi-level security policies 
based on a lattice of security levels, the graph was transitive, but
there has more recently been a body of work on \emph{intransitive noninterference} policies, 
which do not require transitivity. Here,  a domain is best understood as a component
within a security architecture, with edges describing permitted direct flows of 
information, and paths corresponding to permitted indirect flows of information. In particular, this
understanding underlies MILS security \cite{BDRF08,Vanfleet2005}, an approach to high assurance systems
development in which security architectures comprised of trusted and untrusted components
are enforced using a variety of mechanisms including separation kernels. At its most abstract level, 
the architecture can be understood as an intransitive noninterference policy. 
Several large-scale efforts 
\cite{BytschkowQIR14,HeitmeyerALM06,Martin2000,Murray2012,Schellhorn2002} 
have formally verified that intransitive noninterference policies hold in a number of 
separation kernels.

The proper semantics of intransitive noninterference policies 
has proved to be a subtle matter. Rushby \cite{rushby92} 
improved an initial proposal by Haigh and Young \cite{HY87}, 
giving semantics to policies by means of an \emph{intransitive purge} function.
Van der Meyden~\cite{Meyden15} has recently argued that there is 
a subtle weakness in 
this intransitive purge based semantics. He provides an example that demonstrates that 
this semantics may classify as secure a system that contains
flows of information, based on the ordering of events, that seem
to contradict an intuitive understanding of the policy.  He proposes an 
alternative, stronger, definition of security, \emph{TA-security}, that correctly classifies this
example as insecure. Moreover, TA-security is shown to provide a better
fit for Rushby's proof theory. Rushby shows that this proof theory is sound for 
his definition, but fails to provide a completeness result. Van der Meyden explains this, 
by showing that it is both sound and complete for the stronger notion of TA-security. 
This completeness result is obtained both for Rushby's unwinding proof method
and for a variant of Rushby's results demonstrating that access control systems
provide a sound method for enforcement of intransitive noninterference policies.

This line of work has concerned static policies, in which the 
permitted flows of information do not change over time. 
Such policies are too restrictive for many practical purposes: 
the permitted flows of information in a system could change
as a result of policy changes, system reconfiguration, 
in response to detection
of vulnerabilities or attacks in progress, as a result of
personnel changes, promotions, revocations of privileges, 
delegations, transfers of capabilities, 
and for many other reasons.

A full understanding of the  formal underpinnings of security architecture should accommodate
such dynamic policy changes, and it is therefore of interest to 
have a satisfactory theory of dynamic intransitive noninterference 
policies. To date, only a few works have studied this question: Leslie
\cite{Leslie2006} and Eggert et. al.~\cite{Eggert2013}. Both are based on 
generalizations of  Rushby's ``intransitive purge" operator and, in the case of static 
policies, are equivalent to Rushby's semantics. 
Consequently, these works are subject to the weakness of this semantics 
identified by van der Meyden. 
In particular, they misclassify van der Meyden's example.
Leslie provides an unwinding proof theory for her definition that is  equivalent to
Rushby's unwinding in the static case, hence is incomplete. Eggert et. al. do  
provide a complete proof theory and a complexity characterization for their definition.

Our main contribution in this paper is to develop alternate semantics for dynamic 
intransitive noninterference policies that generalize the more adequate notion of 
TA-security in the static case. Indeed, we show (in Section~\ref{sec:knowledgebased}) that there are several candidates for
such a semantics since, in the dynamic setting, it turns out that the  information 
flows permitted by policy edges are not always the complement of the information flows
prohibited by the absence of policy edges. This leads to two distinct
semantics, one of which prioritizes the permissions and the other of
which prioritizes the prohibitions. A further complexity to policy in the dynamic setting 
is the question of what knowledge the agents in the system are permitted
to have of the state of the policy itself. However, this consideration is
also helpful: we show that \emph{locality}, a natural condition on agent knowledge of the policy, 
provides 
necessary and sufficient conditions
for the permissive and the prohibitive interpretations of
the policy to be equivalent. 
Reference to agent knowledge also enables us to  
characterize precisely the maximal information permitted to be known to each agent
in the system at moment of time, in terms of a concrete representation of the 
information that may have been transmitted to the agent. 
We show 
by some examples 
that our definitions improve on those of Leslie 
\cite{Leslie2006} and Eggert et. al.~\cite{Eggert2013},
not just in the treatment of static policies (van der Meyden's example justifying 
TA-security), but also in the classification of some simple dynamic policies.

We then (in Section~\ref{sec:prooftechniques}) develop proof methods that 
may be used to show that a system is secure, 
with a focus on local policies, so that the technique applies
to both of our definitions of security. We first  generalize the 
classical \emph{unwinding} proof method \cite{goguen84} to our new definitions. 
This is used  (in Section~\ref{sec:accesscontrolinterpretation}) to generalize,  
to a dynamic setting, Rushby's \cite{rushby92} results 
on enforcement of a static intransitive noninterference policy by means of 
access control settings. 
In both cases, we demonstrate not 
just soundness but also completeness of the proof method.

As an 
application of the theory, we establish 
(in Section~\ref{sec:flume})
security for 
a capability system, motivated by  the Flume operating system \cite{Krohn2009}. 
We show that the capability system enforces a naturally 
 associated dynamic intransitive noninterference policy. 
This result gives an illustration of how the general theory 
can be applied to a concrete system.

\section{Basic Definitions} 

\label{sec:definitions}

In this section we give the basic definitions concerning the 
systems model about which we will make security judgements, 
and introduce the syntax of  dynamic intransitive noninterference policies. 
The following section develops several semantics for these  policies in this systems model. 

\subsection{Systems} 

We model systems as deterministic automata in which the actions of multiple agents 
cause state transitions. Agents obtain information about
the state of system, and each other's actions, by making observations of the state of the system. 

A \emph{signature}  is a tuple $\tuple{\Dom, \dom, \Actions}$, 
where $\Dom$ is a finite set of \emph{domains}, (which we may also refer to  as \emph{agents}), 
 $\Actions$ is a finite set of \emph{actions}, and $\dom\colon \Actions \rightarrow \Dom$, 
 assigns each action to a domain (agent). Intuitively,  $\dom(a)$ is the domain within 
which the action $a$ is performed (the agent performing that action). 
We write $\Actions^*$ for the set of finite sequences of elements of the set $A$ and
denote the empty sequence by $\epsilon$.

\newcommand{\trans}{\rightarrow}

An {\em automaton} 
for a signature $\tuple{\Dom, \dom, \Actions}$
is a tuple $\tuple{S, s_0, \trans}$, where $S$ is a set of states (not necessarily finite), $s_0\in S$ is
the \emph{initial state}, and ${\trans} \subseteq S\times A\times S$  is a labelled transition relation. 
This relation is required to be {\em input-enabled}, 
in the  sense that for all states $s\in S$ and actions $a\in A$, 
there exists  $t\in S$ such that $(s,a,t) \in {\trans}$. We also assume
that the relation is deterministic in the sense that if $(s,a,t)\in {\trans}$ and $(s,a,t') \in {\trans}$
then $t = t'$. Under these assumptions, we may write $s\cdot a$ for the unique $t$ such that 
$(s,a,t) \in {\trans}$. 

We make several uses of the general notion of automaton. One is the 
machine model that we analyze for security. 
Formally, a \emph{system} for a signature $\tuple{\Dom, \dom, \Actions}$
is an automaton $\tuple{S, s_0, \trans}$ for the signature, together with 
an observation function $\obs\colon \Dom \times S \rightarrow O$.  
Intuitively, $\obs(u,s)$ is the observation 
made in domain $u$ (by agent $u$) when the system is in state $s$. 
Since the domain $u$ will typically be fixed in a context of application, 
but the state $s$ will vary, we write $\obs_u(s)$ for $\obs(u,s)$. 
We write systems in the form 
$M = \tuple{S, s_0, \trans, \Dom, \dom, \Actions, \obs}$
when we wish to make all the relevant components explicit.

The assumption that the transition relation is input-enabled has often been made in the 
literature.  It prevents enablement of actions being a source of information flow, and 
intuitively represents that it is always possible for an agent to attempt to perform an action 
(even if it would have no actual effect). 
A situation $s$ where an agent is not able to successfully perform an action can be represented 
in an input-enabled model by means of a transition $(s,a,s)$ that leaves the state unchanged. If
it is desired to model that the agent is able to observe that the action is not enabled, 
this can be encoded in the observation that the agent makes in either the source or destination state of the 
transition. (This  leaves the modelling of the way that enablement of actions becomes known to agents
in the hands of the system modeller, as opposed to many process algebraic semantics, e.g., 
those using bisimulation, that would mandate that enablement of actions is always observable.) 

The restriction to deterministic systems is made in part because the theory of 
intransitive noninterference is better understood in deterministic systems \cite{rushby92,Meyden15}
than in nondeterministic systems \cite{Engelhardt2012}. However, in many cases, 
nondeterminism can be represented in deterministic systems by using
additional agents to model the source of nondeterminism. (This does 
require that such additional agents be included in the policy as well as in the systems model, but
that can be argued to add clarity to the meaning of the policy.)

A \emph{run} of a system is a sequence $s_0 a_1 s_1 \ldots a_n s_n$, where $n \geq 0$ and 
the $s_i \in S$ are states (with $s_0$ the initial state of the system) and the $a_i\in \Actions$ are actions, such that for $i=1 \ldots n$ we 
have $s_i = s_{i-1} \cdot a$. We call the sequence $\alpha =a_1 \ldots a_n\in \Actions^*$ the \emph{trace} of the run.  
Conversely, since the system is input-enabled and deterministic, every sequence $\alpha \in \Actions^*$ 
is the trace of a unique run. We may therefore refer to runs by their corresponding traces. 
For a sequence $\alpha = a_1 \ldots a_n \in \Actions^*$ and a state $s\in S$, 
we write $s \cdot \alpha$ for the state reached after executing the sequence of actions $\alpha$ starting at $s$, 
i.e., the state $t_n$, obtained inductively by taking  $t_0 = s$ and 
$t_i = t_{i-1} \cdot a_i$, for $i =1\ldots n$.

Two systems $M = \tuple{S, s_o, \trans, \obs}$ and  $M = \tuple{S', s_0', \trans', \obs'}$ for the same signature $\langle\Dom, \dom,\Actions\rangle$ 
are \emph{bisimilar} if for all $u \in \Dom$ and $\alpha\in \Actions^*$, we have $\obs_u(s_0\cdot \alpha) = \obs'_u(s_0'\cdot \alpha)$. 
(We remark that we do not need the usual relational definition of bisimilarity because our systems are deterministic.) 
It is easily checked that all the definitions of security that  follow are preserved under bisimilarity.

\newcommand{\unwound}{\mathit{unfold}}
The \emph{unfolding} of a system $M$ for signature $\langle \Dom, \dom, \Actions\rangle$ is defined to be the system 
$\unwound(M)$ for the same signature, 
with states $A^*$, initial state $\epsilon$, 
transition relation such that $\alpha \cdot a = \alpha a$ for all $\alpha \in A^*$ and $a\in A$, 
and observation functions $\obs_u(\alpha) = \obs^M_u(s_0\cdot \alpha)$.  
The system $\unwound(M)$  is easily seen to be bisimilar to~$M$, 
so $\unwound(M)$ can be thought of as an 
equivalent implementation of the system $M$.

\subsection{Static Policies} \label{sec:static-policies}

A {\em static policy} describes the permitted and prohibited flows of information between domains. 
Formally, a static policy for the signature $\tuple{\Dom, \dom, \Actions}$ 
 is a reflexive binary relation  $\interferes$ (possibly intransitive) on the set of domains $\Dom$.
Intuitively, $u \interferes v$ means that information is permitted to flow from domain $u$ to domain 
$v$, and conversely, $u \noninterferes v$ means that information flows from domain $u$ to domain $v$ are
prohibited.  

Policies are assumed to be reflexive because, intuitively, nothing can be done to prevent 
information flowing from a domain to itself, so this is allowed by default. 
Since $\interferes$ can be taken to be the edge set of a directed graph with vertex set D, we may use terminology from graph theory 
when we describe properties of a policy. For example, we may call $ u \interferes v$  an ``edge'' of the policy $\interferes$, 
and talk about ``paths" in the relation~$\interferes$. 

We remark that the system model introduced above appears to lack a notion of 
internal (silent) transitions, but these can be  
handled by adding a systems domain $Sys$, such that each internal transition 
corresponds to some action of domain $Sys$. This representation also requires
adding an edge $Sys \interferes u$ to the policy for all other domains $u$.  

A significant body of work already exists in relation to static policies, e.g., \cite{Goguen1982,rushby92,Meyden15,Roscoe1999}. 
Since our focus in this paper is to find semantics for dynamic policies that generalizes the semantics of van der Meyden \cite{Meyden15} 
for static intransitive noninterference policies, we reiterate here just that semantics in order to motivate what follows, and refer the reader to 
the literature for more detailed discussion. The semantics involves for each agent $u$ a function 
$\taname_u$ from the set of traces, such that for a trace $\alpha$, the value 
$\taname_u(\alpha)$ intuitively represents a concrete encoding of the maximal 
amount of information that the agent is permitted to have about $\alpha$. 
The function $\taname_u$ transforms a trace $\alpha$ into a binary tree, which contains less information about the order of the actions than the trace does. 
This 
function is defined inductively by%
\footnote{
Note that the functions $\taname_u$ depend on the policy $\interferes$, but 
to avoid clutter,  this is suppressed in the notation, here and for similar
definitions later in the paper.   We may add the policy to the notation in contexts containing several policies, 
e.g., by writing $\taname^{\interferes}_u(\alpha)$ when the policy used is $\interferes$.}
$\ta \epsilon u  = \epsilon$ and for every $\alpha \in \Actions^*$ and $a \in \Actions$ by
\begin{align*}
 \ta {\alpha a} u & =
 \begin{cases}
  (\ta {\alpha} u, \ta \alpha {\dom(a)}, a) & \text{if } {\dom(a) \interferes u}\\
  \ta {\alpha} u & \text{otherwise} 
  \enspace. 
 \end{cases}
\end{align*}
Intuitively, in case $\dom(a)\noninterferes u$, 
information from $\dom(a)$ is not permitted to flow to $u$, so the occurrence of the
action $a$ causes no change  in the maximal information $\ta \alpha u$ that
may be possessed by $u$ after $\alpha$. 
However, if  $\dom(a)\interferes u$, then the occurrence 
of $a$ after $\alpha$ causes the maximal information that agent $u$ is permitted to have to increase, 
by adding to its prior information   $\ta {\alpha} u$ the maximal 
information $\ta \alpha {\dom(a)}$ that the agent performing $a$ is permitted to have after $\alpha$, 
as well as the fact that the action $a$ has occurred. 
Using the functions $\taname_u$, we may 
now declare a system to be secure, in the sense of complying with the information flow policy $\interferes$, 
as follows: 

\begin{definition} 
A system $M$ is TA-secure with respect to policy $\interferes$, if
for all agents $u$ and $\alpha, \beta\in \Actions^*$, if 
$\ta \alpha u =\ta \beta u$ then $\obs_u(s_0 \cdot \alpha) = \obs_u(s_0 \cdot \beta)$. 
\end{definition} 

Intuitively, this says that for all agents $u$, the observation $\obs_u(s_0 \cdot \alpha) $
contains no more information than the maximal permitted information $\ta \alpha u$.

In general, when considering the deductive capabilities of agents in computer security, it 
is appropriate to assume that an adversary will reason based not just on their
current observation, but also using their past observations. This 
suggests that we should consider a 
\emph{perfect recall} attacker. In the case of TA-security, 
a definition stated using a perfect recall attacker is equivalent to 
one stated using just the final observation of traces \cite{Meyden15}. 
We formulate this result in a more general 
way, so that it also applies to the other definitions we state in this paper. 

First, to capture a perfect recall attacker, we define the
\emph{view} of an agent $u$ inductively, by 
$\view_u(\epsilon) = \obs_u(s_0)$ and
\begin{align*}
 \view_u(\alpha a)& = 
 \begin{cases}
   \view_u(\alpha) \; a \; \obs_u(s_0 \cdot \alpha a) & \text{if } \dom(a) = u \\
   \view_u(\alpha) \circ \obs_u(s_0 \cdot \alpha a) & \text{otherwise}
 \end{cases}
\end{align*}
where every $\alpha \in \Actions^*$ and $a \in \Actions$. 
Here, $\sigma \circ x$ denotes the absorptive concatenation of an element $x$ of a set  $X$ to a string whose final element is in $X$, defined
by $\epsilon x = x$, and $\sigma y\circ x = \sigma y x$ if $x\neq y$ and $\sigma y\circ x = \sigma y$ if $y = x$. 
The reason we apply this type of concatenation is to capture the asynchronous behaviour of the system, 
where agents do not necessarily have access to a global clock. 
Intuitively, agent $u$ is aware of its own actions, so these are always available in its perfect recall 
view. This is captured by the first clause of the definition of $\view_u$. However, when another agent performs an action, this may or may not
cause a change in $u$'s observation. If it does, then $u$ detects the change, and the new observation is added to its view. 
However, if there is no change of observation, then there is no change to the view, since the agent is assumed to operate asynchronously, 
so that it does not know for what length of time it makes each observation.

Consider an indexed collection of functions $f= \{f_u\}_{u \in \Dom}$, 
each with domain $\Actions^*$, and defined 
by $f_u(\epsilon) = \epsilon$ and 
\begin{align*}
 f_u(\alpha a) &= 
 \begin{cases}
	(f_u(\alpha), g(\alpha,a,u))& \text{if } 
	 C(\alpha, a, u) \\ 
	f_u(\alpha) & \text{otherwise}
 \end{cases}
\end{align*}
for some boolean condition $C$ and function $g$,  
where $\alpha \in \Actions^*$, $a\in \Actions$ and $u \in \Dom$. 
Note that the collection $\taname$ is in this pattern, with $g(\alpha,a, u) = (f_{\dom(a)}(\alpha),a)$ and $C(\alpha, a, u) = \dom(a) \interferes u$. 
We say that $f$ is \emph{self-aware} if $\dom(a) = u$ implies $C(\alpha,a,  u)$, 
and $g(\alpha, a, u) = g(\beta, b,u)$ implies $a=b$. Intuitively, in the self-aware case, 
$f_u(\alpha)$ encodes at least a record of all the actions of domain $u$ that occur in $\alpha$. 

Say that a system is \emph{$f$-secure}  
if for all $\alpha, \beta \in \Actions^*$ with $f_u(\alpha) = f_u(\beta)$, we have $\obs_u(s_0 \cdot \alpha) = \obs_u(s_0 \cdot \beta)$. 
An alternative  definition with the $\view$ function instead of $\obs$ is: 
a system is \emph{$f$-view-secure} 
if for all $\alpha, \beta \in \Actions^*$ with $f_u(\alpha) = f_u(\beta)$, we have $\view_u(\alpha) = \view_u(\beta)$. 
These definitions are equivalent, subject to self-awareness of $f$: 

\begin{lemma}
\label{lem:obs-view} 
Suppose that $f$ is self-aware. Then a system $M$ is $f$-secure 
iff it is $f$-view-secure. 
\end{lemma}

\begin{proof}
 Assume first that $M$ is not $f$-view-secure. 
 Then there are $u \in \Dom$ and $\alpha, \beta \in \Actions^*$ of minimal combined length with $f_u(\alpha) = f_u(\beta)$ and $\view_u(\alpha) \neq \view_u(\beta)$. 
 At least one of $\alpha$ and $\beta$ is not the empty trace, suppose that it is $\alpha$ and let $\alpha = \alpha' a$ for some $\alpha' \in \Actions^*$ and $a \in \Actions$. 
 Then there are two cases:
 \begin{itemize}
  \item \emph{Case 1:} $f_u(\alpha' a) = f_u(\alpha')$. 
  Then $f_u(\alpha') = f_u(\beta)$, so 
  from the minimality of $\alpha$ and $\beta$ it follows that $\view_u(\alpha') = \view_u(\beta)$ and hence $\view_u(\alpha') \neq \view_u(\alpha' a)$. 
  Since 
   $C(\alpha, a,u)$ is false, we have by self-awareness that 
  $\dom(a) \neq u$, so $\view_u(\alpha' a) = \view_u(\alpha') \obs_u(s_0 \cdot \alpha'a)$. 
  Therefore, we have $\obs_u(s_0 \cdot \alpha' a) \neq \obs_u(s_0 \cdot \beta)$. 
  \item \emph{Case 2:} $f_u(\alpha' a) \neq f_u(\alpha')$. 
  We can assume that $\beta = \beta'b$ for some $\beta' \in \Actions^*$ and $b \in \Actions$ with $f_u(\beta' b) \neq f_u(\beta')$ since otherwise, we proceed with the first case with the roles of $\alpha$ and $\beta$ swapped. 
  From the inductive definition of $f_u$ 
  and self-awareness, it
  follows that $a = b$ and $f_u(\alpha') = f_u(\beta')$. 
  By the minimality of $\alpha$ and $\beta$, we have 
$\view_u(\alpha') = \view_u(\beta')$. There are two cases of the definitions of $\view_u(\alpha' a)$ and $\view_u(\beta' a)$, 
depending on whether $\dom(a) = u$, but in either case, it follows that 
$\obs_u(s_0 \cdot \alpha'a) \neq \obs_u(s_0 \cdot \beta' b)$, since the 
only difference in these sequences can be in the final observation. 
\end{itemize}
The other direction of the proof follows directly from the definition of the $\view$ function. 
\end{proof}

Since the definition of $f=\taname$ is easily seen to be self-aware (using reflexivity of the policy), it follows from Lemma~\ref{lem:obs-view} 
that TA-security and $\taname$-view-security are equivalent.

\subsection{Dynamic Policies}

Our main concern in this paper is with \emph{dynamic intransitive information flow policies}, 
which generalize static policies by allowing the edges of the policy to depend 
on the actions that have been performed in the system. Formally,
a \emph{dynamic policy} is a relation ${\interferes} \subseteq {\Dom \times A^* \times \Dom}$, 
such that ${(u, \alpha, u)} \in {\interferes}$ for all $\alpha \in \Actions^*$ and all $u \in \Dom$. 
We write $\alpha \models u \interferes v$ when $(u, \alpha ,v)\in {\interferes}$. 
Intuitively,   $\alpha \models {u \interferes v}$ says that, after the sequence of actions
$\alpha$ have been performed, information may flow from domain $u$ to domain $v$. 
Thus, $\interferes_\alpha = \{(u,v) \mid (u,\alpha, v) \in {\interferes} \}$ is the (static) policy that applies after the actions $\alpha$  have been performed. 

Dynamic policies can be represented using 
automata.  We say that a tuple $\tuple{P,\interferes'}$
consisting of an automaton $P = \tuple{S,s_0, \trans}$ and
a relation  ${\interferes'} \subseteq \Dom \times S\times \Dom$ 
\emph{represents} a dynamic policy  $\interferes$, if for all 
$\alpha \in \Actions^*$, domains $u,v, \in \Dom$, and state $s \in  \States$ with $s = s_0 \cdot \alpha$, 
we have $\alpha \models {u \interferes v}$ iff 
$(u, s, v) \in {\interferes'}$.  
(In this case we also write $s \models {u \interferes' v}$.) 
Every dynamic policy has such a representation (with an infinite number of states), since we
may take $S= \Actions^*$ and ${\trans} = \{(\alpha, a, \alpha a)~|~\alpha\in \Actions^*, a\in \Actions\}$. 
We may call the policy \emph{finite state} if it has a representation with $S$ finite. 

Let $\langle \Dom, \dom, \Actions\rangle$ be a signature.  
A \emph{policy enhanced system} for this signature is an automaton $\A=\tuple{S, s_0, \trans}$  that is equipped with 
a relation ${\interferes} \subseteq {\Dom \times S\times \Dom}$  and an observation function $\obs$ with domain $\Dom \times S$.
Given a system $M$ and a dynamic policy $ \interferes'$ for the signature, 
we say that the policy enhanced system $\tuple{\A, \interferes, \obs}$ \emph{encodes} $M$ and $\interferes'$ if 
$\tuple{\A,\obs}$ is bisimilar to $M$ and $\tuple{\A,\interferes}$ represents $\interferes'$.

For every pair consisting of a system $M$ and a dynamic policy $\interferes$ for the same signature, we can construct a policy enhanced system
that encodes $M$ and $\interferes$. 
Given an automaton representation $\tuple{\A_P,\interferes'}$ of  $\interferes$ 
where $\A_P = \tuple {S^P, s^P_0, \trans^P}$,  
and  a system $M \hspace{-2pt}= \hspace{-2pt} \tuple {\A_M,\obs^M}$  where $\A_M = \tuple{S^M, s^M_0, \trans^M}$ for the same signature, 
we define the product automaton 
$A_M \times \A_P
=\tuple{S, s_0, \trans}$, 
where $S = S^M \times   S^P$, $s_0 = (s^M_0, s_0^P)$
and $((s, p), a, (s',p')) \in  {\trans}$  iff $(s,a,s') \in {\trans^M}$ and $(p,a,p') \in {\trans^P}$. 
The automaton 
$\A_M \times \A_P$ 
may be equipped with an observation function $\obs$ with domain 
$S$  and a policy  relation ${\interferes''} \subseteq \Dom \times S\times \Dom$ 
by defining $\obs(u, (s,p)) = \obs^M(u,s)$ and ${\interferes''} = \{(u,(s,p),v)~|~u \interferes_p v\}$. 
It is then straightforward to show that the policy enhanced system $\tuple{ \A_M \times \A_P, \obs, \interferes''}$
encodes $M$ and $\interferes$. 

It will often be convenient, in presenting examples, 
to use a diagram  of a policy enhanced system in order to present the policy and the system 
together as a single automaton.  Figure~\ref{fig:conflict} gives an example of this presentation. 
We note the following conventions to be applied in interpreting these diagrams. 
The systems represented are input-enabled, but we elide self loops 
corresponding to edges of the form $(s,a,s)$ in order to reduce clutter. 
States are partitioned into three components. The top component gives the 
name of the state. The middle component depicts the static policy applicable at that 
state. The bottom component gives partial information about the observations that the 
agents make at the state: typically, only one agent's information is relevant to the discussion, 
and we elide the observations of the other agents in the system.

There is a natural order on policies. 
Consider the binary relation $\leq$ on policies with respect to a given signature $\langle\Dom, \dom, \Actions\rangle$, 
defined by ${\interferes} \leq {\interferes'}$ iff for all $\alpha \in A^*$ and $u,v\in \Dom$, 
we have $\alpha \models {u \interferes v}$ implies $\alpha \models {u \interferes' v}$. 
It is easily seen that this relation partially orders the policies with respect to the signature. 
Intuitively, if ${\interferes} \leq {\interferes'}$ then $\interferes$ places more
restrictions on the flow of information in a system than does $\interferes'$. 
This  intuition can be supported for both a permissive and a prohibitive reading
of policies. On a permissive reading, ${\interferes} \leq {\interferes'}$
 intuitively says that every situation in which a flow of information is explicitly 
 permitted by $\interferes$ is one where the flow of information is explicitly 
 permitted by $\interferes'$. Thus, $\interferes$ is more restrictive than $\interferes'$ in the sense of 
 having fewer explicitly permitted flows of information. 
On a prohibitive reading,  ${\interferes} \leq {\interferes'}$ says (contrapositively) that 
 every situation where $\interferes'$ explicitly prohibits a flow of information is
 one where $\interferes$ also explicitly prohibits that flow of information. 
 Thus, $\interferes$ is more restrictive than $\interferes'$ in the sense of 
 having more explicitly prohibited flows of information. 
We may therefore gloss ${\interferes} \leq {\interferes'}$ 
as stating that $\interferes$ is \emph{more restrictive} than $\interferes'$.

\subsection{Logic of Knowledge} 

It will be convenient to formulate some of our definitions using formulas 
from a logic of knowledge. Given a signature $\tuple{\Dom, \dom, \Actions}$, 
we work with formulas $\phi$ expressed using the following 
grammar: 
$$ \phi ::= u \interferes v~|~\phi\land \phi~|~ \neg \phi~|~ K_u\phi~|~ D_G\phi $$ 
where $u,v \in \Dom$ are domains and $G \subseteq \Dom$ is a set of domains. 
Intuitively, the atomic propositions of the logic are assertions of the form
$u\interferes v$ concerning the static policy holding at a particular point of time. 
We write $\Prop$ for the set of atomic propositions. 
The logic contains the usual boolean operators for conjunction and negation, 
and we freely use other boolean operators that can be defined using these, 
e.g., we write $\phi \rimp\psi$ for $\neg (\phi \land \neg \psi)$.  
The formula $K_u \phi$ intuitively says that domain $u$ \emph{knows} that $\phi$ 
holds,  and $D_G \phi$ says that it is \emph{distributed knowledge} to the 
group $G$ that $\phi$ holds, i.e., the group $G$ would be able to deduce  
$\phi$ if they were to pool all the information held by the members of the group. 

The semantics of the logic is a standard Kripke semantics for 
epistemic logic \cite{fhmvbook}: formulas are interpreted
 in Kripke structures
$\M = \tuple{W, \{\sim_u\}_{u \in \Dom}, \models}$ where 
$W$ is a set, 
 for each $u\in \Dom$, we have an equivalence relation $\sim_u$ on $W$, 
 and ${\models} \subseteq {W\times \Prop}$ is a binary relation.  
 
 Intuitively,  $W$ is a set of worlds, representing possible situations in a system of
interest. The equivalence relation $\sim_u$
intuitively corresponds to indistinguishability of worlds to an agent: 
$w \sim_u w'$ will hold when agent $u$ has the same information available to 
it when it is in situation $w$ as it has in situation $w'$. 
For a world $w$ and an atomic proposition $p$, the relation $w\models p$
represents that the proposition $p$ is true at the world $w$. 
We may also write this as $\M, w\models p$ to make the structure $\M$ explicit. 
In turn, this relation can be extended to a satisfaction relation $\M, w\models \phi$
for arbitrary formulas, by means of the following recursion: 
\begin{tabbing}
$\M, w\models {\phi_1 \land \phi_2}$~~\= \kill 
$\M, w\models {\phi_1 \land \phi_2}$\> if $\M, w\models \phi_1$ and $\M,w\models \phi_2$ \\ 
$\M, w\models {\neg \phi}$ \> if not $\M, w\models \phi_1$ \\ 
$\M, w\models K_u \phi$ \> if  $\M, w'\models \phi$ for all $w'\in W$ with $w \sim_u w'$ \\ 
$\M, w\models D_G \phi$ \> if  $\M, w'\models \phi$ for all $w'\in W$ with $w \sim_u w'$ \\ \> for all $u \in G$  
\end{tabbing} 
In our applications, we will work  with the set of worlds $W =A^*$, i.e., worlds will be traces 
for the signature of interest. The basic relation $w \models p$ 
will be the relation $\alpha \models {u \interferes v}$ from some 
dynamic policy for this signature. We discuss the 
equivalence relations we use later.

\section{Semantics for dynamic policies} 

\label{sec:knowledgebased}

In the setting of dynamic policies, several subtle issues arise that a suitable 
definition of security needs to take into account.  

One is that policies can be interpreted with a focus on positive or negative edges. 
One can read $\alpha \models {u \interferes v}$ as stating a permission: 
in state $\alpha$, actions of $u$ may pass information from $u$ to $v$ 
(even if $u$ or $v$ does not know that  $\alpha \models {u \interferes v}$). 
Nothing about the policy or system may override this permission, and an agent
cannot be sanctioned for having caused the information flow. 
Alternately,  and more restrictively, one could focus on the converse: 
and treat $\alpha  \models {u \noninterferes v}$ as stating a prohibition: in state $\alpha$, 
information may not flow from $u$ to $v$.  The interaction between 
various such prohibitions may have the effect that even where there is 
an edge $\alpha \models {u \interferes v}$, it is in fact prohibited 
for information to flow from $u$ to $v$ because of derived prohibitions. 
The following example suggests that these two readings of the policy may be in conflict. 

\begin{example} \label{ex:conflict-example} 
Consider the system in Figure~\ref{fig:conflict}. 
Intuitively, domain $P$ is a policy authority that controls the 
policy between domains $A$ and $B$ by means of its action $p$, and 
$a$ is an action of domain $A$.

\begin{figure}%[H] -- fac cls does not like placement modifiers 
\begin{center} 
\scalebox{0.8}{
 \begin{tikzpicture}[tikzglobal,node distance=1cm]
 
    \newcommand{\agents}{ 
      \node[agent] (P) [left] {$P$};
      \node[agent] (A) [right of=P] {$A$};
      \node[agent] (B) [below of=A] {$B$};
      }

 \node[initial,systemstate] (s0) {
   $s_0$ 
   \\
   \hline
   \agents
   \\
   \hline
   $\obs_B \colon 0$ \\
  };
%   \path[policy] (A) edge (L);
   
 \node[systemstate] (s1) [right=of s0] {
 $s_1$ 
 \\
 \hline
  \agents
  \\
   \hline
   $\obs_B \colon 0$ \\
 };
   \path[policy] (A) edge (B);
   
 \node[systemstate] (s2) [right=of s1] {
 $s_2$
 \\
 \hline
  \agents
  \\
   \hline
   $\obs_B \colon 1$ \\
 };
    \path[policy] (A) edge (B);
   
 \path (s0) edge node {$p$} (s1) 
 (s1) edge node {$a$} (s2)
 ;
\end{tikzpicture}}
\end{center} 
\caption{Conflicting permissive and prohibitive interpretations}
\label{fig:conflict}
\end{figure}
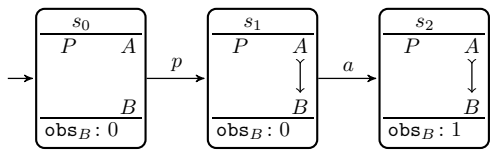 
This system displays a conflict in the transition from state $s_1$ to state $s_2$,  
between the permissive reading of the edge $A \interferes B$ 
and the prohibitive reading of the edge $P \noninterferes B$.  
On the one hand, according to a permissive reading of the policy, the fact that the action $a$ causes a change in the 
observation in domain $B$ is not an insecurity, because the policy explicitly permits 
a flow of information from domain $A$ to $B$ in state $s_1$, and what $B$ learns in state $s_2$ is that 
$A$ has  performed action $a$. On the other hand, in no state of the policy is there an edge from domain $P$ to domain $B$. 
Even the union of all possible policy states does not contain a path from $P$ to $B$. 
Thus, 
we expect from a prohibitive reading of the policy that there should be no flow of information 
from domain $P$ to domain $B$. On this reading, the transition from $s_1$ to $s_2$ displays 
an insecurity in the system, since on making observation 1 domain $B$ can deduce not just that 
action $a$ has occurred, but also that it has been preceded by action $p$ in domain $P$. 
Thus, the system contains a prohibited flow of information from $P$ to $B$, and should be judged to be insecure. 
\end{example}

There are several attitudes one can take in the face of such apparent conflicts within a policy. 
One is to allow explicit permissions to take precedence over prohibitions. Another is to 
require that prohibitions take precedence over permissions.  Finally, we can restrict policies
to instances that do not contain such conflicts. We will consider all these possibilities. 

For static policies, it is safe to assume that the policy is known to all agents in the system:  
indeed, it can be assumed to be common knowledge. A second subtlety for dynamic policies is that the current 
policy state depends on past actions, and not all agents will be permitted to know the 
entire past. This implies that there may be situations where some agents do not 
have complete information about the current policy state. This causes difficulty 
for any agent that is required to enforce compliance with the policy. 
For example, in the policy of Example~\ref{ex:conflict-example}, agent $A$ is never permitted
to know whether $P$ has performed action $p$, so cannot be relied upon to 
enforce a policy restriction on the transfer of information about its actions to agent $B$, as applies in the initial policy state, but not
in state $s_1$. We will approach this issue below by interpreting the policy in a way 
takes into account what  agents know about the policy state.

\subsection{Permissive Interpretation} 

We first formulate a definition that captures a permissive reading of a  
dynamic policy, while generalizing TA-security. 
The methodology underlying the notion of TA-security is to 
construct a concrete representation of the maximal information that
each agent is permitted to have after each sequence of actions, and then to 
state that a system is secure if it has no more than that permitted information. 
For dynamic policies, with a permissive interpretation of edges $\alpha \models {u \interferes v}$, 
this leads to the following inductive definition of an operator  $\maytaname$. It is defined for every domain $u$, 
trace $\alpha \in \Actions^*$ and action $a \in \Actions$ by $\mayta \epsilon u = \epsilon$ and 
\begin{align*}
 \mayta {\alpha a} u & =
 \begin{cases}
  (\mayta {\alpha} u, \mayta \alpha {\dom(a)}, a) & \text{if } \alpha \models {\dom(a) \interferes u}\\
  \mayta {\alpha} u & \text{otherwise} 
  \enspace. 
 \end{cases}
\end{align*}
We also obtain a derived equivalence relation $\sim^\Diamond_u$ on $\Actions^*$ for each $u \in \Dom$, 
defined by $\alpha \sim^\Diamond_u \beta$ iff $\mayta{\alpha} u= \mayta{\beta} u$. 

Intuitively, 
$\alpha\models \dom(a) \interferes u$ 
states that it is  permitted,  
after trace $\alpha$, 
for
$\dom(a)$ to send information to $u$, and $\dom(a)$ exercises this by sending to $u$ all the 
information that it has (viz. $\mayta \alpha {\dom(a)}$), as well as the fact that it is performing the
action $a$. The agent $u$ receives this information and appends it to the information that it already has
(viz., $\mayta {\alpha} u$.) In case $ \alpha \models {\dom(a) \noninterferes u}$, on the other 
hand, there is no transfer of information from $\dom(a)$ to $u$, since the policy 
does not permit this.

If we take $\maytaname$ as a formalization of the maximal information permitted to an agent, we get the 
following definition of security, in the same pattern as TA-security. 

\begin{definition} 
A system $M$ is $\maytaname$-secure with respect to a policy $\interferes$ if for all domains $u$ and sequences of actions $\alpha,\beta$, 
if $\mayta{\alpha} u= \mayta{\beta} u$ then $\obs_{u}(s_0\cdot \alpha) = \obs_{u}(s_0\cdot \beta)$. 
\end{definition}

We remark that the definition of $\maytaname$ is  self-aware 
(since by reflexivity $\dom(a) = u$ implies $\alpha \models {\dom(a) \interferes u}$ for all $\alpha$)
so it follows from Lemma~\ref{lem:obs-view} 
that $\maytaname$-security and $\maytaname$-view-security are equivalent. 

The following example shows that $\maytaname$-security takes a 
permissive interpretation of Example~\ref{ex:conflict-example}.

\begin{example} \label{ex:conflict-ta}
Treating the diagram of Figure~\ref{fig:conflict} 
as a definition of the policy, we obtain that for any 
sequence $\alpha$ of the form $a^*$ or $a^*p$, 
we have $\mayta{\alpha} B  = \epsilon$.  
On the other hand $\mayta{pa} B  = (\epsilon, \epsilon,a)$, 
so every sequence $\beta$ such that $\mayta{pa} B  = \mayta {\beta} B$
has the property that it contains a $p$ followed later by an $a$. Thus, 
if we interpret $\mayta{pa} B $ as a concrete representation of what 
$B$ is allowed to know after the sequence $pa$, we conclude that $B$ 
is permitted to know that there has been an occurrence of action $p$. 
Indeed, for any sequence $\beta$ with $\mayta{pa} B  = \mayta {\beta} B$ we
have $s_0\cdot \beta = s_2= s_0\cdot pa$, so 
there is not a witness to insecurity of the system at
the sequence $pa$. The reader may confirm that this 
system is in fact $\maytaname$-secure. 

We remark here that we have $\mayta{p} A = \epsilon = \mayta{\epsilon} A$ and 
$\mayta{p} B = \epsilon = \mayta{\epsilon} B$, but 
$\mayta{pa} B = (\epsilon, \epsilon, a) \neq  \epsilon = \mayta{a} B$. 
That is, in general, the value $\mayta{\alpha a} B$ depends on more than the values  $\mayta{\alpha} B$ and $\mayta{\alpha} A$ and the action $a$; 
additional information (namely, does $\alpha\models {\dom(a) \interferes u}$ hold) about the sequence $\alpha$ is required to determine 
$\mayta{\alpha a} B$. 
We discuss the significance of this below. 
\end{example}

\subsection{Prohibitive Interpretation} 

Example~\ref{ex:conflict-ta} shows that if we would like instead to prioritize prohibitions over permissions when interpreting a
policy, we need a different representation of the maximal information permitted to an agent than 
that provided by the function $\maytaname$. It is not completely obvious what should be the 
deduced prohibitions that prevent the use of an edge $\alpha \models {\dom(a) \interferes u}$ 
as a justification for the transmission of information from $\dom(a)$ to $u$ in the transition from state 
$\alpha$ to state $\alpha a$. 

However, we may apply some general constraints to suggest a definition. 
First of all, note that for static policies, the function $\taname$ 
has the property that the information available to domain $u$ immediately after 
an action $a$ performed by domain $\dom(a)$, should depend only on the action $a$ and the 
information possessed by domains $u$ and $\dom(a)$. The function $\maytaname$ does not have this property: 
see the remark at the end of Example~\ref{ex:conflict-ta}. 

A general way to ensure that the state of information of domain $u$ 
after action $a$ depends only on the prior states of information of domains $u$ and $\dom(a)$, 
is to condition the definition of the maximal information permitted to be known to $u$ on the 
distributed knowledge of the group $\{u, \dom(a)\}$. This suggests the following inductive definition
of a function $\musttaname$, similar to $\taname$, defined by $\mustta \epsilon u  = \epsilon$ and
if $\alpha \models D_{\{\dom(a),u\}}(\dom(a) \interferes u)$, then 
\begin{equation*}
 \mustta {\alpha a} u = (\mustta {\alpha} u, \mustta \alpha {\dom(a)}, a)
\end{equation*}
and else
$ \mustta {\alpha a} u = \mustta {\alpha} u$.

The formula 
$D_{\{\dom(a),u\}}(\dom(a) \interferes u) \rimp {\dom(a) \interferes u}$ is a 
validity of epistemic logic, so this definition strengthens the condition under which information is transmitted in the definition of $\maytaname$. 
Intuitively, this definition permits transmission of $\dom(a)$'s information to $u$ only when 
it is distributed knowledge to $u$ and $\dom(a)$ that $\dom(a) \interferes u$, 
i.e., there is not the possibility, so far as these domains jointly know, that there is an 
edge $\dom(a) \noninterferes u$ that prohibits the transmission of information
from $\dom(a) $ to $u$. Thus, the  definition takes the point of view that transmission
is permitted whenever the agents would jointly be able to confirm that there is not 
an explicit prohibition to the transmission. 

As it stands, the definition of $\musttaname$ is incomplete, because in order to 
interpret the condition $\alpha \models D_{\{\dom(a),u\}}(\dom(a) \interferes u) $, 
we need an appropriate Kripke structure in order to evaluate the 
distributed knowledge operator. We may take $\Actions^*$ to be the set of 
worlds of this structure, and derive satisfaction of atomic propositions from the policy, but it remains to specify
equivalence relations $\sim_v$ for $v\in \Dom$. 
Intuitively, 
an equivalence $\alpha \sim_v \beta$ holds when agent $v$ is not able to 
distinguish between the traces $\alpha$ and $\beta$: whenever the system is
in state $\alpha$, agent $v$ considers it possible that the system is in state $\beta$. 
What these equivalence relations should be is 
not immediately apparent. We take the approach 
of identifying some constraints on these equivalence relations. 

First, consider the general constraint, that the way agent $u$'s state of 
information is updated when action $a$ is performed, should depend only on the 
prior state of information of $u$ and $\dom(a)$. This can be captured at the 
level of the relations 
$\sim_v$
by means of the following constraint: 
\begin{itemize} 
\item[(WSC)] For all traces $\alpha, \beta$, and actions $a$, if $\alpha \sim_u\beta$ and $\alpha \sim_{\dom(a)}\beta$ then 
$\alpha a\sim_u\beta a$.
\end{itemize}
(We call this constraint WSC since it is essentially the same as Rushby's condition of \emph{Weak Step Consistency} \cite{rushby92}.) 
Although it  does not mention the policy explicitly, this constraint is not inconsistent 
with the state of the policy being a factor in how $u$'s knowledge is updated. However, it 
requires that the way that the policy is taken into account depends only on information
jointly available to domains $u$ and $\dom(a)$. 

Next, to factor in the prohibitions to information transfer 
implied by the 
policy, note that one case where the policy clearly prohibits transfer of 
information from a domain $\dom(a)$ to $u$ when action $a$ is performed
after sequence $\alpha$ is where $\alpha \models \dom(a) \noninterferes u$. 
Thus, domain $u$'s state of information in state $\alpha a$ should be the same 
as its information in state $\alpha$ in this case. This suggests the 
following constraint: 
\begin{itemize} 
\item[(DLR)] For all traces $\alpha$, and actions $a$, if $\alpha \models \dom(a) \noninterferes u$ then 
$\alpha a\sim_u\alpha$.
\end{itemize} 
The nomenclature DLR arises from the fact that this is essentially a dynamic generalization of Rushby's condition of \emph{Locally Respects} \cite{rushby92}. 
A slightly stronger alternative of DLR is the following: 
\begin{itemize} 
\item[(DLR$'$)] For all traces $\alpha$, and actions $a$, if there exists a trace $\beta$ such that 
$\alpha \sim_u\beta$ and $\alpha \sim_{\dom(a)}\beta$ and 
$\beta \models \dom(a) \noninterferes u$ then 
$\alpha a\sim_u\alpha$.
\end{itemize} 
In fact, in the context of WSC, this is equivalent to DLR.

\begin{proposition}\label{prop:LRprime} 
If WSC then DLR iff DLR$'$.  
\end{proposition}  

\begin{proof} 
We assume WSC and DLR and show DLR$'$. (The other direction is straightforward.) 
Suppose that 
$\alpha \sim_u\beta$ and $\alpha \sim_{\dom(a)}\beta$ and 
$\beta \models \dom(a) \noninterferes u$.  We need to show 
$\alpha a\sim_u\alpha$. By WSC, we have $\alpha a \sim_u \beta a$. 
Since $\beta \models \dom(a) \noninterferes u$, we have by DLR that 
$\beta a \sim_u \beta$. Two applications of transitivity now yield that 
$\alpha a \sim_u \alpha$, as required. 
\end{proof}

Suppose that we now accept WSC and DLR as reasonable constraints on the
equivalence relations $\sim_u$ corresponding to the maximal knowledge that an agent $u$ is permitted 
to have. Indexed collections of equivalence relations may be partially ordered by pointwise containment. 
(That is, we take \mbox{$\{\sim_u\}_{u\in D}  \leq \{\sim'_u\}_{u\in D}$} if ${\sim_u} \subseteq {\sim'_u}$ for all $u\in D$.)  
Given the dynamic policy $\interferes$, let $\{\unw_u\}_{u\in D}$ be the 
smallest collection of equivalence relations that is consistent with WSC and DLR. 
Intuitively, these smallest relations are the ones that
allow for the {\em maximal} amount of information flow in the system that is consistent with the 
constraints. (The smaller the equivalence class, the larger the agent's amount of knowledge.)
We call these relations the \emph{unwinding relations}. 
The following fact will be useful below.

\begin{lemma} \label{lem:unwmon}
Suppose  ${\interferes} \leq {\interferes'}$ and let $\{\unw_u\}_{u\in \Dom}$ and $\{\unwprime_u\}_{u\in \Dom}$ be the 
unwinding relations for the policies $\interferes$ and $\interferes'$, respectively. 
Then ${\unwprime_u} \subseteq {\unw_u}$ for all $u\in \Dom$. 
\end{lemma} 
\begin{proof} 
Since $\dom(a) \noninterferes' u$ implies $\dom(a) \noninterferes u$, there are fewer cases of 
basic facts $\alpha a\unw_u \alpha$ than basic facts  $\alpha a\unwprime_u \alpha$
derivable using DLR. By monotonicity of WSC and the rules for an equivalence relation, 
it follows that ${\unwprime_u} \subseteq {\unw_u}$ for all $u\in \Dom$. 
\end{proof}

We now take the relations $\{\unw_u\}_{u\in D}$ as the equivalence relations 
in the Kripke structure needed to interpret the distributed knowledge operator in the 
definition of the functions $\musttaname_u$. We call the derived security notion
associated to these functions \emph{$\musttaname$-security}.  
(Note that since, by reflexivity, if $\dom(a) = u$ then $D_{\{u, \dom(a)\}}(\dom(a) \interferes u)$ 
is valid in this Kripke structure, these functions are self-aware, so 
again, by Lemma~\ref{lem:obs-view},  $\musttaname$-security and $\musttaname$-view-security
are equivalent.) 
For $u\in D$ we also write $\sim^\Box_u$ for the relation on traces defined by $\alpha \sim^\Box_u \beta$ 
iff $\mustta \alpha u = \mustta \beta u$.  
 
However, we note that we could also take the 
unwinding 
relations directly as the 
basis for another definition of security: 

\begin{definition} 
A system is \emph{unwinding-secure} with respect to a policy $\interferes$ 
if, 
for all traces $\alpha,\beta$ and domain~$u$, 
we have $\alpha \unw_u\beta$ implies $\obs_u(s_0 \cdot \alpha) = \obs_u(s_0 \cdot \beta)$.  
\end{definition}  

In fact, it turns out that this does not give a new definition of security. 
The following result shows that the 
unwinding relations 
correspond precisely to the information 
flow modelled by the function $\musttaname$.

\begin{theorem} 
\label{thm:mustunwind} 
We have 
${\unw_u} = {\sim^\Box_u}$ for all $u\in D$. 
\end{theorem} 

\begin{proof}
We establish the containments ${\unw_u} \subseteq {\sim^\Box_u}$ and ${\sim^\Box_u} \subseteq {\unw_u}$. 

For  ${\unw_u} \subseteq {\sim^\Box_u}$, we argue by induction on the number of steps of a derivation 
of the statement $\alpha \unw_u \beta$ using rules WSC, DLR and the rules for an equivalence relation. 
The cases for reflexivity, transitivity an symmetry are direct from the fact that $\sim^\Box_u$ has these properties. 
For the DLR case, suppose that $\alpha \models {\dom(a) \noninterferes u}$, so that we have 
$\alpha a \unw_u \alpha$. By validity of the formula $D_G(\phi) \rimp \phi$, we  have also 
$\alpha \models D_{\{\dom(a), u\}} (\dom(a) \noninterferes u)$, so by the definition of $\musttaname$ 
we have $\alpha a \sim^\Box_u \alpha$, as required. 
For the WSC case, suppose that 
$\alpha \unw_u \beta$ and $\alpha\unw_{\dom(a)}\beta$, so that $\alpha a \unw_u \beta a$ by WSC. 
By the inductive hypothesis, we have 
$\alpha \sim^\Box_u \beta$ and $\alpha\sim^\Box_{\dom(a)}\beta$, 
i.e., $\mustta \alpha u = \mustta \beta u$ and $\mustta \alpha {\dom(a)} = \mustta \beta {\dom(a)}$. 
Moreover, it also follows that 
 $\alpha \models D_{\{\dom(a), u\}} (\dom(a) \noninterferes u)$ iff 
 $\beta \models D_{\{\dom(a), u\}} (\dom(a) \noninterferes u)$, so that the 
 same case of the definition of $\musttaname$ is selected in both 
 $\mustta { \alpha a} u$ and   $\mustta { \beta a} u$. 
 It follows that $\mustta { \alpha a} u = \mustta { \beta a} u$, 
 i.e., $\alpha a\sim^\Box_{\dom(a)}\beta a$, as required. 

For  ${\sim^\Box_u} \subseteq {\unw_u}$, we argue 
by induction on the combined length of $\alpha$ and $\beta$
that $\alpha \sim^\Box_{\dom(a)}\beta $ implies $\alpha \unw_{\dom(a)}\beta $.
The base case of $\alpha = \beta = \epsilon$ is trivial by reflexivity. 
By symmetry, it suffices to consider the case of $\alpha = \alpha a$ and $\beta$, where the result
has already been established for sequences of shorter combined length.
We suppose that  $\alpha \sim^\Box_{\dom(a)}\beta $, i.e., 
$\mustta {\alpha a } u = \mustta {\beta} u$, and show that $\alpha \unw_{\dom(a)}\beta $.
We consider two cases, corresponding to the cases of the 
definition of $\musttaname$. 

In the first case, we have $\alpha \models D_{\{\dom(a), u\}} (\dom(a) \noninterferes u)$, 
so that $\mustta {\alpha} u = \mustta {\alpha a} u$, and hence $\mustta {\alpha} u = \mustta {\beta} u$. 
By induction, we have $\alpha  \unw_u \beta$. 
Since $\alpha \models D_{\{\dom(a), u\}} (\dom(a) \noninterferes u)$, 
there exists a trace $\gamma$ such that 
$\alpha \unw_u \gamma$ and $\alpha \unw_{\dom(a)} \gamma$ and 
$\gamma \models \dom(a) \noninterferes u$. The relation $\unw_u$ satisfies WSC and DLR by definition, 
so by Proposition~\ref{prop:LRprime} it also satisfies DLR$'$. 
Thus,  we obtain that $\alpha a \unw_{u} \alpha $. It now follows by transitivity that $\alpha a \unw_u\beta$, as required. 

In the second case, we have $\alpha \models D_{\{\dom(a), u\}} (\dom(a) \interferes u)$,
so that $\mustta {\alpha a} u = (\mustta \alpha u, \mustta \alpha {\dom(a)}, a)$. 
Assuming that   $\mustta {\alpha a} u = \mustta \beta u$, it follows that $\beta \neq \epsilon$, 
so we may write $\beta = \beta' b$. If 
$\beta' \models D_{\{\dom(a), u\}} (\dom(a) \noninterferes u)$, then we may switch the roles
of $\alpha a$ and $\beta' b$ and argue as above, so it suffices to consider the case where 
$\beta' \models D_{\{\dom(a), u\}} (\dom(a) \interferes u)$. In this case, 
we have $\mustta {\beta' b} u = (\mustta {\beta'} u, \mustta {\beta'} {\dom(b)}, b)$, 
so it follows from $\mustta {\alpha a} u = \mustta \beta u$ that 
$a = b$ and 
$\mustta \alpha u = \mustta {\beta'} u$ 
and $ \mustta \alpha {\dom(a)} = \mustta {\beta'} {\dom(a)}$. 
That is, we have $\alpha \unw_u \beta'$  and $\alpha \unw_{\dom(a)} \beta'$. 
By WSC, it follows that $\alpha a \unw_u \beta' a = \beta$, as required. 
\end{proof}

It is immediate from Theorem~\ref{thm:mustunwind}  that unwinding security and $\musttaname$-security 
are equivalent. Thus, these definitions lend support to each other and help to 
explain each other: the functions $\musttaname$ give an intuitive, concrete description of 
the information flows that are permitted when security is interpreted using the unwinding relations. 

Moreover, Theorem~\ref{thm:mustunwind}  shows that $\musttaname$-security provides a feasible 
general framework for security definitions. To instantiate this framework, we 
must provide a collection of equivalence relations 
$\{\sim_u\}_{u\in D}$ 
to be used 
to interpret the distributed knowledge operator in the definition.   It is reasonable 
to ask that however we do so, it should be the case that 
the resulting relations $\{\sim^\Box_u\}_{u\in D}$ are exactly the relations 
$\{\sim_u\}_{u\in D}$, 
so that the definition of $\musttaname$ is self-consistent.%
\footnote{It is worth noting that the circularity here is similar to the circularity in the 
semantics of {\em knowledge-based programs} \cite{fhmvbook}.} 
Theorem~\ref{thm:mustunwind} shows that this condition 
is in fact satisfiable. 

There remains some circularity in this justification of the definitions of unwinding security and 
$\musttaname$-security. In particular, these definitions
are both based on the assumption that the constraints WSC and DLR are the only constraints
that should be applied to the flow of information in order to satisfy the policy, and that the 
resulting relations give an acceptable notion of permitted agent knowledge. However, if a case
can be made that further constraints should be placed on the semantics, then a comparison
similar to that made above can be considered. 

The next lemma shows that $\musttaname$-security implies $\maytaname$-security. 

\begin{lemma} 
\label{mayta-cont-unwind} 
We have $\mayta \alpha u  = \mayta\beta u$ implies $\alpha \unw_u\beta$ for all $\alpha, \beta \in \Actions^*$ and all $u \in \Dom$. 
Moreover, $\musttaname$-security implies $\maytaname$-security. 
\end{lemma} 

\begin{proof} 
By induction on the combined length of $\alpha$ and $\beta$. 
The base case is trivial by reflexivity. Consider the case of
$\alpha a $ and $\beta$, and suppose $\mayta {\alpha a}  u = \mayta \beta u$.

In case  $\alpha \models {\dom(a) \noninterferes u}$, we have $\mayta {\alpha a} u = \mayta \alpha u = \mayta \beta u$, 
so by induction we obtain $\alpha \unw_u \beta$. Also, by DLR, we have $\alpha a \unw_u \alpha$. 
It follows that $\alpha a \unw_u \beta$, as required. 

Alternately, if $\alpha \models {\dom(a) \interferes u}$, 
then we have $\mayta {\alpha a} u = (\mayta {\alpha} u,  \mayta {\alpha} {\dom(a)} , a)$. 
It follows that $\beta \neq\epsilon$, so let $\beta = \beta' b$ for $b \in A$. 
In the case that $\beta' \models {\dom(b) \noninterferes u}$, we may swap the roles of $\alpha a$ and $\beta' b$ and
argue as in the previous paragraph. In case $\beta' \models {\dom(b) \interferes u}$, we have 
$\mayta {\beta' b} u = (\mayta {\beta'} u,  \mayta {\beta'} {\dom(b)} , b)$,  
and it follows that $a=b$ and $\mayta \alpha u = \mayta {\beta'} u$ and $\mayta \alpha {\dom(a)} = \mayta {\beta'} {\dom(a)}$. 
By induction, we obtain that $\alpha \unw_u \beta'$ and $\alpha \unw_{\dom(a)}\beta'$.
It now follows by WSC that $\alpha a \unw_u \beta' a = \beta$, as required.   
\end{proof}

The following example shows that $\musttaname$-security differs from $\maytaname$-security. 
\begin{example}
 As already seen in Example~\ref{ex:conflict-ta}, the system in Figure~\ref{fig:conflict} is $\maytaname$-secure. 
 But, we have that $p \unw_B \epsilon$ and $p \unw_A \epsilon$ by WSC. 
 By applying DLR, we obtain $pa \unw_B a$. 
 Using Theorem~\ref{thm:mustunwind}, we have $\mustta {pa} B = \mustta a B$, and since 
 $\obs_B(s_0 \cdot pa) \neq \obs_B(s_0 \cdot a)$, the system is not $\musttaname$-secure. 
\end{example}

\subsection{A sufficient condition for equivalence}

The outcome of the discussion so far is that we have identified two definitions of 
security of dynamic policies, one of which ($\maytaname$-security) takes a permissive
interpretation of policies, and the other $(\musttaname$-security, or equivalently, unwinding security)
takes a prohibitive interpretation. As noted above, these two interpretations
may be in conflict for some policies. It is therefore of interest to 
understand when a policy is free of this conflict. The 
following notion is useful in this regard. 

\begin{definition} 
\label{def:locality}
A policy $\interferes$ is {\em local} 
with respect to a collection of equivalence relations 
$\{\approx_u\}_{u\in D}$ on $A^*$ if, 
for all domains $u,v\in D$ and traces $\alpha, \beta\in A^*$, if  
$\alpha \approx_u \beta $ and $\alpha \approx_v \beta $ 
then $\alpha \models {u \interferes v}$ iff $\beta \models {u \interferes v}$. 
\end{definition} 

An equivalent statement is that 
$$\alpha \models D_{\{u,v\}}(u \interferes v) \lor D_{\{u,v\}}(u \noninterferes v)$$
for all traces $\alpha$ and domains $u,v$, 
when the equivalence relation used to  interpret knowledge is $\approx$-equivalence. 
That is, a pair of agents always have distributed knowledge of whether 
one may interfere with the other. 

\begin{example} 
The policy of Example~\ref{ex:conflict-example} is  not a local policy
with respect to the relations $\{\sim^\Diamond_u\}_{u\in D}$. For, 
we have $\epsilon \sim^\Diamond_A p$ and $\epsilon \sim^\Diamond_B p$, 
but $\epsilon \models A \noninterferes B$ and $p \models A \interferes B$. 
We present some examples of classes of local policies in Section~\ref{sec:specialcases}. 
\end{example} 

We could instantiate the definition of locality with respect to either the equivalence relations
$\{\sim^\Diamond_u\}_{u\in D}$ or $\{\sim^\Box_u\}_{u\in D}$ (equivalently, the unwinding relations
$\{\unw_u\}_{u\in D}$). 
In fact, these definitions prove to be equivalent. 
We first remark that, in general, the relations $\unw_u$ differ from the relations  $\sim^\Diamond_u$. In particular, 
$\alpha \unw_u \beta $ does not imply $\alpha \sim^\Diamond_u \beta$. 
\begin{example}
 In this example we proceed with the system of Figure~\ref{fig:conflict}. 
 By DLR, we have from $\epsilon \models {P \noninterferes B}$ that $p \unw_B \epsilon$ and from $\epsilon \models {P \noninterferes A}$ that $p \unw_A \epsilon$. 
 By WSC, we have $p a \unw_B a$. 
 However, the $\maytaname$ values of $B$ differ: $\mayta{p a}{B} = (\epsilon, \epsilon, a) \neq \epsilon = \mayta a B$. 
\end{example}

However, the relations $\unw_u$ and $\sim^\Diamond_u$ \emph{are} identical 
 for policies that are local with respect to the equivalence relations 
$\{\sim^\Diamond_u\}_{u\in D}$.  

\begin{lemma} 
\label{lemma:eqi-mayta-unwind}
Let $\interferes$ be a local policy with respect to $\{\sim^\Diamond_u\}_{u\in D}$.
Then  $\{\sim^\Diamond_u\}_{u\in D} =  \{ \unw_u\}_{u\in D}$. 
\end{lemma} 
\begin{proof} 
Suppose 
that the policy $\interferes$ is local with respect to $\{\sim^\Diamond_u\}_{u\in D}$.  Then 
these relations satisfy WSC. For, suppose that $\mayta \alpha u = \mayta {\beta} u$ and $\mayta \alpha {\dom(a)} = \mayta {\beta} {\dom(a)}$. 
Note that, by locality, we have that   $\alpha \models {\dom(a) \interferes v}$ iff $\beta \models {\dom(a) \interferes v}$. 
It follows that the same case is chosen in the definitions of 
$\mayta {\alpha a} u$ and $\mayta {\beta a} u$, and we obtain that  $\mayta {\alpha a} u = \mayta {\beta a} u$.

It is immediate from the definition of $\maytaname$ that the relations 
$\{\sim^\Diamond_u\}_{u\in D}$ satisfy DLR. By 
the previous paragraph, these relations also 
satisfy WSC. Since $\{ \unw_u\}_{u\in D}$ are the minimal relations satisfying WSC and DLR, 
we obtain that $\{ \unw_u\}_{u\in D} \leq \{\sim^\Diamond_u\}_{u\in D}$. 
The converse also holds by Lemma~\ref{mayta-cont-unwind}.  
\end{proof} 

From this, we obtain the claimed independence of locality on the choice of equivalence relation. 

\begin{lemma} 
\label{lem:locality-equiv} 
A given policy $\interferes$ is local with respect to the associated
relations $\{\sim^\Diamond_u\}_{u\in D}$ iff it is 
local with respect to the relations $\{\sim^\Box_u\}_{u\in D}$. 
\end{lemma} 

\begin{proof}
 First suppose that $\interferes$ is local with respect to $\{\sim^\Diamond_u\}_{u\in D}$. 
 Then by Lemma~\ref{lemma:eqi-mayta-unwind} and Theorem~\ref{thm:mustunwind}, we have 
 $\{ \unw_u\}_{u\in D} = \{\sim^\Diamond_u\}_{u\in D} =\{\sim^\Box_u\}_{u\in D}$. 
 Hence, $\interferes$ is local with respect to $\{\sim^\Box_u\}_{u\in D}$. 
 
 Suppose now that $\interferes$ is local with respect to $\{\sim^\Box_u\}_{u\in D}$. 
 Let $u, v \in \Dom$ and $\alpha, \beta \in \Actions^*$ with $\mayta \alpha u = \mayta \beta u$ and $\mayta \alpha v = \mayta \beta v$. 
 By Lemma~\ref{mayta-cont-unwind}, we have $\alpha \unw_u \beta$ and $\alpha \unw_v \beta$. 
 By Theorem~\ref{thm:mustunwind}, we have $\alpha \sim^\Box_u \beta$ and $\alpha \sim^\Box_v \beta$ and hence, $\alpha \models {u \interferes v}$ iff $\beta \models {u \interferes v}$. 
\end{proof}

Consequently, in the sequel we say simply that the policy is local,  
if it is local with respect to either of these relations.

The following result shows that 
for local policies, the conflict between a permissive and a prohibitive  reading of policies does not arise. 

\begin{theorem} 
\label{thm:local-must-may} 
If $\interferes$ is a local policy then $\maytaname$-security and $\musttaname$-security with 
respect to $\interferes$ are equivalent. 
\end{theorem} 

\begin{proof} 
This is immediate from the fact that if the policy is local, then for all traces $\alpha$, 
we have $\alpha \models {u \interferes v}$ iff  $\alpha \models D_{\{u,v\}}(u \interferes v)$. 
This claim follows straightforwardly from the properties of distributed knowledge.  
\end{proof}

Indeed, locality  \emph{exactly} captures the condition under which the permissive and the prohibitive interpretations of a policy are equivalent. 
We need one technical condition on this statement. 
Say that a domain $u$ is \emph{inactive} when $\dom^{-1}(u) =\emptyset$, i.e., the domain has no actions that it can perform, 
so it plays only the role of an observer in the system. 
Note that policy edges of the form $\alpha \models u\interferes v$
where 
$u$ is inactive 
have no bearing on the definitions of $\maytaname$-security or unwinding-security (and hence also not on the equivalent 
$\musttaname$-security) since all references to the policy in these definitions occur
only in the form $\alpha \models u\interferes v$ with $u = \dom(a)$ for some action $a$. 
Intuitively, the interference relation $u\interferes v$ is  concerned with the ability of domain  $u$ to perform actions that 
transfer information to domain $v$, and for a domain with no actions, such transfer of information is always impossible. 
We say that policy $\interferes$ \emph{has no edges from 
inactive domains}, when for all domains $u$ that are inactive, 
we have 
$\alpha \models u \noninterferes v$ for all traces $\alpha$ and domains $v$. 
By the above observation, every policy is semantically equivalent to one that has no edges from 
inactive
domains, 
so we may assume this condition without loss of generality. Subject to this assumption, we have the following: 

\begin{theorem} \label{thm:must-may-local} 
Let $\mathcal{S} = \tuple{\Dom, \dom, \Actions}$ be a signature, and let $\interferes$ be a 
policy for this signature that has no edges from 
inactive
domains.  Suppose that for all systems $M$ for the signature $\mathcal{S}$, 
$M$ is $\maytaname$-secure with respect to $\interferes$ iff 
$M$ is $\musttaname$-secure with respect to $\interferes$. Then $\interferes$ is local. 
\end{theorem} 

\begin{proof} 
We prove the converse, that is, we show that if $\interferes$ is not local, 
then $\maytaname$-security and $\musttaname$-security with respect to 
$\interferes$ differ on some system $M$. Suppose that $\interferes$ is not local. 
By Lemma~\ref{lem:locality-equiv}, we may use the relations $\{\sim^\Diamond_u\}_{u\in D}$, 
so we have that there exist traces $\alpha, \beta\in A^*$ and domains $u,v$ 
such that $\alpha \sim^\Diamond_u \beta$ and $\alpha \sim^\Diamond_v \beta$, 
and $\alpha \models u \interferes v$ and $\beta \models u \noninterferes v$. 

By Lemma~\ref{mayta-cont-unwind}, we also have $\alpha \unw_u\beta$ and $\alpha \unw_v\beta$. 
Since $\interferes$ has no edges from 
inactive
domains, and $\alpha \models u \interferes v$, 
domain $u$ has actions.  Let $a$ be any action with $\dom(a) = u$.   Then, by condition WSC, it follows that $\alpha a \unw_v \beta a $. 
Moreover, since $\alpha \sim^\Diamond_v \beta$ and $\beta \sim^\Diamond_v \beta a$ (because $\beta \models u \noninterferes v$), 
we have that $\alpha \sim^\Diamond_v \beta a$. By $\alpha \models u \interferes v$ we have that 
$\alpha a \not \sim^\Diamond_v \alpha$, so it follows that $\alpha a \not \sim^\Diamond_v \beta a$. 

Define $M$ to be the system with states $S = A^*$, 
initial state $s_0 = \epsilon$, transitions defined by $\gamma \cdot b =  \gamma b$ for all 
$\gamma \in A^*$ and $b\in A$, and observations defined by 
$\obs_w(\gamma) = 0$ if $w\neq v$ or $\gamma \not \sim^\Diamond_v \alpha a$, 
and $\obs_w(\gamma) = 1$ otherwise, for all $\gamma \in A^*$. 
Obviously, $\obs_v(\alpha a) = 1$. On the other hand, since $\alpha a \not \sim^\Diamond_v \beta a$, 
we have that $\obs_v(\beta a) = 0$. 

Since  $\alpha a \unw_v \beta a $, it is immediate that $M$ is not 
$\musttaname$-secure with respect to $\interferes$. 
To complete the proof, we note that 
$M$ is $\maytaname$-secure with respect to $\interferes$, 
so that the two notions disagree on system $M$. 
To see this, note that any insecurity must involve domain $v$, since 
$\obs_w(\gamma) = 0$ for all $\gamma$ and $w\neq v$. 
However, on domain $v$, if $\gamma \sim^\Diamond_v \gamma'$ then 
$\gamma \not \sim^\Diamond_v \alpha a$ iff $\gamma' \not \sim^\Diamond_v \alpha a$, 
so we have  $\obs_v(\gamma) = \obs_v(\gamma')$, as required for $\maytaname$-security.  
\end{proof}

Given this result, one reasonable approach to the  possibility of conflicting interpretations
of policies is to require that the policy be local, so that we are left with a single 
notion of security that supports both the permissive and prohibitive interpretation of
policies. A potential disadvantage of this is that such a restriction results in a loss of 
expressiveness: certain policies can no longer be expressed. The following result shows
that, in fact, provided that one is interested in $\musttaname$-security, there
is no loss of expressiveness: for every policy, there is a local policy that is
equivalent with respect to $\musttaname$-security.

\begin{lemma}
\label{lem:restrict-to-local}
 For every policy $\interferes$ over signature $\langle \Dom, \dom, \Actions\rangle$, there is a 
 local policy ${\interferes'} \leq {\interferes}$ such that 
 for all systems $M$ with signature $\langle \Dom, \dom, \Actions\rangle$, 
 system $M$ is  $\musttaname$-secure \wrt $\interferes$ iff $M$ is  $\musttaname$-secure \wrt $\interferes'$. 
\end{lemma}

\begin{proof}
We write $[\alpha]_{\unwind}$ for the equivalence class of $\alpha$ with respect to an equivalence relation $\unwind$. 
In the following, the unwinding $\{ \unw_u\}_{u \in \Dom}$ refers to the policy $\interferes$.  
Define the policy $\interferes'$ by
\begin{align*}
 \alpha \models {u \interferes' v} \text{ iff } \beta \models {u \interferes v} \text{ for all } \beta \in [\alpha]_{\unw_u} \cap [\alpha]_{\unw_v} 
 \enspace. 
\end{align*}
It is clear that ${\interferes'} \leq {\interferes}$. 

Let $\{\unwprime_u\}_{u\in D}$ be the smallest unwinding satisfying WSC and DLR \wrt the policy $\interferes'$. 
We claim that ${\unw_u}  = {\unwprime_u}$ for all $u\in D$. 
By Lemma~\ref{lem:unwmon}, we have ${\unwprime_u} \subseteq {\unw_u}$ for all $u\in \Dom$.

We will show that for every $u \in \Dom$ and every $\alpha, \beta \in \Actions^*$, we have $\alpha \unwprime_u \beta$ implies $\alpha \unw_u \beta$. 
The proof is by an induction on the length of the derivation of $\alpha \unwprime_u \beta$. 
Consider a derivation of   $\alpha \unwprime_u \beta$ and suppose that the claim holds for all 
shorter derivations.  We need to show that $\alpha \unw_u \beta$. The cases
where the final step of the derivation are an application of reflexivity, symmetry or transitivity
are straightforward from the fact that $\unw_u$ is an equivalence relation. 
There are two remaining cases: 

\begin{itemize} 
\item \emph{Case 1:} $\alpha \unwprime_u \beta$ is derived using the condition DLR.
Here we have that there exists $a\in A$ such that  $\alpha = \beta a$ and $\beta \models {\dom(a) \noninterferes u}$. 
Thus, there exists a trace $\gamma$ such that $\gamma \unw_u \beta$ and $\gamma \unw_{\dom(a)} \beta$ such that $\beta \models {\dom(a) \noninterferes u}$. 
By Proposition~\ref{prop:LRprime}, we have that $\unw_u$ satisfies DLR$'$. It follows that $\alpha \unw_u \beta$.  

\item 
\emph{Case 2:} $\alpha\unwprime_u \beta$ is derived using condition WSC.
Hence we have $\alpha = \alpha' a$ and $\beta = \beta' a$ with $\alpha' \unwprime_u \beta'$ and $\alpha' \unwprime_{\dom(a)} \beta'$. 
By applying the induction hypothesis, we have $\alpha' \unw_u \beta'$ and $\alpha' \unw_{\dom(a)} \beta'$ and by the WSC condition applied to $\unw_u$, we have $\alpha' a \unw_u \beta' a$. 
\end{itemize} 
This shows that the system $M$ is $\musttaname$-secure \wrt $\interferes$ if and only if it is $\musttaname$-secure \wrt $\interferes'$. 

It remains to show that the policy $\interferes'$ is $\maytaname$-local. 
First, we will show for every $u \in \Dom$ and every $\alpha, \beta \in \Actions^*$, that 
$\mayta \alpha u = \mayta \beta u$ (\wrt $\interferes'$) implies $\alpha \unwprime_u \beta$. 
We prove this by an induction on the combined length of $\alpha$ and $\beta$. 
The base case of $\alpha = \beta = \epsilon$ is trivial by reflexivity of $\unwprime_u$. 

Assume the claim for all traces of combined length shorter than $|\alpha|+|\beta|$. 
Let $\alpha = \alpha' a$ and suppose that $\mayta \alpha u = \mayta\beta u$. 
We need to show that $\alpha' a \unwprime_u \beta$. 
\begin{itemize}
\item  
\emph{Case 1:} $\alpha' \models \dom(a) \noninterferes' u$. 
In this case, we have $\mayta {\alpha' a} u = \mayta {\alpha'} u$, so $\mayta {\alpha'} u = \mayta {\beta} u$. 
By induction hypothesis, we have $\alpha' \unwprime_u \beta$.  By DLR, we have $\alpha' a \unwprime_u \alpha'$, 
so by transitivity we obtain $\alpha' a  \unwprime_u \beta$. 
\item 
\emph{Case 2:} $\alpha' \models {\dom(a) \interferes' u}$. 
Then $\mayta {\alpha' a} u = (\mayta {\alpha'} u , \mayta {\alpha'} {\dom(a)}, a)$, 
so $\beta \neq \epsilon$ and we have 
 $\beta = \beta' b$ for some $b\in A$. 
We may assume that $\beta' \models {\dom(b) \interferes' u}$, 
since otherwise we may argue as in Case 1 for $\beta' b$ and apply symmetry. 
It follows that $\mayta {\beta' b} u = (\mayta {\beta'} u , \mayta {\beta'} {\dom(b)}, b)$, 
and we conclude that 
$a= b$ and $\mayta {\alpha'} u = \mayta {\beta'} u$ and $\mayta {\alpha'} {\dom(a)} = \mayta {\beta'} {\dom(a)}$.
By induction hypothesis we have $\alpha' \unwprime_u \beta'$ and $\alpha' \unwprime_{\dom(a)} \beta'$ and by WSC,  we obtain that $\alpha' a \unwprime_u \beta' a$. 
\end{itemize} 
The $\maytaname$-locality follows: if $\mayta \alpha u = \mayta \beta u$ and $\mayta \alpha v = \mayta \beta v$ (\wrt $\interferes'$), then we have
$\alpha \in [\beta]_{\unwprime_u} \cap [\beta]_{\unwprime_v}$ and hence $\alpha \models {u \interferes' v}$ iff $\beta \models {u \interferes' v}$. 
\end{proof}

We say that two policies $\interferes$ and $\interferes'$ for a common signature $\langle D, \dom,A\rangle$
are \emph{identical up to 
inactive
domains}, if for all domains $u,v\in D$ such that the  set of actions $\dom^{-1}(u)$ of domain $u$ is nonempty, 
we have, for all $\alpha\in A^*$, that $\alpha \models {u \interferes v}$ iff $\alpha\models {u \interferes' v}$. 
Intuitively, this says that the two policies are identical, except that they may differ on edges $\alpha\models {u \interferes' v}$
from 
inactive domains $u$. 
Since  
such a domain $u$ cannot take any action to cause information flow
to $v$, this difference cannot, in practice, be detected. The following result makes this intuition precise.  
Say that two policies  $\interferes$ and $\interferes'$ are two policies for a common signature are \emph{equivalent} with respect to a
notion of security if for all systems $M$ for the signature, $M$ is secure with respect to $\interferes$ iff $M$ is secure with respect to $\interferes'$. 

\begin{theorem} 
Suppose that $\interferes$ and $\interferes'$ are two policies for a common signature $\langle D, \dom,A\rangle$. 
Then $\interferes$ and $\interferes'$ are equivalent with respect to $\maytaname$-security iff
$\interferes$ and $\interferes'$ are identical up to 
inactive 
domains. 
\end{theorem} 

\begin{proof} 
In this proof, $\taname^{\Diamond,\interferes}$ and $\taname^{\Diamond,\interferes'}$ denote the $\taname^{\Diamond}$ function that refers to $\interferes$ and $\interferes'$, respectively. 
Suppose first that $\interferes$ and $\interferes'$ are identical up to 
inactive 
domains. 
Then a straightforward induction on the length of $\alpha\in A^*$ shows that for all $u\in D$, we have 
$\taname^{\Diamond,\interferes}_u (\alpha) = \taname^{\Diamond,\interferes'}_u (\alpha)$. 
(Note that in the inductive case for $\mayta {\alpha a} u$, we have that $a \in \dom^{-1}( \dom(a)) \neq \emptyset$, 
so $\dom(a) \interferes u$ iff $\dom(a) \interferes' u$.) 
It follows that $\interferes$ and $\interferes'$ are equivalent with respect to $\maytaname$-security. 

Conversely, suppose that $\interferes$ and $\interferes'$ are not identical up to 
inactive 
domains. 
We show that they are not equivalent with respect to $\maytaname$-security.  
We have that there exists $\alpha \in A^*$ and domains $u,v$ with $\dom^{-1}(u)$ nonempty, 
such that (without loss of generality) $\alpha \models {u \interferes v}$  but not $\alpha\models {u \interferes' v}$. 
Let $a$ be an action with $\dom(a) = u$. Then 
$\taname^{\Diamond,\interferes}_v (\alpha a ) \neq  \taname^{\Diamond,\interferes}_v (\alpha)$
but 
$\taname^{\Diamond,\interferes'}_v (\alpha a ) =  \taname^{\Diamond,\interferes'}_v (\alpha)$. 
Define the system $M$ with states $S= A^*$, transitions $\beta\cdot b = \beta b$ for all $\beta\in A^*$ and $b\in A$, 
and observations given by 
$\obs_w(\beta) = 0$ if either $w \neq v$ or $\taname^{\Diamond,\interferes}_w (\beta )\neq  \taname^{\Diamond,\interferes}_w (\alpha)$, 
otherwise $\obs_w(\beta) = 1$. 
Then, by construction, $M$ is $ \maytaname$-secure with respect to $\interferes$. 
But we have $\obs_v(\alpha a) = 0$ and $\obs_v(\alpha) = 1$,  so 
$M$ is  not $ \maytaname$-secure with respect to $\interferes'$. 
Thus, these policies disagree on $M$. 
\end{proof}

In particular, it follows that if $\interferes'$ is a nonlocal policy that has no edges from 
inactive 
domains 
(so that the reason for nonlocality is not trivial, in the sense that it must involve edges from a
domain with a nonempty set of actions), then there is no local policy that is equivalent 
with respect to $\maytaname$-security.  Thus, $\maytaname$-security is inherently a 
more expressive notion than $\musttaname$-security.

\subsection{Special Cases}
\label{sec:specialcases}

 To judge the adequacy of a definition of security with respect to 
dynamic policies, it is useful to 
verify that the definition gives the expected answer in 
 a number of simple scenarios where
we have clear intuitions about the desired behaviour of the definition. 
The following provide some examples. 

\subsubsection{Static Policies:} 

We say that a 
dynamic 
policy 
$\interferes \subseteq \Dom\times A^*\times \Dom$ 
is \emph{static}, if for all $\alpha\in A^*$ and $u,v\in \Dom$, we have 
$\alpha \models {u \interferes v}$ iff $\epsilon \models {u \interferes v}$. 
That is, the policy never changes from its initial state. 
A static policy $\interferes$ may equivalently be presented as the 
relation %${\dyninterferes{\epsilon}} \subseteq \Dom\times \Dom$, 
$\pi_{1,3}(\interferes) \subseteq \Dom \times \Dom$, the projection 
of the dynamic policy to the first and third components, 
which is 
in 
precisely the format for classical intransitive noninterference 
policies \cite{Goguen1982}. 

Static policies are a special case of dynamic policies, that have been subject to 
a significant amount of study \cite{HY87,Roscoe1999,rushby92,Meyden15}, so it is reasonable to expect that 
any definition for dynamic policies should reduce to an accepted definition 
when applied to the static case. The following result shows that this is the case. 

\begin{theorem} 
If $\interferes$ is a static policy, then the following 
are equivalent: 
\begin{itemize} 
\item $M$ is $\maytaname$-secure with respect to $\interferes$,  
\item $M$ is $\musttaname$-secure with respect to $\interferes$,  
\item $M$ is TA-secure with respect to $\pi_{1,3}(\interferes)$. 
\end{itemize}  
\end{theorem} 

\begin{proof} 
This is immediate from the definitions, together with the fact that for static policies $\interferes$, we have 
$\alpha \models u \interferes v$ iff $\alpha \models D_{\{u,v\}} (u\interferes v) $ iff 
$(u,v) \in \pi_{1,3}(\interferes)$. 
\end{proof} 

This result shows that $\maytaname$-security and $\musttaname$-security agree and 
collapse to the standard notion of $\taname$-security when restricted
to static policies. This provides support for these definitions. 

\subsubsection{Globally known policies:}

Another setting 
that can serve as a useful test case for the adequacy of 
definitions 
of security for dynamic policies 
is a group of agents subject to an information flow policy that is set by a 
policy authority $P$, and in which all agents are permitted to know the 
policy 
state
at all times. 

Formally, these are dynamic policies in which the set of domains contains $P$, 
and we have the following: 
\begin{itemize} 
\item[(GK1)] 
For all traces $\alpha$ and domains $u$,  we have that 
$\alpha \models {P \interferes u}$. That is, domain $P$, which controls
the policy state, is always permitted to transmit information to $u$. 
This ensures that $u$ is always able to know the policy state. 
 
\item[(GK2)] For all traces $\alpha,\beta$ and domains $u,v$,  we have that 
$\alpha |P = \beta |P$ implies that 
$\alpha \models {u \interferes v}$ iff  $\beta \models {u \interferes v}$.  That is, the policy setting depends only on the 
past actions of domain $P$. 
\end{itemize} 
We say that a dynamic policy $\interferes$ satisfying these conditions 
is {\em globally known}. 

\begin{theorem} 
\label{thm:globally-known-implies-local}
If the policy $\interferes$ is globally known, then it is local. 
\end{theorem} 

\begin{proof} 
Suppose that the policy $\interferes$ is globally known. 
We need to show that for all $u,v\in D$ and $\alpha, \beta \in A^*$, 
we have $\alpha \unw_u\beta$ and $\alpha  \unw_v\beta$ implies 
$\alpha \models {u \interferes v}$ iff  $\beta \models {u \interferes v}$.  
We claim that, in fact, $\alpha  \unw_u\beta$ implies $\alpha |P = \beta |P$, 
from which the desired result follows using GK2. 

The proof of the claim is by induction on the length of the derivation of $\alpha  \unw_u\beta$. 
The base case of a step in which $\alpha  \unw_u\beta$ is reflexivity by reflexivity is trivial, 
and the cases of steps using symmetry or transitivity are also straightforward. 

Consider a step using DLR. Then $\alpha = \beta a$ and we have $\alpha \models {\dom(a) \interferes u}$. 
By GK1, we have $\dom(a) \neq P$. Thus $\alpha |P = \alpha a |P = \beta | P$, as required. 

Finally, consider a step using WSC. Then $\alpha = \alpha' a$ and $\beta = \beta' a$ for some
action $a$, and we have $\alpha' \unw_u \beta'$ and $\alpha' \unw_{\dom(a)} \beta'$. 
\From  $\alpha' \unw_u \beta'$, we obtain by induction that $\alpha |P = \beta|P$. 
It follows that $\alpha|P = \alpha' a |P = \beta' a|P = \beta |P$, as required. 
\end{proof}

Since by Theorem~\ref{thm:local-must-may}, for local policies the definitions of 
$\maytaname$-security and  $\musttaname$-security are
equivalent, these definitions are equivalent for all globally known policies, 
so these definitions lend mutual support to each other in this case.

\subsubsection{Locally known policies:} 

A slightly more general case in which the notions of 
$\maytaname$-security and  $\musttaname$-security coincide
are  policies in which agents
are always aware of their incoming and/or outgoing policy edges. 
Formally, define say that the policy $\interferes$ is 
\emph{locally known to the sender} (with respect to equivalence relations $\{\unw_u\}_{u \in \Dom} $) 
if we have, for all $u, v\in \Dom$, that  
$\alpha \unw_u \beta$ implies 
$\alpha \models {u \interferes v}$  iff $\beta \models {u \interferes v}$. 
Similarly, say that 
the policy $\interferes$ is 
\emph{locally known to the receiver} if we have, for all $u, v\in \Dom$, that  
$\alpha \unw_u \beta$ implies 
$\alpha \models {v \interferes u}$  iff $\beta \models {v \interferes u}$. 
It is easily seen that if the policy is locally known to the sender and/or the receiver
(with respect to either $\sim^{\Box}$ or $\sim^{\Diamond}$), 
then it is local. Thus, this again is  a case where the 
notions of $\maytaname$-security  and $\musttaname$-security are equivalent.

\subsection{Related Work using Automaton Semantics} \label{sec:related:aut}

To make the case for our definitions above, we now compare 
them 
to related 
prior work that is also stated in an automaton based semantics. 
(We defer discussion of related programming language-based work to Section~\ref{sec:relatedwork}.)  
To the best of our knowledge, the only other 
semantics for dynamic intransitive noninterference 
policies in 
automaton-based systems are those of Leslie~\cite{Leslie2006} and Eggert et. al.~\cite{Eggert2013}. 
Both works deal with policy enhanced systems. 
Recall that in this model, 
we write $s \models {u \interferes v}$ iff information flow from $u$ to $v$ is allowed in the state~$s$.

Leslie's definition of intransitive noninterference is a direct adaption of Rushby's ipurge function~\cite{rushby92} to the dynamic setting. 
Similar to Rushby's ipurge function, 
Leslie's purge function is based on a ``source" function. 
Formally, 
the dynamic version of this function is defined 
for every $u \in \Dom$ and  $s \in \States$
by $\dsrc{\epsilon}{u}{s} = \{u \}$
and for $\alpha \in \Actions^*$ and~$a \in \Actions$ by 
$ \dsrc{a \alpha}{u}{s} = \dsrc{\alpha}{u}{s \cdot a} \cup \{ \dom(a) \} $
 if there is some $v \in \dsrc{\alpha}{u}{s \cdot a}$ with $s \models {\dom(a) \interferes v}$, 
and else
$%\begin{align*}
  \dsrc{a \alpha}{u}{s} = \dsrc{\alpha}{u}{s \cdot a}%\enspace. 
$. %\end{align*}
%
% \begin{align*}
%  \dsrc{\epsilon}{u}{s} & = \{u \} \\
%  \dsrc{a \alpha}{u}{s} & = 
%  \begin{cases}
%   \dsrc{\alpha}{u}{s \cdot a} \cup \{ \dom(a) \} & \text{if }\dom(a)  \\
%   & \hfill \dyninterferes{s} \dsrc{\alpha}{u}{s \cdot a} \\
%   \dsrc{\alpha}{u}{s \cdot a} & \text{otherwise} 
%   \enspace. 
%  \end{cases}
% \end{align*}
% 
%ron: I don't get this intuition, it doesn't even mention u! 
Intuitively, the set $\dsrc{a\alpha}{u}{s}$ is the set of all agents, to which the information that the action $a$ has been performed, is transmitted by the sequence of actions $\alpha$. 

Leslie's purge function for intransitive noninterference is inductively defined 
for every $u \in \Dom$, $s \in \States$, $a \in \Actions$, and~$\alpha \in \Actions^*$ by
$\lesliepurge \epsilon u s  = \epsilon$ and 
\begin{align*}
 \lesliepurge {a \alpha} u s =  a \; \lesliepurge \alpha u {s \cdot a} 
\end{align*}
 if there is some $v \in \dsrc{a \alpha} u s$ with $s \models {\dom(a) \interferes v}$, 
and else
$%\begin{align*}
 \lesliepurge {a \alpha} u s = \lesliepurge \alpha u {s \cdot a}% \enspace. 
$. %\end{align*}
% 
% 
% \begin{align*}
%   \lesliepurge \epsilon u s & = \epsilon \\
%   \lesliepurge {a \alpha} u s & = 
%   \begin{cases}
%    a \; \lesliepurge \alpha u {s \cdot a} & \text{if } \dom(a) \in \dsrc{a \alpha} u s \\
%    \lesliepurge \alpha u {s \cdot a} & \text{otherwise} \enspace. 
%   \end{cases}
% \end{align*}
%
This function removes exactly those actions from a trace that are from agents where interference to $u$ is forbidden by the policy. 
% However, note, that even when an action $a$ is removed from trace by applying the $\lpurgename$ operator, the state in third argument changes from $s$ to $s \cdot a$. 
% Therefore, some information that the action $a$ has been performed is kept and can possibly affect further applications of the function to the remaining trace. 
% [[ Fixme: More of a story needed to justify our current definitions. maybe add a reference here to your paper for details of an example where this gives undesirable results??]] 
%
The corresponding security definition is:
\begin{definition}
 A system $M$ is $\lpurgename$-secure with respect to a policy $\interferes$ if for all domains $u$ and sequences of actions $\alpha$ holds:  $\obs_u(s_0 \cdot \lesliepurge \alpha u {s_0}) = \obs_u(s_0 \cdot \alpha)$. 
\end{definition}

Our notion of $\maytaname$-security is incomparable to $\lpurgename$-security, in the sense that neither of them implies the other. 
One the one hand, on systems with a static policy, $\lpurgename$-security is equivalent to Rushby's IP-security and hence,  
as shown in~\cite{Meyden15}, strictly weaker than \tasecurity, which is the equivalent to $\maytaname$-security on such systems. 
On the other hand, the following example presents a $\maytaname$-secure, but not $\lpurgename$-secure system. 
% 
% [[ I don't get how this shows incomparability. Possibly still $\taname$-security implies  $\lpurgename$-security in general. Why is this the only 
% appropriate comparator, how about $\maytaname$-security? Perhaps also this section should talk about locality of the examples. 
% Local distinguishers would be the best, since that works for both of our definitions. For the following example, note that 
% Example 2 also talks about $\maytaname$-security of Figure 1, so probably this argumentation belongs there if we include it. 
% 
% It would be good to see a diagram showing relationships between the definitions. We don't even say anywhere that 
% $\maytaname$-secure implies $\musttaname$-secure, though I think the essence of a proof is in the suppressed results. 
% ]]
% 

\vspace*{-10pt} 
\begin{example}
We consider again the system in Figure~\ref{fig:conflict}. 
As shown in Example~\ref{ex:conflict-ta}, this system is $\maytaname$-secure. 
% 
% In the system in Figure~\ref{fig:conflict}, agent $B$ can only observe $1$ if agent $P$ has performed an action $p$ and $A$ an action $a$ afterwards. 
% Hence all traces leading to $1$ are of the form $a^*pp^*a(p+a)^*$ and all traces leading to the observation $0$ are of the form $a^*pp^*$. 
% Agent $B$'s $\maytaname$ values of all these traces that lead to $0$ are $\epsilon$, and $B$'s $\maytaname$ values of all traces that lead to $1$ are different from $\epsilon$, since they contain a $a$ their right-most component. 
% Therefore, there are no two traces with the same $\maytaname$ values leading to different observations. 
%
However the system is not $\lpurgename$-secure, since we have
%ron2: \begin{multline*}
$ \obs_B(s_0 \cdot \lesliepurge{pa}{B}{s_0}) 
   = \obs_B(s_0 \cdot \lesliepurge{a}{B}{s_1}) %\\
   = \obs_B(s_0 \cdot a)
   = 0 
   \neq 1 
   = \obs_B(s_0 \cdot pa)
%   \enspace. 
%\end{multline*}
$. 
\end{example}

\vspace*{-10pt} 

% For proving this result, we generalize the ${\tt Lpurge}$ function to arbitrary  set of agents $U$ instead of a single agent $u$. 
% This set $U$ is the same set of domains as computed by the  ${\tt dsrc}$ function. 
% For every $U \subseteq \Dom$, $s \in \States$, $a \in \Actions$ and $\alpha \in \Actions^*$ define
% \begin{align*}
%   \lesliepurge \epsilon U s & = \epsilon \\
%   \lesliepurge {\alpha a} u s & = 
%   \begin{cases}
%    \lesliepurge \alpha {U \cup \{\dom(a)\}} {s} & \text{if } \dom(a) \dyninterferes{s \cdot \alpha} v \text{ for some }v \in \Dom \\
%    \lesliepurge \alpha U {s} & \text{otherwise} \enspace. 
%   \end{cases}
% \end{align*}
% 
% Clearly, we have $\lesliepurge \alpha u s = \lesliepurge \alpha {\{u\}} s$. 

Similar to $\maytaname$-security, ${\tt Lpurge}$-security is not monotonic with respect to the restrictiveness order on policies.

 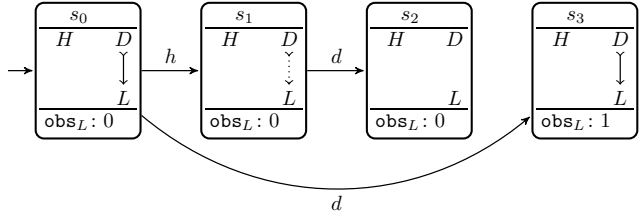
\begin{figure}%[htp]
 \centering
 \scalebox{0.78}{
 \begin{tikzpicture}[tikzglobal,node distance=1cm]
 
    \newcommand{\agents}{ 
      % defines the agents and their positions in the policy
      \node[agent] (H) [left] {$H$};
      \node[agent] (D) [right of=H] {$D$};
      \node[agent] (L) [below of=D] {$L$};
      }

 \node[initial,systemstate] (s0) {
   $s_0$ 
   \\
   \hline
   \agents
   \\
   \hline
   $\obs_L \colon 0$ \\
  };
   \path[policy] (D) edge (L);
   
 \node[systemstate] (s1) [right=of s0] {
 $s_1$ 
 \\
 \hline
  \agents
  \\
   \hline	
   $\obs_L \colon 0$ \\
 };
 \path[policy,dotted] (D) edge (L);
   
 \node[systemstate] (s2) [right=of s1] {
 $s_2$
 \\
 \hline
  \agents
  \\
   \hline
   $\obs_L \colon 0$ \\
 };

  \node[systemstate] (s3) [right=of s2] {
 $s_3$
 \\
 \hline
  \agents
  \\
   \hline
   $\obs_L \colon 1$ \\
 };
   \path[policy] (D) edge (L);

 \path (s0) edge node {$h$} (s1) 
 (s0) edge [bend right=40] node[below] {$d$} (s3)
 (s1)  edge node {$d$} (s2)
 ;
\end{tikzpicture}
}
\caption{Non-monotony of \lpurgename-security}
\label{fig:unintuitive-property-leslie-security}
\end{figure}

\begin{example}
\label{example:counterintuitive-leslie}
The system in Figure~\ref{fig:unintuitive-property-leslie-security} has two different configurations of the dynamic policy indicated by the dotted policy edge in state $s_1$: in one configuration there is an edge from $D$ to $L$ and in the other, this edge is 
absent. 
First, consider the configuration where this edge exists. 
With respect to \lpurgename-security, this system is insecure, since we have 
% \begin{align*}
$\lesliepurge{hd}{L}{s_0} = \lesliepurge{d}{L}{s_1} = d$, 
% \end{align*}
but $\obs_L(s_0 \cdot hd) = 0 \neq 1 = \obs_L(s_0 \cdot d)$. 

However, if we remove the edge from $D$ to $L$ in state $s_1$, this system becomes \lpurgename-secure: 
all traces starting with $h$ will be purged to $\epsilon$ and all traces starting with $d$ have a purge value different from $\epsilon$. 
% We have $\lesliepurge{hd}{L}{s_0} = \epsilon$ and $\obs_L(s_0 \cdot hd) = 0 = \obs_L(s_0)$. 
% 
% This is highly counterintuitive since the policy without this edge is more restrictive than the policy with this edge. 
\end{example}

Our notion of $\musttaname$ is strictly stronger than $\lpurgename$-security. 
That these two are indeed different follows from the same argument as above for systems with static policies. 

% [[ is there a justification of the strictly part? ]] 

\begin{lemma} \label{must-imp-Lpurge} 
 Every $\musttaname$-secure system is ${\tt Lpurge}$-secure.
\end{lemma}

\begin{proof*} 
 We show this lemma by contraposition. 
 Suppose that the system is not $\lpurgename$-secure. 
 Then there are $u \in \Dom$ and $\alpha \in \Actions^*$ with 
 \begin{align*}
  \obs_u(s_0 \cdot \alpha) \neq \obs_u(s_0 \cdot \lesliepurge \alpha u {s_0})
  \enspace. 
 \end{align*}
	We will show the following claim: 
	For every $\beta, \beta' \in \Actions^*$ with $\alpha = \beta \beta'$ and for every $\gamma \in \Actions^*$ with $\lesliepurge \alpha u {s_0} = \gamma \; \lesliepurge{\beta'} u {s_0 \cdot \beta}$ and for every $v \in \dsrc{\beta'} u {s_0 \cdot \beta}$, we have $\beta \unw_v \gamma$. 
	
	We prove this claim by an induction on $\beta$. 
	For the base case with $\beta = \epsilon$, we have that $\alpha = \beta'$ and hence $\gamma = \epsilon$ from which the claim is immediate. 
	For the inductive step, let $\beta = \tilde \beta b$ for some action $b$. 
	The induction hypothesis is: For some $\gamma$ with $\lesliepurge{\alpha} u {s_0} = \gamma \; \lesliepurge{b \beta'} u {s_0 \cdot \tilde \beta}$ and every $v \in \dsrc{b \beta'} u {s_0 \cdot \tilde \beta}$, we have $\tilde\beta \unw_v \gamma$. 

	We consider the following two cases: 
	\begin{itemize}
	 \item 
	\emph{Case 1: } $\dom(a) \notin \dsrc{b \beta'} u {s_0 \cdot \tilde \beta}$. 
		In this case, we have
	\begin{align*}
	\dsrc{b \beta'} u {s_0 \cdot \tilde \beta} = \dsrc{\beta'} u {s_0 \cdot \tilde \beta b} 
	\end{align*}
	and 
	\begin{align*}
	\lesliepurge{b \beta'}{u}{s_0 \cdot \tilde \beta} = \lesliepurge{\beta'}{u}{s_0 \cdot \tilde \beta b}
	\enspace. 
	\end{align*}
	Hence, the new value of $\gamma$ is the same as the previous value.
	And since ${s_0 \cdot  \tilde \beta} \models {\dom(b) \interferes v}$, by DLR, for every $v \in \dsrc{\beta'} u {s_0 \cdot \tilde \beta b}$, we have $ \tilde \beta b \unw_v  \tilde \beta$, and combined with $ \tilde \beta \unw_v  \gamma$ 
	the claim follows. 
	\item
	\emph{Case 2: } $\dom(a) \in \dsrc{b \beta'} u {s_0 \cdot \tilde \beta}$. 
	In this case, we have $\dsrc{b \beta'} u {s_0 \cdot \tilde \beta} = \{\dom(b) \} \cup \dsrc{\beta'} u {s_0 \cdot \tilde \beta b}$. 
	Since 
	\begin{align*}
	\lesliepurge{\alpha}{u }{s_0} & = \gamma \; \lesliepurge{b \beta'}{u }{s_0 \cdot \tilde \beta} \\
															& = \gamma b \; \lesliepurge{\beta'} u {s_0 \cdot \tilde \beta b} 
															\enspace, 
	\end{align*}
	the new value of $\gamma$ is $\gamma' =\gamma b$.
	For every $v \in \dsrc{\beta'} u {s_0 \cdot \tilde \beta b}$, we have $\tilde \beta \unw_v \gamma$ by induction hypothesis
	 and additionally $ \tilde \beta \unw_{\dom(b)}  \gamma$. 
	From the condition WSC it follows for every such $v$: $ \tilde \beta b \unw_v  \gamma b = \gamma'$. 
	From this claim it follows $ \alpha \unw_u  \; \lesliepurge{\alpha}{u}{s_0}$ with $\beta = \alpha$ and hence the system is \emph{not} $\musttaname$-secure. \qedhere
		\end{itemize}
\end{proof*}
In \cite{Eggert2013}, some problems with $\lpurgename$-security are identified and a new purge-based security definition is proposed in response.  
This notion uses
the same dynamic sources function as Leslie's does, but the purge function does not change in the state parameter if an action has been removed. 
More precisely, their dynamic intransitive purge function is defined for every $u \in \Dom$ and $s \in \States$ by
$\dipurge{\epsilon}{u}{s}  = \epsilon$
and for every $a \in \Actions$ and $\alpha \in \Actions^*$ if $\dom(a) \in \dsrc{a \alpha}{u}{s}$, 
then
$%\begin{align*}
 \dipurge{a \alpha}{u}{s} = a \; \dipurge{\alpha}{u}{s\cdot a}
$ %\end{align*}
and else
$%\begin{align*}
   \dipurge{a \alpha}{u}{s} = \dipurge{\alpha}{u}{s} % \enspace. 
$. %\end{align*}
% 
% 
% $a \in \Actions$, $\alpha \in \Actions^*$ and $s \in \States$ by
%  \begin{align*}
%   \dipurge{\epsilon}{u}{s} & = \epsilon \\
%   \dipurge{a \alpha}{u}{s} & = 
%   \begin{cases}
%    a \; \dipurge{\alpha}{u}{s\cdot a} & \text{if } \dom(a) \in \dsrc{a \alpha}{u}{s} \\
%    \dipurge{\alpha}{u}{s} & \text{otherwise}
%    \enspace. 
%   \end{cases}
%  \end{align*}
%
The notion of 
\emph{i-security} is defined as usual as the property that for all agent $u$, all states $s$, and all action sequences $\alpha$ and $\beta$ with $\dipurge \alpha u s = \dipurge \beta u s$, we have $\obs_u(s \cdot \alpha) = \obs_u(s \cdot \beta)$.

This notion of security is incomparable with our security definitions of $\maytaname$-security and $\musttaname$-security.
On the one hand, we have, due to the static case, that i-security does not imply $\maytaname$-security. 	
The following example gives a system, with a local policy, that is $\musttaname$-secure and $\maytaname$-secure, but not i-secure.

 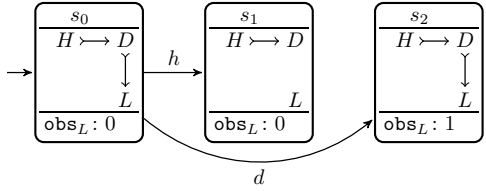
\begin{figure}%[htp]
 \centering
 \scalebox{0.8}{
 \begin{tikzpicture}[tikzglobal,node distance=1cm]
 
    \newcommand{\agents}{ 
      % defines the agents and their positions in the policy
      \node[agent] (H) [left] {$H$};
      \node[agent] (D) [right of=H] {$D$};
      \node[agent] (L) [below of=D] {$L$};
      }

 \node[initial,systemstate] (s0) {
   $s_0$ 
   \\
   \hline
   \agents
   \\
   \hline
   $\obs_L \colon 0$ \\
  };
   \path[policy] (D) edge (L);
   \path[policy] (H) edge (D);
   
 \node[systemstate] (s1) [right=of s0] {
 $s_1$ 
 \\
 \hline
  \agents
  \\
   \hline	
   $\obs_L \colon 0$ \\
 };
  \path[policy] (H) edge (D);
  
%  \node[systemstate] (s2) [right=of s1] {
%  $s_2$
%  \\
%  \hline
%   \agents
%   \\
%    \hline
%    $\obs_L \colon 0$ \\
%  };
%   \path[policy] (D) edge (L);
%   \path[policy] (H) edge (D);
   
    \node[systemstate] (s2) [right=of s1] {
 $s_2$
 \\
 \hline
  \agents
  \\
   \hline
   $\obs_L \colon 1$ \\
 };
   \path[policy] (D) edge (L);
   \path[policy] (H) edge (D);

 \path (s0) edge node {$h$} (s1) 
       (s0) edge [bend right=40] node[below] {$d$} (s2)
%        (s1) edge node {$d$} (s2)
 ;
\end{tikzpicture}
}
\caption{A $\musttaname$-secure, but \emph{not} i-secure system}
\label{fig:dtasecure-not-disecure}
\end{figure} 

\vspace*{-10pt} 
\begin{example}
\label{ex:dtasecure-not-disecure}
The system in Figure~\ref{fig:dtasecure-not-disecure} is clearly \emph{not} i-secure. 
We consider the traces $h d $ and $d$ since the $\dipurgename$ values of these two traces are the same for $L$ if one starts purging in $s_0$:
$%\begin{align*}
 \dipurge{h d}{L}{s_0} = \dipurge{d} L {s_0} = d
% \enspace. 
$. %\end{align*}
But the observations after these two traces are different for $L$. 

However, for any trace that starts with an $h$ action, the $\maytaname$ value is $\epsilon$, and for any trace that starts with a $d$ action, the $\maytaname$ value differs from $\epsilon$. 
Hence, for any two traces that lead to different observations for $L$ have different $\maytaname$ values. 
Note, that the dynamic policy of this system is local, and hence, this system is also $\musttaname$-secure. 

We would argue that this system is intuitively secure, so this gives grounds to prefer our two definitions over i-security. 
First, note that in every state $D$ is permitted to learn about $H$ actions, so, in particular, it should be permitted
for $D$ to know in state $s_0$ that $H$ has not yet acted. In state $s_2$, $L$ can deduce that $D$ acted before $H$ did, 
but $D$ was permitted to transfer its information to $L$ when it acted, so this is not an insecurity. Finally, $L$ cannot distinguish state $s_1$ from $s_0$, 
so  this state is secure on the grounds that the initial state of a system always is (noninterference definitions, intuitively, aim
to prohibit agents from learning information about what has happened in a system). 
\end{example}

\vspace*{-10pt} 
\begin{figure}%[htp]
\centering
 \begin{tikzpicture}
 \node (mustta) {$\musttaname$-security};
 \node (mayta) [right=of mustta] {$\maytaname$-security};
 \node (isecurity) [below=of mustta] {i-security};
 \node (lpurge) [right=of isecurity] {$\lpurgename$-security};
 
 %ron: I get problems with [implies] that I don't have time to diagnose now 
 %\draw (mustta) edge[implies] (mayta);
 %\draw (mustta)	 edge[implies] (lpurge);
 %\draw (isecurity) edge[implies] (lpurge);
 \draw (mustta) edge[->] (mayta);
 \draw (mustta)	 edge[->] (lpurge);
 \draw (isecurity) edge[->] (lpurge);
 \end{tikzpicture}
 \caption{Relations between dynamic intransitive noninterference definitions}
 \label{fig:relation-security-definitions}
\end{figure}
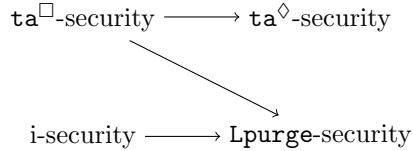

Although i-security implies $\lpurgename$-security, Example~\ref{ex:dtasecure-not-disecure} shows that these security definitions still differ on systems with local policies. 
The implications of the security definitions analyzed in this work are depicted in Figure~\ref{fig:relation-security-definitions}.
All these implications are strict in general, and the definitions of $\musttaname$-security and $\maytaname$-security only collapse on systems with local policies. 

We can summarize our reasons for preferring the definitions of the present paper to these prior works as 
follows. First, both  $\lpurgename$-security and i-security 
are 
equivalent to Rushby's IP-security
in the case of static policies. This means that these definitions are 
are subject to the problems identified by van der Meyden \cite{Meyden15}. 
One of these problems is an example presented in \cite{Meyden15}, which shows that Rushby's IP-security fails to detect a covert
channel based on ordering of actions.  A second weakness is that IP-security has a weaker relationship
to the unwinding proof method for intransitive non-interference proposed by Rushby \cite{rushby92} than does $\taname$-security. 
As shown in  \cite{Meyden15}, this proof method is sound but not complete for IP-security, 
but both sound and complete for \tasecurity. Since the proof method is itself is based on well-founded intuitions,
 this lends further support to \tasecurity. (We establish a related completeness result for
 $\musttaname$-security and $\maytaname$-security with respect to local dynamic policies in the following section, that shows
 that a similar justification exists for these extensions of \tasecurity to the dynamic case.) 
Secondly, when we consider $\lpurgename$-security and i-security  with respect to 
dynamic policies, there are   examples for each where they give intuitively incorrect results. 
We refer  to  \cite{Eggert2013} for examples showing $\lpurgename$-security yields
undesirable conclusions. Example~\ref{ex:dtasecure-not-disecure} shows that i-security 
can also yield undesirable results.

\subsection{Policy Refinement} 

Recall that ${\interferes} \leq {\interferes'}$ says that policy $\interferes$ is more restrictive than
policy $\interferes'$. We therefore expect that if a system is secure with respect to $\interferes$, 
it is also secure with respect to $\interferes'$. The following result shows that $\musttaname$-security behaves as expected
with respect to this order. 

\begin{proposition} 
\label{prop:mustta-monotonic} 
If ${\interferes} \leq {\interferes'}$  and system $M$ is $\musttaname$-secure with respect to $\interferes$, 
then $M$ is $\musttaname$-secure with respect to $\interferes'$. 
\end{proposition} 

\begin{proof} 
This follows from the equivalence of $\musttaname$-security and unwinding-security. 
Let ${\interferes} \leq {\interferes'}$ and let $\{\unw_u\}_{u\in \Dom}$ and $\{\unwprime_u\}_{u\in \Dom}$ be the 
unwinding relations for the policies $\interferes$ and $\interferes'$, respectively. 
By Lemma~\ref{lem:unwmon}, we have ${\unwprime_u} \subseteq {\unw_u}$ for all $u\in \Dom$. 
Suppose that 
$M$ is unwinding secure with respect to $\interferes$. Then $\alpha \unwprime_u \beta$ 
implies $\alpha \unw_u \beta$, which yields $\obs_u(s_0\cdot \alpha) = \obs_u(s_0\cdot \beta)$, 
as required for $M$ to be secure with respect to $\interferes'$. 
\end{proof}

The following example illustrates a further subtlety concerning 
$\maytaname$-security: unlike $\musttaname$-security, it is not monotonic with respect to 
the restrictiveness order on policies. 

\begin{example} 
Figure~\ref{fig:refinement} shows a system and a policy $\interferes$
 for two agents $A,B$. Only the observations of agent $B$ are depicted;
 for agent $A$ we assume that $\obs_A(s) = 0$ for all states $s$, 
 so that there can be no insecurity with respect to agent $A$. 
 For agent $B$, the only possible cases of $\maytaname$-insecurity are 
 when $\obs_B(s_0\cdot \alpha) = 1 $ and $\obs_B(s_0\cdot \beta)=0$. 
 This implies that we have  $\alpha = aba \alpha'$ and that $aba$ is not a prefix of 
 $\beta$. However, in this case, $\mayta \beta B$ cannot contain a subterm $a$, 
 since there are no strict prefixes $\gamma$ of $\beta$ with $\gamma \models {A \interferes B}$. 
 On the other hand, $\mayta {aba \alpha'} B$ always contains a subterm $a$, since 
$\mayta {aba}  B = ((\epsilon,\epsilon,b), (\epsilon,\epsilon,a), a)$. Thus, we cannot have 
 $\mayta {aba \alpha} B = \mayta \beta B$. It follows that the system is $\maytaname$-secure with respect to $\interferes$. 
We note that the policy $\interferes$ is not local, because 
$\mayta {ab}  A = \mayta {ba}  A = (\epsilon, \epsilon,a)$ and  $\mayta {ab}  B = \mayta {ba}  B = (\epsilon, \epsilon,b)$, 
but $ab \models {A \interferes B}$ while $ba \models {A \noninterferes B}$.

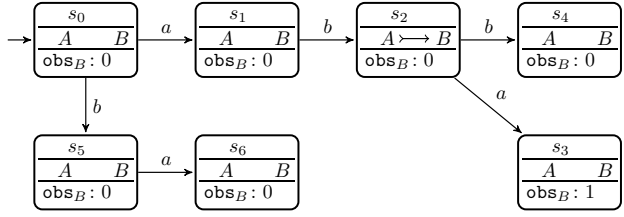
\begin{figure}%[htp]
\begin{center}
\scalebox{0.75}{
 \begin{tikzpicture}[tikzglobal,node distance=1cm]
 
    \newcommand{\agents}{ 
      \node[agent] (A) [left] {$A$};
      \node[agent] (B) [right of=A] {$B$};
      }

 \node[initial,systemstate] (s0) {
   $s_0$ 
   \\
   \hline
   \agents
   \\
   \hline
   $\obs_B \colon 0$ \\
  };
%   \path[policy] (A) edge (L);
   
 \node[systemstate] (s1) [right=of s0] {
 $s_1$ 
 \\
 \hline
  \agents
  \\
   \hline
   $\obs_B \colon 0$ \\
 };
%   \path[policy] (A) edge (B);
   
 \node[systemstate] (s2) [right=of s1] {
 $s_2$
 \\
 \hline
  \agents
  \\
   \hline
   $\obs_B \colon 0$ \\
 };
    \path[policy] (A) edge (B);

  \node[systemstate] (s4) [right=of s2] {
 $s_4$
 \\
 \hline
  \agents
  \\
   \hline
   $\obs_B \colon 0$ \\
 };

 \node[systemstate] (s3) [below=of s4] {
 $s_3$
 \\
 \hline
  \agents
  \\
   \hline
   $\obs_B \colon 1$ \\
 };
 %   \path[policy] (A) edge (B);

 %   \path[policy] (A) edge (B);

 \node[systemstate] (s5) [below=of s0] {
 $s_5$ 
 \\
 \hline
  \agents
  \\
   \hline
   $\obs_B \colon 0$ \\
 };
%   \path[policy] (A) edge (B);
   
 \node[systemstate] (s6) [right=of s5] {
 $s_6$
 \\
 \hline
  \agents
  \\
   \hline
   $\obs_B \colon 0$ \\
 };
  %  \path[policy] (A) edge (B);

 \path (s0) edge node {$a$} (s1) 
 (s1) edge node {$b$} (s2)
 (s2) edge node {$a$} (s3) 
 (s2) edge node {$b$} (s4)
 (s0) edge node {$b$} (s5) 
 (s5) edge node {$a$} (s6);
\end{tikzpicture}
}
\end{center}
\caption{The permissive interpretation is not monotonic: policy $\interferes$}
\label{fig:refinement}
\end{figure}

On the other hand, consider the policy $\interferes'$, which is defined identically to $\interferes$, except that on 
state $s_6$ we take $s_6 \models {A \interferes' B}$. Plainly, ${\interferes} \leq {\interferes'}$. 
However, the system is not $\maytaname$-secure with respect to $\interferes'$. 
Note that with respect to policy $\interferes'$, 
we have $\mayta {aba} B = ((\epsilon,\epsilon,b), (\epsilon,\epsilon,a), a) = \mayta {baa} B$, 
whereas $\obs_B(s_0\cdot aba) = 1 $ and $\obs_B(s_0\cdot baa)=0$. 
\end{example} 

However, for local policies, $\maytaname$-security behaves as expected. 
This follows from Proposition~\ref{prop:mustta-monotonic}, the equivalence of $\maytaname$-security and $\musttaname$-security 
for local policies. The following result shows that, in fact it holds slightly more generally, 
and we can also use monotonicity to derive $\maytaname$-security for non-local policies. 

\begin{proposition} 
\label{prop:may-monotonic} 
Suppose that ${\interferes} \leq {\interferes'}$ and that $\interferes$ is 
local. 
If a system $M$ is $\maytaname$-secure with respect to $\interferes$, then $M$ is 
$\maytaname$-secure with respect to $\interferes'$.
\end{proposition} 

\begin{proof}
 Let $M$ be a $\maytaname$-secure system with respect to $\interferes$. 
 Since $\interferes$ is local, we obtain from Theorem~\ref{thm:local-must-may} that $M$ is $\musttaname$-secure with respect to $\interferes$. 
 By Proposition~\ref{prop:mustta-monotonic}, the system is $\musttaname$-secure with respect to $\interferes'$. 
 Since $\musttaname$-security implies $\maytaname$-security (by Lemma~\ref{mayta-cont-unwind} and Theorem~\ref{thm:mustunwind}), we have that $M$ is $\maytaname$-secure with respect to $\interferes'$. 
\end{proof}

\section{Proof Techniques} 

\label{sec:prooftechniques}

In the theory of static  (intransitive) noninterference  policies, 
{\em unwinding relations} provide a proof technique that 
can be used to show that a system is secure.  An unwinding
relation is a relation on the states of the system that 
satisfies certain conditions. We show in this section that a similar technique applies 
to dynamic policies. In fact, the conditions that we need are very similar to those 
defined on the policy above, with the main difference being that they are stated
with respect to states of the system rather than traces.

\subsection{Proof technique for $\musttaname$-security}

Let 
$M$ be a system and let $\approx_u$ be an equivalence 
relation on the states of $M$ for each $u\in \Dom$. 

\begin{itemize} 
\item[(DLR$_M$)]   If $\alpha \models {\dom(a) \noninterferes u}$ then $s_0\cdot \alpha a \approx_u s_0 \cdot \alpha$. 
\item[(WSC$_M$)]   If $s \approx_u t$ and $s \approx_{\dom(a)} t$ then $s\cdot a \approx_u t\cdot a$. 
\item[(OC$_M$)]   If $s \approx_u t$ then $\obs_u(s) = \obs_u(t)$. 
\end{itemize} 
We call a collection $\{\approx_u \}_{u\in \Dom}$ of equivalence relations satisfying DLR$_M$, WSC$_M$ and 
OC$_M$ a $\Box$-unwinding on $M$ with respect to $\interferes$.

We remark that for policy-enhanced systems, we have the following version of DLR$_M$: 
\begin{itemize} 
\item[(DLR-PE$_M$)]   If $s \models \dom(a) \noninterferes u$ then $s\cdot  a \approx_u s$. 
\end{itemize} 
which formulates the conditions to uniformly quantify over states of the system rather than both 
over states and traces. 
It is straightforward to show that if a policy-enhanced system $M$ satisfies 
DLR-PE$_M$ then it satisfies DLR$_M$.

The following result states that to prove $\musttaname$-security, it suffices to  
show that there exists a $\Box$-unwinding on the system.  

\begin{theorem} 
\label{thm:unwind-mustsecure}
Suppose that there exist equivalence relations $\{\approx_u\}_{u \in \Dom}$ on the states
of a system $M$ satisfying DLR$_M$, WSC$_M$ and OC$_M$ with respect to a 
policy $\interferes$. Then $M$ is $\musttaname$-secure with respect to $\interferes$.  
\end{theorem} 

\begin{proof} 
To show that $M$ is $\musttaname$-secure, we argue that the relations $\{\unw_u\}_{u \in \Dom}$, which by definition satisfy WSC and DLR, also satisfy OC. 
For this, we claim that for all $\alpha, \beta\in A^*$, if 
$\alpha \unw_u\beta$ then $s_0\cdot \alpha \approx_u s_0\cdot\beta$.
It is  then immediate from OC$_M$ that OC holds for $\unw_u$. 

The proof of the claim is by induction on the derivation of $\alpha \unw_u\beta$. 
The cases of $\alpha \unw_u\beta$ obtained by reflexivity, symmetry or
transitivity are straightforward, using the fact that $\approx_u$ has these
properties. For DLR, if $\alpha \unw_u\beta$ is obtained by an application of 
DLR then $\alpha = \beta a$ for some $a\in A$ with $\alpha \models {\dom(a) \noninterferes u}$. 
By DLR$_M$, we have ${s_0\cdot \alpha} \approx_u {s_0\cdot \beta}$, as required. 
For WSC, suppose that $\alpha \unw_u\beta$ is obtained by an application of WSC, 
so that we have $\alpha = \alpha' a$ and $\beta = \beta' a$. where 
$\alpha \unw_u \beta$ and $\alpha \unw_{\dom(a)} \beta$. 
By induction, we have $s_0 \cdot \alpha' \approx_u s_0 \cdot \beta'$ and $s_0 \cdot \alpha' \approx_{\dom(a)} s_0 \cdot \beta'$. 
We now obtain $s_0 \cdot \alpha' a\approx_u s_0 \cdot \beta' a $ by WSC$_M$, i.e., 
$s_0 \cdot \alpha\approx_u s_0 \cdot \beta$, as required. 
\end{proof} 

An unwinding on the system can often be naturally defined.  
We give an example below in the case of access control systems.

We may also show that unwindings provide a complete proof method, 
modulo the fact that we may need to work in a bisimilar system, rather
than the system as presented. (See \cite{Meyden15} for an example that demonstrates
that, already for static policies, we do not obtain completeness if we 
require that the unwinding be defined over the states of the system $M$ itself. 
A more complex type of unwinding on the system $M$ itself is shown to be complete for \tasecurity of static policies 
in \cite{EggertMSW13}. It would be interesting to develop a generalization of this to the 
dynamic case, but we leave this for future work.)

\begin{theorem} 
Suppose that $M$ is $\musttaname$-secure with respect to $\interferes$. Then there 
exists a system $M'$ that is bisimilar to $M$ (specifically, we may take $M'= \unwound(M)$), 
such that there exist equivalence relations $\{\approx'_u\}_{u \in \Dom}$ on the states
of  $M'$ satisfying DLR$_{M'}$, WSC$_{M'}$ and OC$_{M'}$ with respect to 
policy $\interferes$. 
\end{theorem} 

\begin{proof} 
Suppose that $M$ is $\musttaname$-secure with respect to $\interferes$. 
Then the relations $\{\unw_u\}_{u \in \Dom}$ are defined on the states of 
$M'= \unwound(M)$, and  by definition, satisfy WSC and $DLR$. 
These are equivalent to WSC$_{M'}$ and DLR$_{M'}$. 
Moreover, by Theorem~\ref{thm:mustunwind}, we have that ${\unw_u} = {\sim^\Box_u}$ for all $u\in D$, 
so OC$_{M'}$ follows from the fact that $M$ is $\musttaname$-secure. 
\end{proof}

\subsection{Proof technique for $\maytaname$-security} 

A similar proof theory can be developed for the notion $\maytaname$-security. 
We replace the condition WSC$_M$ by the following:  

\begin{itemize} 
\item[(WSC$^\Diamond_M$)]   If $\alpha \models {\dom(a) \interferes u}$ and $\beta  \models {\dom(a) \interferes u}$  and ${s_0\cdot \alpha}  \approx_u {s_0 \cdot \beta}$ 
and ${s_0\cdot \alpha} \approx_{\dom(a)} {s_0 \cdot \beta}$ 
then ${s_0\cdot \alpha a} \approx_u {s_0 \cdot \beta a}$. 
\end{itemize} 
On policy enhanced systems $M$, it is  easily seen that the following variant is sufficient for WSC$^\Diamond_M$ to hold:
\begin{itemize} 
\item[(WSC-PE$^\Diamond_M$)]   If $s\models {\dom(a) \interferes u}$ and $t  \models {\dom(a) \interferes u}$  and $s \approx_u t $ 
and $s  \approx_{\dom(a)} t $ 
then ${s\cdot  a} \approx_u {t \cdot a}$. 
\end{itemize} 

\begin{theorem} 
\label{thm:unwind-maysecure}
Suppose that there exist equivalence relations $\{\approx_u\}_{u \in \Dom}$ on the states
of a system $M$ satisfying DLR$_M$, WSC$^\Diamond_M$ and OC$_M$ with respect to a 
policy $\interferes$. Then $M$ is $\maytaname$-secure with respect to $\interferes$.  
\end{theorem} 

\begin{proof*} 
We claim that for all $\alpha,\beta \in A^*$ we have $\mayta \alpha u = \mayta \beta u$ implies
that $s_0 \cdot \alpha  \approx_u s_0\cdot \beta $. It follows from this using  OC$_M$ that 
$M$ is $\maytaname$-secure. 

To prove the claim we proceed by induction on $|\alpha| + |\beta|$. 
The base case of $\alpha = \beta = \epsilon$ is trivial. 
For the inductive case, suppose (without loss of generality) that $\alpha = \alpha' a$, and we have 
$\mayta \alpha u= \mayta \beta u$. There are two possibilities: 
\begin{itemize} 
\item $\alpha \models {\dom(a) \noninterferes u}$. In this case we have $\mayta {\alpha'} u = \mayta \alpha u = \mayta \beta u$, so by the inductive hypothesis, 
we have ${s_0\cdot \alpha'} \approx_u {s_0\cdot \beta}$. By DLR$_M$, we also have that ${s_0\cdot \alpha' a} \approx_u {s_0\cdot \alpha'}$. 
Since $\approx_u$ is an equivalence relation, we have ${s_0\cdot \alpha} \approx_u {s_0\cdot \beta}$. 

\item $\alpha \models {\dom(a) \interferes u}$. 
In this case, $\mayta \alpha u = ( \mayta \alpha u, \mayta \alpha {\dom(a)}, a)$. It follows that $\beta \neq \epsilon$,
so we may write $\beta = \beta' b$ for some $\beta'\in A^*$ and $b\in A$. If $\beta' \models {\dom(b) \noninterferes u}$, 
we may reverse the roles of $\alpha$ and $\beta$ and obtain that  ${s_0\cdot \alpha} \approx_u {s_0\cdot \beta}$ using the argument of the previous case.  
If $\beta'\models {\dom(b) \interferes u}$, then $\mayta \beta u = ( \mayta {\beta'} u, \mayta {\beta'}  {\dom(b)}, b)$, and we
conclude that $a=b$ and $\mayta {\alpha'} u = \mayta {\beta'} u$ and $\mayta {\alpha'} {\dom(a)} = \mayta {\beta'} {\dom(a)}$. 
By the inductive hypothesis, we have ${s_0\cdot\alpha'} \approx_u {s_0\cdot\beta'}$ and ${s_0\cdot\alpha'} \approx_{\dom(a)} {s_0\cdot\beta'}$. 
Using WSC$^\Diamond_M$ it follows that ${s_0\cdot\alpha' a} \approx_u {s_0\cdot\beta' a}$, i.e., ${s_0\cdot\alpha} \approx_u {s_0\cdot\beta}$. 
\qedhere 
\end{itemize} 
\end{proof*} 
We also have completeness of the technique, provided 
we allow its application on a bisimilar system: 

\begin{theorem} 
\label{thm:unwind-mustsecure}
Suppose that $M$ is $\musttaname$-secure with respect to $\interferes$.  
Then there exists a system $M'$ that is bisimilar to $M$ and 
there exist equivalence relations $\{\approx_u\}_{u \in \Dom}$ on the states
of a system $M'$ satisfying DLR$_{M'}$, WSC$^\Diamond_{M'}$ and OC$_{M'}$ with respect to 
policy $\interferes$. 
\end{theorem} 

\begin{proof} Let $M = \tuple{S, s_0, \trans, \Dom, \dom, \Actions, \obs}$. 
We take $M' = \tuple{S', s'_0, \trans', \Dom, \dom, \Actions, \obs'}$  to be the system 
over the same signature $(\Dom, \dom, \Actions)$
with states $S' = A^*$, initial state $s_0' = \epsilon$, 
transitions $\trans'$ defined by $\alpha \cdot a = \alpha a$ for all $\alpha \in A^*$, and observations defined
$\obs'_u(\alpha) = \obs_u(s_0\cdot\alpha)$ for all $\alpha \in S'$. It is easily checked that this system
is bisimilar to $M$. 

We show that the relations $\{\sim^\Diamond_u\}_{u \in \Dom}$ satisfy  DLR$_{M'}$, WSC$^\Diamond_{M'}$ and OC$_{M'}$ with respect to 
policy $\interferes$. Note that for all $\alpha\in A^*$, we have $s_0'\cdot \alpha  = \alpha$. 
For condition DLR$_{M'}$, note that if $\alpha \models \dom(a) \noninterferes u$ then $\mayta {\alpha a} u = \mayta \alpha u$, 
so $\alpha a \sim^\Diamond_u\alpha$. 
For condition WSC$^\Diamond_{M'}$, note that  if $\alpha \models \dom(a) \noninterferes u$
and  $\beta \models \dom(a) \noninterferes u$ and $\alpha  \sim^\Diamond_u\beta$ and 
$\alpha  \sim^\Diamond_{\dom(a)}\beta$ then 
\[ 
\begin{array}{rcl} 
\mayta {\alpha a} u & = & (\mayta \alpha u , \mayta \alpha {\dom(a)}, a) \\ 
& = & (\mayta \beta u , \mayta \beta {\dom(a)}, a) \\ 
& = & \mayta {\beta a} u
\end{array} 
\] 
hence $\alpha a \sim^\Diamond_u \beta a$. 
Condition OC$_{M'}$  is immediate from the fact that $M$ (and hence $M'$) is $\maytaname$-secure. 
\end{proof}

\section{Access Control 
in Dynamic Systems}
\label{sec:accesscontrolinterpretation}

Access control monitors are a fundamental security mechanism
that can be used to enforce a variety of policies. Rushby \cite{rushby92} 
showed that access control can provide a sound enforcement 
mechanism for static intransitive noninterference policies, and that a simple static check of the access control table suffices
to verify the policy will be enforced. It was later shown by van der Meyden 
\cite{Meyden15} that, in fact, access control is also a \emph{complete} enforcement mechanism
for such policies, in the sense that every system that satisfies the policy is 
bisimilar to one using access control, with a setting that satisfies the static check.  
In this section, we develop a generalization of these results to 
dynamic noninterference policies. 

A system with structured state, is a system equipped with 
\begin{itemize}
  \item a set of all possible objects $\Objects$, 
  \item a set of all possible values $\Values$, 
  \item a function $\contentsnme \colon \Objects \times \States \rightarrow \Values$, 
   \item a function $\observename \colon \Dom \times \States \rightarrow \powerset{\Objects}$, 
   \item a function $\altername \colon \Dom \times \States \rightarrow \powerset{\Objects}$. 
\end{itemize}
Intuitively, $\contents{o}{s}$ is the value of object $o$ in state $s$, 
$\observe{u}{s}$ is the set of objects that domain $u$ is permitted to observe (or read) in state $s$, 
and $\alter{u}{s}$ is the set of objects that domain $u$ is permitted to alter (or write) in state $s$. 
The pair of functions $\observename, \altername$ can be thought of as an encoding
of an access control table that varies with the state of the system.

We make an  assumption on the set $O$ of objects, namely, that this contains  
a special object $\oset{u}$ for each domain $u \in \Dom$.
We write the  set of all these special objects as $\Osets(\Dom) = \{ \oset{u} \mid u \in \Dom \}$ and assume $\Osets(D) \subseteq \Objects$. 
We assume that 
$\contents{\oset{u}}{s} = \observe{u}{s}$, 
so that the content of $\oset{u}$ is the set of all objects 
that are observable to $u$. 
Furthermore, we assume that $\oset{u} \in \observe u s$ for all domains $u\in D$ and states $s\in\States$, i.e., 
the object $\oset{u}$ is always observable by $u$. 
This is intuitively reasonable, given the meaning of this object, since domain $u$ can  always determine the value of 
$\oset{u}$ by checking directly which objects it is able to observe.

The motivation for these special objects 
$\oset{u}$ is that they are a  kind of ``meta-object''
that can be used to specify if an agent is allowed to change the observable objects of another agent. 
More precisely, we may let $\oset{u} \in \alter{v}{s}$ if and only if in state $s$, domain $v$ is allowed to influence which objects are observable by $u$. 

We now state a number of conditions (resembling Rushby's ``reference monitor conditions" \cite{rushby92}) that are 
intended to capture the intuitive semantics of the access control functions $\observename$, $\altername$, 
as they might be enforced by a reference monitor that restricts agents' ability to read and write objects
as these access control settings change. 
The conditions are cast in terms of the relations $\dynacrel u$, for $u\in \Dom$, defined by  
%\begin{multline*}
$  s \dynacrel u t \text{ iff } 
   \text{for all } o \in \observe u s, %\\
   \text{ we have } \contents o s = \contents o t $. %\enspace. 
%\end{multline*}
Intuitively, $s\dynacrel u t$ when the objects observable by domain $u$ have the same values in the states
$s$ and $t$. Note that because $\oset{u} \in \observe u s$, it is always that case that 
$\observe u s = \observe u t$ when $s \dynacrel u t$. It follows that the relations $\dynacrel u$ are in fact equivalence 
relations. Our new ``dynamic reference monitor" conditions
are the following: 
\begin{enumerate}[\DRM-1]
  \item  If $s \dynacrel u t$, then $\obs_u(s) = \obs_u(t)$. 
  \label{cond:wa:outputconsistency}

  \item
     If $o \in \alter{\dom(a)} s \cap \alter{\dom(a)} t$ and 
     $s \dynacrel{\dom(a)} t$
      and  $\contents o s = \contents o t$ then  $\contents o {s\cdot a} = \contents o {t \cdot a}$. 
  \label{cond:wa:contentconsistency}

  \item
  If  $\contents o {s \cdot a} \neq \contents o s$, then 
  $o \in \alter{\dom(a)} s$. 
  \label{cond:wa:allowedinterference}

  \item ${\observe u {s \cdot a}\! \setminus\! \observe u s}\! \subseteq \! \observe{\dom(a)} s$.
    \label{cond:wa:grantobserve}

  \item 
  If $\alpha \models v \interferes u$ and $\beta \models v \interferes u$ and  
  $s_0\cdot \alpha \dynacrel{u} s_0\cdot \beta$ and $s_0\cdot \alpha \dynacrel{v} s_0\cdot \beta$, then $\observe u {s_0\cdot \alpha} \cap \alter v {s_0\cdot\alpha} = \observe u {s_0\cdot\beta}  \cap \alter v {s_0\cdot \beta}$. 
  \label{cond:wa:weakpolicyconsistency}

 \item
  If $\alter u {s_0\cdot \alpha} \cap \observe v {s_0\cdot \alpha} \neq \emptyset$, then $\alpha \models {u \interferes v}$. 
  \label{cond:wa:commonvar}

\end{enumerate}
Intuitively, condition \DRM-\ref{cond:wa:outputconsistency} says that a domain's observation depends only on the
contents of the objects observable to the domain. 
Condition \DRM-\ref{cond:wa:contentconsistency} says that when a domain is permitted to alter an object, the new value of the 
object depends only on its previous value and the values of the objects observable by the acting domain. 
(The version of this constraint given here follows \cite{Meyden15} rather than \cite{rushby92}.)  
Condition \DRM-\ref{cond:wa:allowedinterference} can be understood as stating that an 
action can change the value of an object only when the domain of the action is permitted by 
the access control setting to alter the object.  
Condition \DRM-\ref{cond:wa:grantobserve} states that any object that is made newly observable to a
domain $u$ by an action $a$ must be observable to the domain $\dom(a)$. 

Condition \DRM-\ref{cond:wa:policyconsistency} is a kind of locality constraint, 
which states that 
when $v$ is permitted to interfere with $u$, in situation $s$, 
the set $\observe u s \cap \alter v s $ 
of objects that $u$ can observe and $v$ can alter 
depends only on information local to 
domains $u$ and $v$, as captured by the relations $\dynacrel{u}$ and $\dynacrel{v}$. 
We note that this constraint actually only states a weak property of $\alter v s$, as 
we already have that $\observe u s$ depends only on the local state of $u$ as captured by $\dynacrel{u}$.
\newcommand{\strongpolicyconsistency}{\xspace{\DRM-5$'$}\xspace}
The following stronger version of this condition is also of interest: 
 \begin{itemize} 
\item[\DRM-5$'$] If $s \dynacrel{u} t$ and $s \dynacrel{v} t$, then $\observe u s \cap \alter v s = \observe u t \cap \alter v t$. 
 \label{cond:wa:policyconsistency}
\end{itemize} 
Here we have dropped the requirement that $v$ be permitted to interfere with $u$ in order for the
dependency to hold. This stronger condition more closely resembles locality.

We remark that have not introduced special objects $\mathtt{aset}(u)$ for each domain $u$, with the property 
that $\contents {\mathtt{aset}(u)} s = \alter u s$, as might be suggested by analogy with the 
objects $\oset{u}$. The results to follow do not require the introduction of these objects. 
However, in systems with such objects, and also satisfying the  property
$\mathtt{aset}(u) \in \observe u s$ condition \strongpolicyconsistency\   would be necessarily satisfied.

Finally, condition \DRM-\ref{cond:wa:commonvar} relates the access control setting to the 
policy. Note that if there exists an object that domain $u$ is able to write and domain $v$ is able
to read, then this object provides an obvious channel for information to flow from $u$ to $v$. 
The condition states that in this situation, the resulting channel for information flow is permitted 
by the information flow policy $\interferes$. 

\begin{theorem}
\label{thm:ac-secure}
Suppose $M$ is a system with structured state that satisfies properties \DRM-{1} to \DRM-{6}  with respect to policy $\interferes$. 
Then $M$ is  $\maytaname$-secure with respect to~$\interferes$. 
If $M$ also satisfies the stronger condition~\strongpolicyconsistency, then 
$M$ is  $\musttaname$-secure with respect to~$\interferes$. 
\end{theorem}

\begin{proof} 
Assume that $M$ is a system with structured state that satisfies properties \DRM-{1} to \DRM-{6}  with respect to policy $\interferes$. 
We first show that $M$ satisfies DLR$_M$ and WSC$^\Diamond_M$ and OC$_M$ with respect to the relations $\{\dynacrel u\}_{u\in \Dom} $.  

For property DLR$_M$,  suppose that $\alpha \models \dom(a) \noninterferes u$. Then by 
\DRM-\ref{cond:wa:commonvar}, we have 
$\alter {\dom(a)} {s_0\cdot \alpha} \cap \observe u {s_0\cdot \alpha} = \emptyset$. 
Thus, for all $o\in \observe u {s_0\cdot \alpha}$, we have $ o \not \in \alter {\dom(a)} {s_0\cdot \alpha}$, and by 
\DRM-\ref{cond:wa:allowedinterference} we obtain that $\contents o  {s_0\cdot \alpha } = \contents o {s_0\cdot \alpha a}$. 
This shows that  $s_0\cdot \alpha a\dynacrel{u} s_0 \cdot \alpha $, as required. 

For WSC$^\Diamond_M$, suppose that 
$\alpha \models \dom(a) \interferes u$ and $\beta \models \dom(a) \interferes u$ and 
$s_0\cdot \alpha  \dynacrel{u} s_0\cdot \beta $ and $s_0\cdot \alpha \dynacrel{\dom(a)} s_0\cdot \beta$. 
We need to show that ${s_0 \cdot \alpha a} \dynacrel{u} {s_0 \cdot \beta a}$, 
 i.e.,  that for all  $o \in \observe u {s_0\cdot \alpha a}$, we have  $\contents o {s_0\cdot \alpha a} = \contents o {s_0 \cdot \beta a}$. 
 Let $o \in \observe u {s_0\cdot \alpha a}$. By \DRM-\ref{cond:wa:grantobserve}
we have that $o \in \observe u {s_0\cdot \alpha } \cup \observe {\dom(a)} {s_0\cdot \alpha}$. 
 From $s_0\cdot\alpha  \dynacrel{u} s_0\cdot\beta$ and $s_0\cdot\alpha \dynacrel{\dom(a)} s_0\cdot\beta$ we obtain 
 that $\observe u {s_0\cdot\alpha}  = \observe u {s_0\cdot\beta}$ and $\observe  {\dom(a)} {s_0\cdot\alpha} = \observe {\dom(a)} {s_0\cdot\beta}$, 
 and for all $o \in \observe u {s_0\cdot\alpha} \cup \observe {\dom(a)} {s_0\cdot\alpha}$ we have $\contents o {s_0\cdot\alpha} = \contents o {s_0\cdot\beta}$.  
 By two applications of \DRM-\ref{cond:wa:policyconsistency} (one for the pair of domains $u$, $\dom(a)$, and 
 the other for the pair $\dom(a)$ ,$\dom(a)$), we obtain that $o \in \alter {\dom(a)} {s_0\cdot\alpha}$ iff $o \in \alter {\dom(a)} {s_0\cdot\beta}$. 
 It therefore suffices to consider two cases: 
\begin{itemize} 
\item Suppose $ o \in \alter {\dom(a)} {s_0\cdot\alpha} \cap \alter {\dom(a)} {s_0\cdot\beta}$. Then by \DRM-\ref{cond:wa:contentconsistency}
we get $\contents o {s_0\cdot\alpha a }  = \contents o {s_0\cdot\beta a}$.  
 \item Suppose $ o \not \in  \alter {\dom(a)} {s_0\cdot\alpha}$ and $o \not \in  \alter {\dom(a)} {s_0\cdot\beta}$. Then by \DRM-\ref{cond:wa:allowedinterference}
 we have $\contents o {s_0\cdot\alpha a }  = \contents o {s_0\cdot\alpha} = \contents o {s_0\cdot\beta} = \contents o {s_0\cdot\beta a}$. 
 \end{itemize}  

The property OC$_{M}$ is trivial, since this is just \DRM-\ref{cond:wa:outputconsistency}. 
We now have that $M$ satisfies DLR$_M$  and WSC$^\Diamond_M$ and OC$_M$.  By Theorem~\ref{thm:unwind-maysecure}, 
it follows that $M$ is $\maytaname$-secure with respect to $\interferes$.

Suppose that $M$ also satisfies the stronger condition~\strongpolicyconsistency. 
We show that $M$ is  $\musttaname$-secure with respect to $\interferes$. Consider the policy $\interferes^M$  defined by $\alpha \models {u \interferes^M v}$ if 
$\alter u {s_0\cdot \alpha } \cap \observe {v} {s_0\cdot \alpha} \neq \emptyset$. It is easily checked that $M$ satisfies 
\DRM-{1} to \DRM-{6}  with respect to policy $\interferes^M$, so by the above, we conclude that $M$ is  $\maytaname$-secure with respect to $\interferes^M$.
We claim that the policy $\interferes^M$ is local. This being the case, we obtain using Theorem~\ref{thm:local-must-may}  that $M$ is  $\musttaname$-secure with respect to $\interferes^M$.
By \DRM-\ref{cond:wa:commonvar}, we have that ${\interferes^M} \leq {\interferes}$, so by Proposition~\ref{prop:mustta-monotonic} we get that $M$ is  $\musttaname$-secure with respect to $\interferes$. 

To show that $\interferes^M$ is local, define the equivalence relations $\{\sim_u\}_{u\in \Dom}$ on $A^*$ by $\alpha \sim_u\beta$ iff $s_0\cdot \alpha \dynacrel u s_0\cdot \beta$. 
It follows from the fact that $M$ satisfies DLR$_M$ with respect to $\{\dynacrel u\}_{u\in \Dom}$ that $M$ satisfies DLR with respect to $\{\sim_u\}_{u\in \Dom}$. 
An argument similar to that above for WSC$^\Diamond_M$, but using \strongpolicyconsistency in place of \DRM-\ref{cond:wa:policyconsistency}, shows that 
$M$ also satisfies WSC with respect to   $\{\sim_u\}_{u\in \Dom}$. Thus, for the smallest equivalence relations $\{\unw_u\}_{u\in \Dom}$ 
satisfying WSC and DLR with respect to $\interferes^M$, we have ${\unw_u} \subseteq {\sim_u}$ 
for all $u\in \Dom$. Suppose that $\mayta \alpha u = \mayta \beta u$ and $\mayta \alpha v = \mayta \beta v$ (with respect to policy $\interferes^M$). 
We need to show that $\alpha \models {u \interferes^M v}$ iff $\beta \models {u \interferes^M v}$. 
By Lemma~\ref{mayta-cont-unwind}, we have that $\alpha \unw_u \beta$ and $\alpha \unw_v\beta$. 
By what was shown above, we obtain that $\alpha \sim_u \beta$ and $\alpha \sim_v\beta$, 
so $s_0\cdot \alpha \approx_u s_0\cdot \beta$ and $s_0\cdot\alpha \approx_v s_0\cdot \beta$. 
By \strongpolicyconsistency we now obtain that 
$\observe u {s_0\cdot \alpha} \cap \alter v {s_0\cdot \alpha} =  \observe u {s_0\cdot \beta} \cap \alter v {s_0\cdot \beta}$. 
It follows that $\alpha \models {u \interferes^M v}$ iff $\beta \models {u \interferes^M v}$, as required. 
\end{proof} 

Theorem~\ref{thm:ac-secure} shows that access control mechanisms 
provide an implementation method that guarantees 
that a system constructed using this method is secure. 
It states a soundness result: any system that satisfies the 
reference monitor conditions with respect to a policy is secure for that 
policy.

It is also possible to prove a converse to this result, stating that 
any system that is secure with respect to a policy 
could have been constructed using access control 
mechanisms so as to satisfy the policy. To obtain such a 
result, we need to allow consideration of 
a bisimilar system. 
In particular, we focus on the 
\emph{unfolding} of a system $M$. 

Say that a system \emph{admits an access control interpretation consistent with a policy 
$\interferes$}, if it can be extended into a system with structured state by adding
definitions of $O, V,  \contentsnme, \observename$ and $\altername$ in such a way as to 
satisfy conditions \DRM-1 to \DRM-6.

\begin{theorem} 
\label{thm:ac-complete} 
Let system $M$ be  $\maytaname$-secure   \wrt $\interferes$. 
Then $\unwound(M)$ admits an access control interpretation consistent with $\interferes$. 

If $M$ is $\musttaname$-secure  \wrt $\interferes$, 
then 
there exists a local policy ${\interferes'} \leq {\interferes}$ such that 
$\unwound(M)$ admits an access control interpretation consistent with $\interferes'$,
that additionally satisfies \DRM-$5'$. 
\end{theorem}

\begin{proof}
Suppose first that $M$ is $\maytaname$-secure with respect to $\interferes$.  
Extend $\unwound(M)$ to a system with structured state by the following definitions: 
 \begin{align*}
  & \Objects  = \Dom \cup \Osets(D) \\
  & \observe u {\alpha} = \{ u, \oset{u} \} \\
  & \alter u {\alpha} = \{ v \in \Dom \mid \alpha \models u \interferes v \} \\
  & \contents x {\alpha} = 
  \begin{cases}
   \mayta \alpha x & \text{if } x \in \Dom \\
   \{u, \oset{u}\} 
   & \text{if } x = \oset{u} 
  \end{cases}   
 \end{align*}
We will show that the dynamic reference monitor conditions \DRM-1 to \DRM-6 hold. 
 \begin{enumerate}
  \item[\DRM-\ref{cond:wa:outputconsistency}:]
    Let ${ \alpha} \dynacrel{u} { \beta}$. 
    Then we have $u \in \observe u { \alpha} = \observe u { \beta}$ and $\mayta \alpha u = \contents u { \alpha} = \contents u { \beta} = \mayta \beta u$. 
    It now follows by $\maytaname$-security of $M$ that $\obs_u( \alpha) = \obs^M_u( s_0\cdot \alpha) = \obs^M_u(s_0\cdot \beta) = \obs_u( \beta)$. 
  \item [\DRM-\ref{cond:wa:contentconsistency}:]
    Let $o \in \alter{\dom(a)}{ \alpha} \cap \alter{\dom(a)}{ \beta}$ and ${ \alpha} \dynacrel{\dom(a)} { \beta}$ and $\contents o { \alpha} = \contents o { \beta}$.     
    From ${ \alpha} \dynacrel{\dom(a)} { \beta}$ and the fact that, by definition,  $\dom(a) \in \observe {\dom(a)} \alpha$ we obtain $\mayta \alpha {\dom(a)} = \mayta \beta {\dom(a)}$. 
    From $o \in \alter{\dom(a)}{ \alpha} \cap \alter{\dom(a)}{ \beta}$ it follows that $o\in D$ and $\alpha \models {\dom(a) \interferes o}$ and $\beta \models {\dom(a) \interferes o}$. 
    From $\contents o { \alpha} = \contents o { \beta}$ it follows that $\mayta \alpha o = \mayta \beta o$.  By the definition of the $\maytaname$-operator, we obtain $\mayta {\alpha a} o = \mayta {\beta a} o$. 
  \item[\DRM-\ref{cond:wa:allowedinterference}:] We prove the contrapositive. 
    Suppose $o \notin \alter{\dom(a)}{ \alpha}$. 
    We can assume that $o \in \Dom$, since the values of objects in $\Osets$ do not change. 
    Thus, we have $\alpha \models {\dom(a) \interferes o}$. 
    Then we have $\mayta {\alpha a} o = \mayta {\alpha} o$ and hence $\contents o { \alpha a} = \contents o { \alpha}$.

  \item[\DRM-\ref{cond:wa:grantobserve}:]  
    Since the function $\observename$ is constant, we have $\observe u {\alpha \cdot a} \setminus \observe u \alpha = \emptyset$. 
  \item[\DRM-\ref{cond:wa:policyconsistency}:] 
   By definition, we have that 
    \begin{align*}
     \observe u { \alpha} \cap \alter v { \alpha} & = 
     \begin{cases}
      \{u \} & \text{if } \alpha \models {v \interferes u} \\
      \emptyset & \text{otherwise}
     \end{cases} \\
      \observe u { \beta} \cap \alter v { \beta} & = 
     \begin{cases}
      \{u \} & \text{if } \beta \models {v \interferes u} \\
      \emptyset & \text{otherwise}
     \end{cases}
    \end{align*}
    It is immediate that if 
    $\alpha \models {v \interferes u}$ and $\alpha \models {v \interferes u}$ then 
    $\observe u { \alpha} \cap \alter v { \alpha} = \{u\} =  \observe u { \beta} \cap \alter v { \beta}$. 
    (We remark that we do not need the antecedents ${ \alpha} \dynacrel{u} { \beta}$ and ${ \alpha} \dynacrel{v} { \beta}$ of  \DRM-\ref{cond:wa:policyconsistency}, so we 
    have actually proved something stronger.) \\
  \item[\DRM-\ref{cond:wa:commonvar}:] We prove the contrapositive. 
    Suppose $\alpha \models {u \noninterferes v}$. 
    Then clearly 
    $\alter u { \alpha} \cap \observe v \alpha = \{ w \in \Dom \mid \alpha \models {u \interferes w} \} \cap \{ v, \oset{v} \} = \emptyset$. 
 \end{enumerate}
 
Next, suppose that $M$ is $\musttaname$-secure with respect to $\interferes$. 
By Lemma~\ref{lem:restrict-to-local} there exists a local policy ${\interferes'} \leq {\interferes}$ such that $M$ is $\musttaname$-secure with respect to $\interferes'$. 
By Theorem~\ref{thm:local-must-may}, $M$ is $\maytaname$-secure with respect to $\interferes'$. 
Hence, we may apply the above construction with respect to $\interferes'$ instead of $\interferes$, 
to obtain an access control interpretation consistent with $\interferes'$. 
We show that the stronger condition \strongpolicyconsistency\ also holds in this case. 
 Let ${ \alpha} \dynacrel{u} { \beta}$ and ${ \alpha} \dynacrel{v} { \beta}$. 
    Then we have $\mayta \alpha u  = \mayta \beta u$ and $\mayta \alpha v = \mayta \beta v$, i.e., 
    $\alpha \sim^\Diamond_u\beta$ and      $\alpha \sim^\Diamond_v\beta$. We may work with the 
    relations $\sim^\Diamond$ in the definition of locality, so 
    from locality of $\interferes'$, it  follows that $\alpha \models v \interferes' u$ iff $\beta \models v \interferes' u$. 
    If $\alpha \models v \interferes' u$ and $\beta \models v \interferes' u$ then 
      $\observe u { \alpha} \cap \alter v { \alpha} = \{u\}  = \observe u { \beta} \cap \alter v { \beta}$. 
      Alternately, if $\alpha \models v \noninterferes' u$ and $\beta \models v \noninterferes' u$ then 
      $\observe u { \alpha} \cap \alter v { \alpha} = \emptyset = \observe u { \beta} \cap \alter v { \beta}$. 
      That is, in either case, we have the desired result. 
\end{proof}

Theorem~\ref{thm:ac-complete} can be understood as 
saying that every system that is secure with respect to a local policy could have been constructed 
using access control mechanisms satisfying the reference monitor 
conditions, in such a way that its security follows using Theorem~\ref{thm:ac-secure}. 
The phrase ``could have been constructed" here is formally interpreted as 
allowing bisimilar systems as equivalent alternate implementations. 
Thus, this result states that the reference monitor conditions provide
both a sound and complete methodology for the construction of 
systems that are secure with respect to local policies. 
This result, and the intuitive nature of the dynamic reference monitor conditions, 
provides further  justification for our definitions. 

\section{Example: a capability system} 
\label{sec:flume}

In the operating systems literature, {\em capabilities} are
a commonly used representation for access control. They
provide a conceptual model that allows for dynamic security 
policies, by allowing an agent possessing a capability to 
transmit it to another agent, thereby granting the 
receiver the powers associated with the capability. 
A number of works on {\em Decentralized Information Flow Control} 
systems \cite{Krohn2009,Myers1997,Vandebogart2007,ZeldovichBKM06} have proposed capability-based systems as a mechanism
for implementing dynamic information flow policies. 
Many of the proposals in the literature  contain covert channels.
In this section, we  develop a simplified model that is in the spirit of  these proposals, 
but tailored to be secure with respect to an associated policy. (Our model was
motivated by a consideration of Flume \cite{Krohn2009}, but we have taken
significant liberties in order to obtain a system that satisfies the desired result.
We discuss the motivation for our changes to Flume in Section~\ref{sec:flume-comments}.)

\subsection{States and Observations} 

We assume a set $\processes$ of processes, each corresponding to a domain, 
so $\Dom = \processes$. There is a fixed set $\btags$ of \emph{basic tag names}, 
from which the set $\tags = \btags \cup \{\ptag n p ~|~n \in \btags,~ p \in \processes\}$ of \emph{tags}
is generated. Here, each tag $t\in \btags$ is a tag that exists at system startup, and 
$\ptag n p$ is a tag that is labelled with the process $p$ that may create it.   A security level
is represented as a subset of $\tags$. 
A \emph{capability} is a value of the form either $t^+$ or $t^-$, where $t$ is a tag. 
Intuitively, $t^+$ represents the capability to add tag $t$ to a security level, and 
$t^-$ represents the capability to remove tag $t$ from a security level. 
We write $\Caps$ for the set of all capabilities.

Each process $p$ has an associated set $\Obj_p = \Data_p \cup \PolicyState_p$  of 
objects.  
The set $\PolicyState_p$ represents process $p$'s component of the policy, and 
consists of the following special objects, each with an associated type (of possible values for the object) 
\begin{itemize} 
\item $S_p$ of type ${\cal P}(\tags)$: the value of this object is a set of tags, and
intuitively represents the types of information that process $p$ may receive. 
The set $S_p$ is called the \emph{secrecy} set of process $p$, and 
represents its security level. 
\item $O_p$ of type ${\cal P}(\Caps)$  the value of this object is a set of capabilities. 
The set $O_p$ is called the \emph{capability set} of process $p$, and represents the 
set of powers that the process has to add and remove tags from its security level.  

\end{itemize} 
The set $\Data_p$ is the set of data objects that are local to process $p$. 
We assume that the set $\Data_p$ contains at least the following: 
\begin{itemize} 
\item $m_p$, an outgoing message buffer, used by $p$ to construct a message
prior to its transmission
\item $in_p$, of type a sequence of messages. This is the input buffer through which 
$p$ receives messages from other processes. 
\end{itemize} 

A {\em state} is a function $s$  assigning a value of the appropriate type to each 
object in $\bigcup_{p\in \processes} \Obj_p$.  We write $S$ for the set of states. 
A {\em candidate initial state} is a state $s$ such that 
for all processes $p$, we have $O_p \subseteq \{ n^x ~|~n \in \btags, x\in \{+,-\}\}$
and $S_p \subseteq \btags$. That is, in an initial state, only basic tags exist; tags local 
to a process need to be created as the system runs. 
We allow the initial state of a system to be any candidate initial state: different 
choices of initial state correspond to different initial 
configurations of the security policy. 

The observation functions of a capability system may  be defined as any
functions that satisfy the following constraint, for each process $p$ and all states $s$ and $t$: 
\begin{itemize}
\item[{\tt Obs}.]  If $s\restrict \Obj_p = t\restrict \Obj_p$ then $\obs_p(s) = \obs_p(t)$. 
\end{itemize} 
That is, observations depend, for each process, only on its own objects.

Given the intuition that the set $S_p$ represents the types of 
information that process $p$ may receive. Note that if $S_p$ contains
a tag $t$ then $p$ may have received information of type $t$. 
If $t$ is not in $S_q$ for some other process $q$, then there is a risk that if $p$ 
transfers information to $q$ then $q$ will obtain information of type $t$, 
which is prohibited. This suggests that with each state $s$, we may naturally 
associate a static policy $\interferes$, such that $p \interferes q$ iff $S_p(s) \subseteq S_q(s)$. 
In what follows, it is convenient to write $s\models S_p \subseteq S_q$ when 
$s(S_p) \subseteq s(S_q)$.

\subsection{Actions and Transitions}

The following actions are possible in the system. 
We write actions in the format $p.a$, where $p\in \processes$ is a
process, and define the domain of action $p.a$ to be 
$\dom(p.a) = p$.

\newcommand{\addcap}[3]{#1.{\sf add\_cap}(#2,#3)}
\newcommand{\addtag}[2]{#1.{\sf add\_tag}(#2)}
\newcommand{\removetag}[2]{#1.{\sf remove\_tag}(#2)}
\newcommand{\sendmessageto}[2]{#1.{\sf send\_message\_to}(#2)}
\newcommand{\sendcapto}[3]{#1.{\sf send\_cap}(#2,#3)}
\newcommand{\dropcap}[2]{#1.{\sf drop\_cap}(#2)}
%

%ron2: saving space by reducing the margin here 
%\newenvironment{spdesc}{\begin{list}{}{\setlength{\leftmargin}{5pt}}}{\end{list} }

%\begin{spdesc}
\begin{description}
\item[Data actions $p.a$] ~\\ We assume that $p$ has some set of data manipulating actions 
that act on its data objects. This set of actions can be instantiated in any way, subject to the 
constraint that these actions read only the local state of the process and write only the data objects of the process. 
More precisely, for any data action $p.a$, 
we assume the following, for all states $s$ and $t$:  
\begin{itemize} 
\item[{\tt Data1}.]  For all objects $x\not\in \Data_p$, $(s\cdot p.a)(x) = s(x)$. 
\item[{\tt Data2}.]  If $s(x) = t(x)$ for all $x\in \Obj_p$, then $(s\cdot p.a)(x) = (t\cdot p.a)(x)$ for all $x\in \Data_p$. 
\end{itemize}  
Constraint {\tt Data1} says that data actions of process $p$ change only $p$'s data objects, 
and no others. Constraint {\tt Data2} says that the way that such an action changes $p$'s 
data objects depends only on the value of $p$'s data and policy state objects. That is, a data action $p.a$ 
is deterministic and may read only objects in $\Obj_p$ and write only objects in $\Data_p$. \\

\item[$\addcap p c n$] where $n\in \btags$ and  $c\subseteq \{+, -\}$ with code
  \begin{tabbing} 
\quad   $O_p := O_p \cup \{\ptag n p^x~|~x\in c\}$ 
  \end{tabbing} 
 Intuitively, this action creates a new tag $\ptag n p$, if it does not already exist, and grants $p$  
 the rights $c$ over this tag, by adding the capabilities 
 $ \ptag n p^x$ for $x\in c$ to $p$'s capability set $O_p$. 
(If the tag already exists, the operation may be used to add capabilities   
to the existing ones.) Note that the fact that the tag created is indexed by the process $p$ 
means that two distinct processes cannot create a tag with the same name.\\
  
 \item[$\dropcap{p}{c}$]  where $c\in \Caps$, with code  
    \begin{tabbing} 
\quad $O_p := O_p \setminus \{c\} $ 
  \end{tabbing} 
This action removes the capability $c$ from $p's$ capability set, if
present. \\

 \item[$\addtag p t$] where $t\in \tags$, with code  
  \begin{tabbing} 
 \quad  If $t^+ \in O_p $ \=then $ S_p := S_p \cup \{t\}$ 
   \end{tabbing} 
  This action adds a tag to the secrecy set of the process, 
  provided it has the capability to do so. A consequence of this is
  that some outgoing edges from $p$ may be removed from the policy, 
  and some ingoing edges may be added. \\
  
 \item[$\removetag p t$]  where $t \in \tags$, with code 
    \begin{tabbing} 
  \quad If $t^- \in O_p$ \=then $ S_p := S_p \setminus \{t\} $  
  \end{tabbing} 
 This action removes a tag from the secrecy set of the process, provided it has the 
 capability to do so, and 
 consequently might add outgoing policy edges from $p$ or insert incoming edges.  
 Intuitively, by removing the tag, $p$ enables declassification of information of type $p$; 
 whereas it may have received information of  type $t$, the policy is no longer 
 constrained to restrict $p$ from communicating with processes that are 
 not permitted to receive information of type $t$. \\

 \item[$\sendmessageto p q$]   where $q$ is a process,  with code  
  \begin{tabbing} 
  \quad If $S_p \subseteq  S_q$ \=then $ in_q := in_q \cdot m_p$  
  \end{tabbing} 
This action transmits the message in the message object $m_p$ to $q$ by adding it
to the end of $q's$ input buffer $in_p$, 
provided the policy permits transmission of information from $p$ to $q$ 
(i.e., $S_p \subseteq  S_q$). The action has no effect on the policy, even 
if the message contains information about $p$'s tags and/or capabilities. 
Intuitively, we assume that the operating system is responsible for maintaining
the policy state, and mere knowledge of a capability name or value does not
suffice to be able to exercise the capability. In order to transmit the ability to exercise
capabilities, the following action needs to be used. \\

 \item[$\sendcapto p c q$]  where $c \in \Caps$ and $q$ is a process,  with code
  \begin{tabbing} 
  \quad If $S_p \subseteq  S_q$ and $c \in O_p$ \=then $O_q := O_q \cup \{c\} $ .  
  \end{tabbing}
  This action transmits the capability $c$ to $q$, 
  provided both that $p$ is permitted to transmit information to $q$ (i.e., $S_p \subseteq  S_q$)
  and $p$ actually possesses the capability (i.e., $c\in O_p$). 
 The effect of this action may be to remove outgoing edges from $q$
 or add incoming edges to $q$. 
% \end{spdesc}
 \end{description}

\subsection{Security of the capability model} 

We now establish the security of any capability system constructed as 	described above. 
Note that there are multiple such systems, since the description 
above leaves open the choice of  processes, basic tag names, data objects, data actions, initial state
and observation functions. It also remains to specify a dynamic security policy. 
Given a capability system $M$ with actions $A$ and processes $\processes$, and initial state $s_0$, we say that the 
\emph{associated policy} is the policy $\interferes^M$ defined by 
$\alpha \models {p \interferes^M q}$ when $s_0\cdot \alpha \models S_p \subseteq S_q$. 
 The following states that a capability system is secure with respect to the 
the associated security policy. 

\begin{theorem} 
\label{thm:capsystem-associated} 
Let $M$ be a capability system. Then the associated policy  $\interferes^M$ is local, and  
the system $M$ is both $\maytaname$-secure and $\musttaname$-secure
with respect to $\interferes^M$. 
\end{theorem} 

\begin{proof*} 
We show that $M$ admits an access control interpretation consistent with $\interferes^M$ 
and apply Theorem~\ref{thm:ac-secure}. 

For the set of objects of the access control interpretation, we take
$\Objects = \cup_{p \in \processes} \Obj_p \cup \{ \oset{p}\} $. 
(Since both the access control model and the capability model use a notion of 
\emph{object}, and these are not quite the same when relating the two models,
the reader needs to bear in mind that $\Objects$ refers to the objects in the acccess control 
model and $\Obj$ refers to objects in the capability model.) 
We define the access control functions for states $s$ and processes $p$ by 
$\observe p s  = \Obj_p \cup \{ \oset{p}\}$ 
and $\alter p s= \Obj_p \cup \{ in_q, O_q~|~s\models S_p \subseteq S_q\}$. 
For each state $s$ and object $o \in \cup_{p \in \processes} \Obj_p$, 
we define $\contents o s  = s(o)$, and (as is necessary) for $o = \oset{p}$, 
we define $\contents o s  = \observe s p$. 
We need to show that conditions \DRM-1 to \DRM-4, \DRM-5' and \DRM-6 are satisfied 
by these definitions. 

Note that, by definition of $\dynacrel{p}$, and the definitions above, we have 
$s \dynacrel{p} t $ iff  for all $o\in \Obj_p$ we have $s(o) = t(o)$
(the contents of $\oset{p}$ are static). 

Additionally, $o \in \alter p s \cap \observe q s$ iff 
$p = q$ and $o \in \Obj_p$ or 
$p \neq q$ and $s\models S_p \subseteq S_q$ and $o \in \{in_q, O_q\}$.  
    
\begin{enumerate}
\item[\DRM-\ref{cond:wa:outputconsistency}:]
  This is immediate from the above characterization of $\dynacrel{p}$ and the constraint 
  {\tt Obs}. 

\item [\DRM-\ref{cond:wa:contentconsistency}:]
Suppose $o \in \alter{\dom(a)} s \cap \alter{\dom(a)} t$ and  $s \dynacrel{\dom(a)} t$
and  $\contents o s = \contents o t$. We need to show that $\contents o {s\cdot a} = \contents o {t \cdot a}$. 
By definition, we have $o \in \alter{\dom(a)} s \cap \alter{\dom(a)} t$ 
just when either $o \in \Obj_{\dom(a)}$ or $o \in \{in_q, O_q\} $ for some process $q$ such that 
$s\models S_{\dom(a)} \subseteq S_q$. 
\From $s \dynacrel{\dom(a)} t$ we have that $s\restrict \dom(a) = t \restrict \dom(a)$, 
and from $\contents o s = \contents o t$ we have 
$s(o) = t(o)$. 

We argue that  $\contents o {s\cdot a} = \contents o {t \cdot a}$, 
i.e., $(s\cdot a)(o) = (t\cdot a)(o)$, for each possibility for the action $a$. 
\begin{itemize} 
\item If $a$ is a data action $p. b$ (so that $p = \dom(a)$), then
if $o \notin Data_{\dom(a)}$, we have by {\tt Data1} that $(s\cdot a)(o) = s(o) = t(o) = (t\cdot a)(o)$.
In the case $o \notin Data_{\dom(a)}$, we have $o \in \Obj_{\dom(a)}$, and
$(s\cdot a)(o)  = (t\cdot a)(o)$ follows by {\tt Data2}. 

\item The code for each of the actions $\addcap p c n$, $\dropcap{p}{c}$, 
$\addtag p t$, $\removetag p t$, reads only objects in $\Obj_p$ and 
writes only objects in $\Obj_p$. If $a$ is one of these actions, then
$p = dom(a)$ and  the claim holds. 

\item 
The only object that can be changed by the code for the action
$a = \sendmessageto p q$ is $in_q$, so this is the only case of $o$ that 
we need to consider. If $\dom(a) = p = q$, then the condition $S_p \models S_q$ is a 
tautology, so the code performs the assignment $in_p := in_p \cdot m_p$ 
in both states $s$ and $t$. If $\dom(a) = p \neq q$, then from 
$o = in_q \in \alter{\dom(a)} s \cap \alter{\dom(a)} t$ we have both 
$s\models S_p \subseteq S_q$ and $t\models S_p \subseteq S_q$. 
Thus, again, the code takes the same branch in both cases, 
and since $m_p \in \Obj_p$ we have $s(m_p) = t(m_p)$ 
and we have $(s\cdot a)(in_q) =s(in_q) \cdot  s(m_p) = t(in_q) \cdot  t(m_p) = (t\cdot a)(in_q)$, 
as required. 

The argument for the action $\sendcapto p c q$ is similar. 
\end{itemize} 

\item[\DRM-\ref{cond:wa:allowedinterference}:] 
Suppose that $o \notin \alter{\dom(a)} s$. Since $\Obj_{\dom(a)} \subseteq \alter{\dom(a)} s$, 
we have $o \notin \Obj_{\dom(a)}$.  It follows that data actions $a$ and actions $a$ of the forms
$\addcap p c n$, $\dropcap{p}{c}$, 
$\addtag p t$, $\removetag p t$ do not change the value of $o$, since
these actions change only the values of objects in $\Obj_{\dom(a)}$.  
In the case of $a = \sendmessageto p q$ and $\sendcapto p c q$, 
we have from $o \notin \alter{\dom(a)} s$ that  $o$ is not one of the
two possible objects $in_q$ and $O_q$ whose values these actions can change, 
so again, these actions do not change the value of $o$. 

\item[\DRM-\ref{cond:wa:grantobserve}:]  
This is immediate from the fact that the value of $\observe s p$ is a constant. 
 
\item[\strongpolicyconsistency:] 
Suppose $s \dynacrel{p} t$ and $s \dynacrel{q} t$. 
We need to show that $\alter p s \cap \observe q s = \alter p t \cap \observe q t$.  
Now $\observe q s$ is a constant, and if $p = q$ then $\alter p s = \Obj_p$ is also a constant, 
so we need only consider the case $p \neq q$. 
In this case, $\alter p s \cap \observe q s$ is $ \{ in_q, O_q\}$ if $s\models S_p \subseteq S_q$
and is $\emptyset$ otherwise. From $s \dynacrel{p} t$ and $s \dynacrel{q} t$ we get that 
$s(S_p) = t(S_p)$ and $s(S_q) = t(S_q)$, so $s\models S_p \subseteq S_q$ iff $t\models S_p \subseteq S_q$. 
Thus, $\alter p s \cap \observe q s = \alter p t \cap \observe q t$.  

\item[\DRM-\ref{cond:wa:commonvar}:] 
Note that for $p \neq q$, the sets $\Obj_p$ and $\Obj_q$ are disjoint. 
It follows that if $\alter p {s_0\cdot \alpha} \cap \observe q {s_0\cdot \alpha} \neq \emptyset$, then  
we have either $p= q$ or $s_0\cdot \alpha \models S_p \subseteq S_q$. In either case, 
we have $s_0\cdot \alpha \models S_p \subseteq S_q$, i.e., $\alpha \models {p \interferes^M q}$. 
\qedhere
\end{enumerate}      
\end{proof*}

It follows from Theorem~\ref{thm:capsystem-associated}
that a capability system is secure with respect to any policy that is no more
restrictive than the associated policy. 

\begin{corollary} 
\label{corol:capability-monotonicity} 
Let $M$ be a capability system. Then for any dynamic security policy $\interferes$ with 
${\interferes^M} \leq {\interferes}$, the system $M$ is both $\maytaname$-secure and $\musttaname$-secure
with respect to $\interferes$. 
\end{corollary}

\begin{proof} 
Immediate from Theorem~\ref{thm:capsystem-associated} using Proposition~\ref{prop:mustta-monotonic} and Proposition~\ref{prop:may-monotonic}. 
\end{proof}

\subsection{Comments on Flume} 
\label{sec:flume-comments}

The capability system described above is somewhat motivated by Flume \cite{Krohn2009}, 
but differs in a number of respects.
 We comment here on the differences.  

First, Flume associates both a {\em secrecy set} $S_p\subseteq \tags$ and an 
{\em integrity set} $I_p \subseteq \tags$  to each agent $p$. These operate dually, 
so we have avoided inclusion of the integrity set just for simplicity. Our formulation 
could be extended to include it.

Next, Flume permits transmission of information from $p$ to $q$ when 
$S_p \setminus D_p \subseteq S_q \cup D_q$, where $D_u$ is the set of tags 
for which $u$ possesses  both the capability to add and remove the tag. 
The intuition for this definition is that if $p$ wishes to transmit information to $q$, 
but is constrained from doing so because $S_p$ contains a tag $t$ 
that are not in $S_q$, then this obstacle would be overcome
if $p$ is able to remove the tag from $S_p$ or if $q$ is able to add the tag to $S_q$. 
Assuming no other obstacles, the transmission could then proceed. Moreover, if, 
correspondingly, $p$ is also able to add the tag or
$q$ is able to remove it, then after the transmission, the policy could be restored
to its original state. Thus, the condition $S_p \setminus D_p \subseteq S_q \cup D_q$
captures the situations where this sequence of moves would allow transmission from 
$p$ to $q$ while leaving the policy state unchanged after the transmission is complete. 
The thinking then is that since such a sequence of moves cannot be prevented, the policy may as well 
allow the transmission directly. 

We are not convinced that this is a good idea. One reason for $p$ to have a tag in $D_p$ 
but not in $S_p$ might be that $p$ wishes to protect itself against {\em inadvertent} 
transmission of information to $q$. Mechanisms such as {\em sudo} 
and the discretionary access controls on a user's own objects in 
Unix are explicitly designed to disable the effects of actions that the 
user nevertheless has the capacity to enable, and it seems useful
to have a similar discretionary layer of control for information flow policies. 
So we have not followed Flume in adopting this modification of the basic policy.  

The Flume model includes a set $\hat{O}$ of {\em global capabilities}, that 
are implicitly owned by all domains. That is, Flume works with the sets of owned
capabilities  $\overline{O}_p = O_p \cup \hat{O}$. When creating a tag, an agent 
is able to place the associated capabilities into $\hat{O}$. The motivation 
for this appears to be the following (\cite{Krohn2009} section II.D): 

\begin{quote} 
The most important policy in Flume is {\em export protection}, wherein
untrustworthy processes can compute with secret data without the ability
to reveal it. An export protection tag is a tag $t$ 
such that $t^+ \in \hat{O}$ and $t^-\in \hat{O}$. For a process $p$ to 
achieve such a result, it creates a new tag $t$ and grants $t^+$ to the 
global set $\hat{O}$, while closely guarding $t^-$.  To protect a file $f$, $p$ creates the 
file with secrecy label $\{t\}$. If a process $q$ wishes to read $f$, it must first 
add $t$ to $S_q$, which it can do since $t^+ \in \hat{O} \subseteq \overline{O}_q$. 
Now, $q$ can only send messages to other processes with $t$ in their labels. 
It requires $p$'s authorization to remove $t$ from its label or send data to the 
network. 
\end{quote}  

We have not followed Flume in including the set of global capabilities $O$. 
One reason is that we find the above argument unconvincing. Once 
$t^+$ is in $\hat{O}$, not just $q$, but any other agent $r$ may add $t$ to its secrecy set, 
and any protection afforded by tag $t$ is then eliminated. 
The main reason that Krohn and Tromer have in mind that $r$ might not be able
to add $t$ to its secrecy set is that $t$ is randomly generated, and 
$r$ does not know the appropriate value to use, making it a low probability 
event that $r$ should add $t$ unless it somehow receives $t^+$ in a chain 
of transmission from $p$. But in that case, we may as well rely upon the 
direct transmission of $t^+$ to the set $O_r$ as the enabler for 
$r$'s adding $t$ to its secrecy set. There does not appear to be 
any necessary role for $\hat{O}$ in the above argument. 

Another reason for not following Flume in the inclusion of the global set $\hat{O}$ 
is that it  causes a covert channel. 
An agent $p$ may choose to add a capability $t$ to $O$, 
and any other agent $q$ can test whether the 
capability $t$ is in $O$ or not. This causes a flow of
information from $p$ to $q$.  Krohn and Tromer are aware of this covert channel 
and seek to mitigate it by randomizing tag names. 
However this means that they require a much more complicated 
probabilistic version of noninterference semantics. 
It is not clear to us whether their semantics is well-formulated, 
since some subtle points appear to have been overlooked. 
Note that their systems have both probabilistic transitions
(random generation of tag names) and nondeterministic 
transitions (agent's choice of their actions). 
This means that there is not a straightforward probability 
distribution on runs. How to handle the mix of probability and nondeterminism 
is an issue that has been studied in the literature, 
and involves a number of subtleties \cite{JSM97,Cheung}. It does not
appear that Krohn and Tromer are aware of these issues, indeed, they do not even 
define the probability space that they use. It would seem 
that, at the very least, some non-trivial further details are needed to make 
their result fully precise, which we will not pursue here.

Finally, Flume allows for dynamic creation of processes by means of a forking operation. 
It uses this to represent newly created objects as processes. 
Our policy model assumes a fixed set of domains, so this aspect of 
Flume cannot be directly represented. It would be interesting 
to develop an extension of our framework  with the expressiveness
to encompass creation of objects and agents.

\section{Related Work} 
\label{sec:relatedwork}

Related work using an automaton-based semantics similar to the one we
have used in the paper has already been discussed in detail in Section~\ref{sec:related:aut}. 
We focus here on work in the context of programming language security
that has considered dynamic information flow policies.  
A recent survey of dynamic 
policies in the context of programming languages is \cite{BrobergDS15}.  
They attempt to develop a classification based on various ``facets" that 
are factored into the wide spectrum of definitions of security in this area.
In general, there exist nontrivial gaps between the semantic frameworks used in 
this literature and our semantic model. Even \emph{within} the programming language literature, 
formal relationships between the various semantics studied are rarely presented.  
We therefore give only some informal comparisons 
with our work.

One of the facets discussed in \cite{BrobergDS15} is the distinction between 
``whitelisting flows" and ``blacklisting flows", which is similar to our
distinction between a permissive and a prohibitive reading of policies.  
Interestingly, the whitelisting/permissive reading is claimed to be the 
norm in the programming language security literature.

However, there are also significant points of difference.
Generally these works concern a programming language framework in which 
the secrets to be concealed are already encoded into the initial state, rather than 
our ``interactive" model in which secrets are  generated
on the fly as the result of nondeterministic choices of action made by the agents. 
Frequently, the programming languages studied are deterministic, 
so a direct relationship to our setting is not immediately clear. 
There exists an approach by Clarke and Hunt~\cite{ClarkH08} to 
handling interactivity in a programming language setting  
by including stream variables in the initial state to represent the
sequence of future inputs that will be selected by an agent over the 
course of the run. However, this, in effect, assumes that the scheduler 
also is deterministic. Adding a variable for the schedule and 
allowing other agents to learn the complete schedule would 
be more permissive than a definition such as \tasecurity, 
since this restricts the information that agents may learn about the schedule. 
It therefore seems that establishing an exact correspondence would 
require some detailed work. 

Another significant point of difference is that in 
the programming language literature, the policy and the system model (program) 
are often tightly intertwined, whereas we have a separate policy and system model, 
just as a specification and implementation  are usually considered distinct. 
Other works that do allow for such a  separation include 
Morgan \cite{Morgan09} develops a refinement calculus that allows for mixtures
of information flow specifications and code, enabling  an implementation to 
be derived from a specification via steps that mix the two, and  
Delft et. al.~\cite{DelftHS15}, who develop a type theory for dynamic information 
flow that operates on a program separately from the dynamic policy. 

One of the ways that  intransitivity is handled in the programming 
language setting also involves a dynamic aspect to policies: 
special constructs are introduced that allow temporary violations of a static policy, 
either through use of special basic actions or the establishment of a 
program region within which an alternate policy is enforced. 
For example, Mantel and Sands~\cite{Mantel:Sands:APLAS04} 
work with a classical transitive security policy, but add downgrading commands 
that may violate this policy (subject to some alternate constraints.) 
They give a security definition that uses a specially defined bisimulation 
relation.

Broberg and Sands \cite{Broberg2006,Broberg2009} 
define \emph{flow locks}, which enable variables 
to be typed with conditional annotations that constrain the 
security levels to which contents of these variables may flow. 
The values referred to in the conditions may be altered during the execution 
of the program, thereby making the policy dynamic. 
 The authors give a complex bisimulation-based semantics 
 to these policies  \cite{Broberg2006}, but in later work 
 \cite{Broberg2009} prefer an alternate knowledge-based semantics, with a 
 perfect recall interpretation of knowledge. It appears that their definitions are closer to those for 
a classical purge-based definition than to our intransitive policies. 
To capture a downgrader policy $H \interferes D \interferes L$, 
such as to allow information from $H$ to reach $L$ via $D$, it seems necessary to 
temporarily change the policy to effectively add an edge $H\interferes L$ to the policy, 
allowing a direct flow from $H$ to~$L$. 
This is not in the spirit of intransitive noninterference, which intends to capture 
that all flows from $H$ to $L$ must pass through~$D$. 

Chong and Myers \cite{ChongM04} propose to handle declassification violations of 
transitive policy using a notion  of \emph{noninterference until declassification}  
which applies a classical noninterference definition on runs 
just up to the point where a sequence of events has occurred
on the basis of which the policy permits a declassification of secret information 
to the domain under consideration.
Beyond this point,  the definition does not constrain the behaviour of the system. 
Hicks et al \cite{HTHZ05} refine this idea to impose constraints on longer runs, 
declaring these secure if they consist of a sequence of 
policy-update free fragments, each of which satisfies the standard noninterference condition. 
This extension is also applied in Swamy et al \cite{SwamyHTZ06}, who develop a language with 
dynamic role-based information flow policies. This approach  
is also similar to the approach used by Almeida-Matos and Boudol \cite{MatosB09}, who  
work in a language-based setting where a transitive policy can temporarily be extended
by addition of new edges (yielding another transitive policy).  
Their security condition states, using a bisimulation,  
that each transition in an execution is secure with respect to the transitive 
policy in place at the time the transition is taken. 
We note that these approaches generally do not handle covert channels resulting from the policy
changes themselves, as illustrated in our Example~\ref{ex:conflict-example}.

All of the above approaches, like ours in this paper, work with perfect recall definitions  of attacker knowledge. 
Askarov and Chong \cite{Askarov2012} point to an interesting issue with respect to dynamic security policies, 
which is that events may transmit ``old news": information that would have been allowed to be transmitted 
to an attacker according to an earlier state of the security policy, but was actually not transmitted.   
The current state of the policy may prohibit the flow of this information. In this case, 
security definitions that are based on a perfect recall model of attacker knowledge, 
fail to identify such transmissions of ``old news" as violations of security. 
  
In response to this observation, Askarov and Chong propose to work with attackers that are weaker than the perfect
recall attacker. They define a notion of knowledge for such attackers, and state definitions 
of security that require that any new knowledge gained by an attacker at security level (domain) $u$  should 
only concern security levels $v$ such that $v\interferes u$ according to the security policy 
at the time of the event. They identify some potential weaker attackers, but do not identify a definitive set of attackers; 
since the attacker is a parameter of their proposed security definition one could view the set of attackers
against which a system should be secure as a matter of policy. 

One of the attackers considered by Askarov  and Chong is the perfect recall attacker, and here it is 
reasonable to attempt to compare their framework to ours. Unfortunately, making such a comparison formally is 
not completely straightforward because their semantic setting is somewhat different from ours: 
whereas in our framework there are multiple causal agents (the domains), who are free to 
choose any action at any time, in theirs there is a single causal agent (a deterministic program) that reads
and writes information from different security domains. However, there is one respect where 
a clear difference of intent is apparent in their intuitions and definitions. 

Consider a situation with domains $D = \{ A,B,C, P\}$. 
Let $\Delta$ be the policy $\{u \interferes u ~|~ u \in D \} \cup \{ P \interferes u~|~ u \in D\}$. 
Intuitively, $P$ here is the policy authority, and $\Delta$ states that all domains are permitted to know the state of the policy. 
In the initial state of the system,  take the policy to be $\Delta \cup \{A\interferes B\}$. 
Any actions performed by $A$ are, intuitively, permitted by this policy to be recorded in the state of $B$, 
but should remain unknown to $C$. Suppose that after some actions by $A$, the policy agent changes the policy to 
$\Delta \cup \{B\interferes C\}$. Now, $B$ is permitted to communicate with $C$. According to our intuitions and formal definitions, 
$B$ may now send  to $C$ the information it has about the actions that $A$ performed before the policy change. 
Such flows from $A$ to $C$ via $B$ are typical of the flows that intransitive policies are intended to permit \cite{HY87}. 

However, according to Askarov and Chong's definitions (both for the perfect recall attacker and weaker attackers), if 
an action now copies this information from $B$'s state and it is observed by $C$, this is a violation of 
security, because in a transition where it makes the observation, $C$ learns new information about $A$, whereas the current state of the 
policy prohibits it from doing so. Thus, although Askarov and Chong permit intransitive 
policies, we would argue that their definitions do not handle dynamic intransitivity in a way that 
fits our intuitions: their semantics can best be characterized as a dynamic version of 
the classical purge-based definition of security, rather than a dynamic version of a semantics
for intransitive noninterference. 
The details of Askarov and Chong's semantics are refined
in \cite{DelftHS15}, but similar remarks apply to to this work.

\section{Conclusion} 

We have developed two generalizations for intransitive dynamic policies 
of the notion of TA-security that has recently been argued to provide a 
better semantics for static policies. One is based on a permissive reading
of policy edges, the other on a prohibitive reading of policy edges. Under
the natural condition of locality, the two readings are equivalent. 
We have also provided proof theory for these notions in the form of unwinding 
conditions and a sound and complete policy enforcement method using  
access control. Finally, we have demonstrated the applicability of the theory
by applying it to establish security for a natural capability model.

Several interesting directions remain open for further investigation. One
concerns the circularity in the justification of $\musttaname$-security. 
While $\musttaname$-security is, as we have shown, self-consistent, 
there may be further natural constraints that also lead to self-consistent
notions of security. For example, one candidate is to require for information transfer that the sender know that there is an edge from it 
to the receiver, in place of the assumption we made in $\musttaname$-security that
the sender and receiver must jointly know that there is such an edge. 
Whether such a constraint should apply is arguably a matter of policy. 
This suggests that there may be a case for a policy language that
more explicitly provides a way to express conditions on  agent knowledge. 
It would also be of interest to develop algorithms for verifying whether a system satisfies a policy. 
Both models and dynamic policies as we have defined them are highly expressive, 
so this would require some restrictions: finite state models and policies that 
are representable using finite state automata would be one interesting starting
point: it remains to be shown that this case is decidable. 

We have shown that  while $\maytaname$-security is more expressive than $\musttaname$-security, 
the difference between these notions applies only in the case of non-local policies. 
Both are associated with natural  intuitions, but we leave open the question of whether the
extra expressiveness of $\maytaname$-security is necessary in practice, i.e., whether there are any 
interesting systems requiring non-local policies where $\maytaname$-security, rather than $\musttaname$-security, is the appropriate notion 
of security. The main realistic example we have presented, the capability system of Section~\ref{sec:flume} involves a local policy, 
so it does not resolve this question. The notion of $\maytaname$-security allows information flows relating to policy 
settings that are not explicitly represented as edges in the policy itself. Conceivably, a richer policy format may allow these implicit flows 
to be more explicitly represented. 

It  would also be  interesting  to formally relate our work 
to work in the programming language setting, and to 
combine the insights from 
Askarov and Chong's work on imperfect recall attackers with the approach we have developed in the 
present paper.

\bibliographystyle{abbrv}
\bibliography{bibliography}

\begin{thebibliography}{10}

\bibitem{Askarov2012}
A.~Askarov and S.~Chong.
\newblock Learning is change in knowledge: Knowledge-based security for dynamic
  policies.
\newblock In S.~Chong, editor, {\em CSF}, pages 308--322. IEEE, 2012.

\bibitem{BDRF08}
C.~Boettcher, R.~DeLong, J.~Rushby, and W.~Sifre.
\newblock The {MILS} component integration approach to secure information
  sharing.
\newblock In {\em Proc. 27th IEEE/AIAA Digital Avionics Systems Conference},
  pages 1.C.2--1--1.C.2--14, Oct. 2008.

\bibitem{Broberg2006}
N.~Broberg and D.~Sands.
\newblock Flow locks: Towards a core calculus for dynamic flow policies.
\newblock In P.~Sestoft, editor, {\em ESOP}, volume 3924 of {\em Lecture Notes
  in Computer Science}, pages 180--196. Springer, 2006.

\bibitem{Broberg2009}
N.~Broberg and D.~Sands.
\newblock Flow-sensitive semantics for dynamic information flow policies.
\newblock In {\em PLAS}, pages 101--112, 2009.

\bibitem{BrobergDS15}
N.~Broberg, B.~van Delft, and D.~Sands.
\newblock The anatomy and facets of dynamic policies.
\newblock In {\em {IEEE} 28th Computer Security Foundations Symposium, {CSF}
  2015, Verona, Italy, 13-17 July, 2015}, pages 122--136, 2015.

\bibitem{BytschkowQIR14}
D.~Bytschkow, J.~Quilbeuf, G.~Igna, and H.~Ruess.
\newblock Distributed {MILS} architectural approach for secure smart grids.
\newblock In {\em Smart Grid Security - Second International Workshop,
  SmartGridSec 2014, Munich, Germany, February 26, 2014, Revised Selected
  Papers}, pages 16--29, 2014.

\bibitem{Cheung}
L.~Cheung.
\newblock {\em Reconciling Nondetermistic and Probabilistic Choices}.
\newblock PhD thesis, IPA, Radboud University, Nijmegen, 2006.

\bibitem{ChongM04}
S.~Chong and A.~C. Myers.
\newblock Security policies for downgrading.
\newblock In {\em Proceedings of the 11th {ACM} Conference on Computer and
  Communications Security, {CCS} 2004, Washington, DC, USA, October 25-29,
  2004}, pages 198--209, 2004.

\bibitem{ClarkH08}
D.~Clark and S.~Hunt.
\newblock Non-interference for deterministic interactive programs.
\newblock In {\em Formal Aspects in Security and Trust, 5th International
  Workshop, {FAST} 2008, Malaga, Spain, October 9-10, 2008, Revised Selected
  Papers}, pages 50--66, 2008.

\bibitem{Cohen77}
E.~S. Cohen.
\newblock Information transmission in computational systems.
\newblock In {\em {SOSP}}, pages 133--139, 1977.

\bibitem{Eggert2013}
S.~Eggert, H.~Schnoor, and T.~Wilke.
\newblock Noninterference with local policies.
\newblock In K.~Chatterjee and J.~Sgall, editors, {\em MFCS}, volume 8087 of
  {\em Lecture Notes in Computer Science}, pages 337--348. Springer, 2013.

\bibitem{EggertMSW13}
S.~Eggert, R.~van~der Meyden, H.~Schnoor, and T.~Wilke.
\newblock Complexity and unwinding for intransitive noninterference.
\newblock {\em CoRR}, abs/1308.1204, 2013.

\bibitem{Engelhardt2012}
K.~Engelhardt, R.~van~der Meyden, and C.~Zhang.
\newblock Intransitive noninterference in nondeterministic systems.
\newblock In {\em ACM Conference on Computer and Communications Security},
  pages 869--880, 2012.

\bibitem{fhmvbook}
R.~Fagin, J.~Y. Halpern, Y.~Moses, and M.~Y. Vardi.
\newblock {\em Reasoning About Knowledge}.
\newblock MIT-Press, 1995.

\bibitem{Goguen1982}
J.~A. Goguen and J.~Meseguer.
\newblock Security policies and security models.
\newblock In {\em IEEE Symposium on Security and Privacy}, pages 11--20, 1982.

\bibitem{goguen84}
J.~A. Goguen and J.~Meseguer.
\newblock Unwinding and inference control.
\newblock In {\em IEEE Symposium on Security and Privacy}, 1984.

\bibitem{HY87}
J.~T. Haigh and W.~D. Young.
\newblock Extending the noninterference version of {MLS for SAT}.
\newblock {\em IEEE Transactions on Software Engineering}, SE-13(2):141--150,
  Feb 1987.

\bibitem{JSM97}
J.~He, K.~Seidel, and A.~McIver.
\newblock Probabilistic models for the guarded command language.
\newblock {\em Sci. Comput. Program.}, 28(2-3):171--192, 1997.

\bibitem{HeitmeyerALM06}
C.~L. Heitmeyer, M.~Archer, E.~I. Leonard, and J.~D. McLean.
\newblock Formal specification and verification of data separation in a
  separation kernel for an embedded system.
\newblock In {\em Proc. of the 13th {ACM} Conf. on Computer and Communications
  Security, {CCS}}, pages 346--355, 2006.

\bibitem{HTHZ05}
M.~Hicks, S.~Tse, B.~Hicks, and S.~Zdancewic.
\newblock Dynamic updating of information-flow policies.
\newblock In {\em Proc. of Foundations of Computer Security Workshop (FCS)},
  2005.

\bibitem{Krohn2009}
M.~Krohn and E.~Tromer.
\newblock Noninterference for a practical {DIFC}-based operating system.
\newblock In {\em IEEE Symposium on Security and Privacy}, pages 61--76, 2009.

\bibitem{Leslie2006}
R.~Leslie.
\newblock Dynamic intransitive noninterference.
\newblock In {\em Proc. IEEE International Symposium on Secure Software
  Engineering}, 2006.

\bibitem{Mantel:Sands:APLAS04}
H.~Mantel and D.~Sands.
\newblock Controlled declassification based on intransitive noninterference.
\newblock In {\em Proc. Asian Symp. on Programming Languages and Systems},
  volume 3302 of {\em LNCS}, pages 129--145. Springer-Verlag, Nov. 2004.

\bibitem{Martin2000}
W.~Martin, P.~White, F.~Taylor, and A.~Goldberg.
\newblock Formal construction of the mathematically analyzed separation kernel.
\newblock In I.~C.~S. Press, editor, {\em Proc. 15th IEEE Conf. on Automated
  Software Engineering}, 2000.

\bibitem{MatosB09}
A.~A. Matos and G.~Boudol.
\newblock On declassification and the non-disclosure policy.
\newblock {\em Journal of Computer Security}, 17(5):549--597, 2009.

\bibitem{Morgan09}
C.~Morgan.
\newblock The shadow knows: Refinement and security in sequential programs.
\newblock {\em Sci. Comput. Program.}, 74(8):629--653, 2009.

\bibitem{Murray2012}
T.~Murray, D.~Matichuk, M.~Brassil, P.~Gammie, and G.~Klein.
\newblock Noninterference for operating system kernels.
\newblock In C.~Hawblitzel and D.~Miller, editors, {\em Certified Programs and
  Proofs}, volume 7679 of {\em Lecture Notes in Computer Science}, pages
  126--142. Springer Berlin Heidelberg, 2012.

\bibitem{Myers1997}
A.~C. Myers and B.~Liskov.
\newblock A decentralized model for information flow control.
\newblock In {\em Proceedings of the Sixteenth ACM Symposium on Operating
  Systems Principles}, SOSP '97, pages 129--142, New York, NY, USA, 1997. ACM.

\bibitem{Roscoe1999}
A.~W. Roscoe and M.~H. Goldsmith.
\newblock What is intransitive noninterference?
\newblock In {\em IEEE Computer Security Foundations Workshop}, pages 228--238,
  1999.

\bibitem{rushby92}
J.~Rushby.
\newblock Noninterference, transitivity, and channel-control security policies.
\newblock Technical Report {CSL-92-02}, SRI International, 1992.

\bibitem{Schellhorn2002}
G.~Schellhorn, W.~Reif, A.~Schairer, P.~A. Karger, V.~Austel, and D.~C. Toll.
\newblock Verified formal security models for multiapplicative smart cards.
\newblock {\em Journal of Computer Security}, 10(4):339--368, 2002.

\bibitem{SwamyHTZ06}
N.~Swamy, M.~Hicks, S.~Tse, and S.~Zdancewic.
\newblock Managing policy updates in security-typed languages.
\newblock In {\em 19th {IEEE} Computer Security Foundations Workshop,
  {(CSFW-19} 2006), 5-7 July 2006, Venice, Italy}, pages 202--216, 2006.

\bibitem{DelftHS15}
B.~van Delft, S.~Hunt, and D.~Sands.
\newblock Very static enforcement of dynamic policies.
\newblock In {\em Proc. Principles of Security and Trust - 4th Int. Conf.,
  {POST}}, pages 32--52, 2015.

\bibitem{Meyden15}
R.~van~der Meyden.
\newblock What, indeed, is intransitive noninterference?1.
\newblock {\em Journal of Computer Security}, 23(2):197--228, 2015.
\newblock An earlier version of this work appeared in ESORICS'07.

\bibitem{Vandebogart2007}
S.~Vandebogart, P.~Efstathopoulos, E.~Kohler, M.~Krohn, C.~Frey, D.~Ziegler,
  F.~Kaashoek, R.~Morris, and D.~Mazi\`{e}res.
\newblock Labels and event processes in the {Asbestos} operating system.
\newblock {\em ACM Trans. Comput. Syst.}, 25(4), Dec. 2007.

\bibitem{Vanfleet2005}
W.~Vanfleet, R.~Beckworth, B.~Calloni, J.~Luke, C.~Taylor, and G.~Uchenick.
\newblock {MILS}:architecture for high assurance embedded computing.
\newblock {\em Crosstalk: The Journal of Defence Engineering}, 2005.

\bibitem{ZeldovichBKM06}
N.~Zeldovich, S.~Boyd{-}Wickizer, E.~Kohler, and D.~Mazi{\`{e}}res.
\newblock Making information flow explicit in {HiStar}.
\newblock In {\em 7th Symposium on Operating Systems Design and Implementation
  {(OSDI} '06), November 6-8, Seattle, WA, {USA}}, pages 263--278, 2006.

\end{thebibliography}

\end{document}